\documentclass[12pt]{article}
\usepackage[utf8]{inputenc}

    \usepackage[dvipsnames]{xcolor}
    \usepackage[T1]{fontenc}
	\usepackage[utf8]{inputenc}
	\usepackage[english]{babel}
	\usepackage{amsfonts}
	\usepackage{amssymb}
	    \usepackage{amsmath}
	
	\usepackage{enumerate}
	\usepackage{bbm}
	\usepackage{mathtools}
	\usepackage{booktabs}	
	\usepackage{setspace}
	\usepackage{graphicx}
	\usepackage{fullpage}
	
	\usepackage{soul}
	\usepackage{xspace}
	\usepackage{csquotes}
	\usepackage[margin = 1in]{geometry}
	\usepackage{placeins}
	\usepackage[footfont=normalsize,font=normalsize]{floatrow}
	\usepackage{lscape}
	\usepackage{longtable}
	\usepackage{mathabx}
    \usepackage{accents}

    \usepackage{tikz}
    \usepackage{threeparttable}
        \usepackage{subcaption}          
    \usetikzlibrary{calc,intersections}

	\usepackage[longnamesfirst]{natbib}
	\usepackage{bibunits}
	\defaultbibliography{Bibliography.bib}  
	\defaultbibliographystyle{aer}

\usepackage{silence}
	\usepackage{hyperref}
	\hypersetup{
    colorlinks=true,
    linkcolor=Red,
    filecolor=magenta,      
    urlcolor=blue,
    citecolor = Blue,
    pdfencoding = auto
}

 \usepackage{cleveref}

    \usepackage{tcolorbox}
    \newtcbox{\feedback}{nobeforeafter,colframe=black,colback=white,boxrule=0.5pt,arc=2pt,
      boxsep=0pt,left=2pt,right=2pt,top=2pt,bottom=2pt,tcbox raise base}

    \usepackage{amsthm}
    \usepackage{thmtools}

    \declaretheoremstyle[
  qed=$\blacktriangle$, %
  headfont=\normalfont\bfseries,
  notefont=\mdseries, notebraces={(}{)},
  bodyfont=\normalfont,
  postheadspace=\newline
]{qedstyle}
    \declaretheorem[style=qedstyle, name=Remark]{rem}
    \declaretheorem[style=qedstyle, name=Example]{example}
    
    \newtheorem{asm}{Assumption}
    
    \newtheorem{prop}{Proposition}[section]
    \newtheorem{lem}{Lemma}[section]
    \newtheorem{cor}{Corollary}[section]
    \theoremstyle{definition}

\crefname{thm}{theorem}{theorems}
\Crefname{thm}{Theorem}{Theorems}

\crefname{asm}{assumption}{assumptions}
\Crefname{asm}{Assumption}{Assumptions}

\crefname{con}{conjecture}{conjectures}
\Crefname{con}{Conjecture}{Conjectures}

\crefname{prop}{proposition}{propositions}
\Crefname{prop}{Proposition}{Propositions}

\crefname{lem}{lemma}{lemmas}
\Crefname{lem}{Lemma}{Lemmas}

\crefname{cor}{corollary}{corollaries}
\Crefname{cor}{Corollary}{Corollaries}

\crefname{rem}{remark}{remarks}
\Crefname{rem}{Remark}{Remarks}

\crefname{defn}{definition}{definitions}
\Crefname{defn}{Definition}{Definitions}

\crefname{example}{example}{examples}
\Crefname{example}{Example}{Examples}

\crefname{appendixfigure}{Appendix Figure}{Appendix Figures}
\crefname{appendixtable}{Appendix Table}{Appendix Tables}

\usepackage{array}
\newcolumntype{L}[1]{>{\raggedright\let\newline\\\arraybackslash}m{#1}}
\newcolumntype{C}[1]{>{\centering\let\newline\\\arraybackslash\hspace{0pt}}m{#1}}
\newcolumntype{R}[1]{>{\raggedleft\let\newline\\\arraybackslash\hspace{0pt}}m{#1}}

\usepackage{ragged2e}
\newlength\ubwidth

\newcommand\numberthis{\addtocounter{equation}{1}\tag{\theequation}}

\usepackage{pgfplots}
\pgfplotsset{compat=1.18} 
\usetikzlibrary{patterns}
\usetikzlibrary{arrows.meta}

\usepackage{clipboard}

\makeatletter
\renewcommand\paragraph{\@startsection
  {paragraph}{4}{\z@}%
  {1.75ex \@plus .2ex \@minus .2ex}%
  {-1em}%
  {\normalfont\normalsize\bfseries}}
\makeatother

	\newcommand{\reals}{\mathbb{R}}

\newcommand{\indep}{\mathrel{\perp\mspace{-10mu}\perp}}

\newcommand\fnsep{\textsuperscript{,}}

\newclipboard{output-clipboard}

\title{Testing Mechanisms\thanks{We are grateful to Alberto Abadie, Don Andrews, Isaiah Andrews, Clément de Chaisemartin,
    Kevin Chen, Xavier D'Haultfoeuille, John Friedman, Martin Huber, Peter Hull, Toru Kitagawa, Pat Kline, Simon Lee, Caleb Miles, Ismael Mourifié, Ben Roth, Pedro Sant'Anna, Yuya Sasaki, Jesse
    Shapiro, Zhenting Sun, Chris Walters, Yuanyuan Wan, and participants at numerous seminars and conferences for helpful comments and
    suggestions. We thank Chen Cheng, Scott Lu, Jin Niu, and Eddie Wu for excellent research
    assistance. Roth gratefully acknowledges funding from the NIH under grant NIGMS 1R35GM155224 and from the Alfred P. Sloan Foundation.}}

\author{Soonwoo Kwon\thanks{Brown
    University. \href{mailto:soonwoo\_kwon@brown.edu}{soonwoo\_kwon@brown.edu}}
  \and Jonathan Roth\thanks{Brown
    University. \href{mailto:jonathanroth@brown.edu}{jonathanroth@brown.edu}}}

\usepackage[subtle]{savetrees}

\definecolor{cbSkyBlue}{HTML}{56B4E9}
\definecolor{cbOrange}{HTML}{E69F00}

\begin{document}
\maketitle

\begin{abstract}
Economists are often interested in the mechanisms by which a treatment affects an outcome. We develop tests for the ``sharp null of full mediation'' that a treatment $D$ affects an outcome $Y$ only through a particular mechanism (or set of mechanisms) $M$. Our approach exploits connections between mediation analysis and the econometric literature on testing instrument validity. We also provide tools for quantifying the magnitude of alternative mechanisms when the sharp null is rejected: we derive sharp lower bounds on the fraction of individuals whose outcome is affected by the treatment despite having the same value of $M$ under both treatments (``always-takers''), as well as sharp bounds on the average effect of the treatment for such always-takers. An advantage of our approach relative to existing tools for mediation analysis is that it does not require stringent assumptions about how $M$ is assigned. We illustrate our methodology in two empirical applications.

\end{abstract} \vspace{2cm}
\newpage

\section{Introduction}

Social scientists are often able to identify the causal effect of a treatment $D$ on some outcome of interest $Y$, either by explicitly randomizing $D$ or using some ``quasi-experimental'' variation in $D$. Once the causal effect of $D$ on $Y$ is established, a natural question is \emph{why} does it work, i.e. what are the \emph{mechanisms} by which $D$ affects $Y$?

To fix ideas, consider the setting of \citet{bursztyn_misperceived_2020}, which will be one of our empirical applications below. The authors show that the vast majority of men in Saudi Arabia underestimate how open other men are to women working outside of the home. They then run an experiment in which some men are randomized to receive information about other men's beliefs. At the end of the experiment, all of the men are given the choice between signing their wives up for a job-search service or taking a gift card. The authors observe that the information treatment increases the probability that men sign their wives up for the job-search service, and also increases the probability that their wives apply for and interview for jobs over the subsequent five months. A natural question in interpreting these results is then whether the increase in longer-run outcomes (e.g. job applications) is explained by the short-run sign-up for the job-search service, or whether the treatment also affects labor market outcomes through other longer-run changes in behaviors. 

The literature on mediation analysis (see \citet{huber_review_2019} for a review) provides formal methodology for disentangling how much of the average effect of a treatment $D$ (e.g. information about others' beliefs) on an outcome $Y$ (e.g. job applications) is explained by the indirect effect through some potential mediator $M$ (e.g. job-search service sign-up). A challenge, however, is that even if the treatment $D$ is randomly assigned, it will often be the case that the mediator of interest $M$ is not randomly assigned.\footnote{One exception is ``mechanism experiments'' \citep{ludwig_mechanism_2011}, where the researcher explicitly randomizes an $M$ of interest. Our focus is on the common setting where $M$ was not randomized (e.g. due to lack of foresight, budget, or feasibility of randomization) and thus potentially endogenous.} Existing approaches typically make strong assumptions that allow for the identification of the causal effect of $M$ on $Y$ (see Related Literature below). A common assumption in the biostatistics literature, for example, is that $M$ is as good as randomly assigned given $D$ and some observable characteristics. This assumption will often be restrictive in applications---for example, we may worry that sign-up for the job-search service is correlated with unobservables related to women's labor supply. 

In this paper, we develop methodology that sheds light on mechanisms without having to impose strong assumptions to identify the effect of $M$ on $Y$. We make progress by considering an easier question than what is typically studied in the literature on mediation analysis, but one that we think will still be informative in many applications. Rather than trying to identify how much of the average effect is explained by the indirect effect through $M$, we start by testing what we refer to as the \emph{sharp null of full mediation}: is the effect of $D$ on $Y$ fully explained through its effect on $M$? In our motivating application, the sharp null asks whether the effect of treatment on job applications is fully explained by the short-run take-up of the job-search service. More precisely, letting $Y(d,m)$ be the potential outcome as a function of treatment $d$ and mediator $m$, the sharp null posits that $Y(d,m)$ depends only on $m$. If we can reject this null in our motivating example, then we can conclude that the treatment affects long-run outcomes through some change in behavior other than job-search service sign-up. In addition to testing the sharp null, our approach also provides useful information about the extent to which the null is violated. In particular, we develop lower bounds on the fraction of individuals whose outcome is affected by the treatment despite having the same value of $M$ under both treatments. In our motivating example, this means we can lower bound the fraction of women whose labor market outcomes are affected by the information treatment despite the treatment having no effect on whether they sign up for the job service. 

Our main theoretical results impose two assumptions. First, we suppose that the treatment $D$ is as good as randomly assigned, i.e. $D$ is independent of the potential outcomes $Y(\cdot,\cdot)$ and potential mediators $M(\cdot)$. In our motivating example, this is guaranteed by design since $D$ is randomly assigned. (We consider extensions to ``quasi-experimental'' settings in \Cref{sec: nonexperimental}.) Second, we allow the researcher to impose restrictions on how the mediator $M$ responds to treatment. A leading example is the monotonicity assumption that the potential mediator $M(d)$ is increasing in $d$. In our motivating application, this imposes that providing men with information that other men are \emph{more} open to women working outside of the home can only \emph{increase} whether they sign up for the job-search service (in our main analysis, we restrict attention to the majority of men who initially underestimate others' openness, so the information plausibly updates beliefs in a common direction). We first consider the setting where monotonicity holds, and then introduce a more general framework that allows the researcher to impose arbitrary restrictions on the distribution of $(M(0),M(1))$, which nests monotonicity and relaxations thereof as special cases. 

A key observation is that under the sharp null of full mediation and the independence and monotonicity assumptions just described, the treatment $D$ is a valid instrumental variable for the local average treatment effect (LATE) of $M$ on $Y$. In the case of binary $D$ and binary $M$, the LATE assumptions are known to have testable implications \citep{balke_bounds_1997,kitagawa_test_2015, huber_testing_2015,mourifie_testing_2017}. Existing tools for testing the LATE assumptions can thus be used ``off-the-shelf'' for testing the sharp null of full mediation when $D$ and $M$ are binary, as we describe in more detail in \Cref{sec: binary case}. In our motivating example, the testable implications of the sharp null appear to be violated (significant at the 5\% level), and thus we can conclude that the effect of the information treatment does not operate entirely through job-search service sign-up. 

While existing tools can be used to test the sharp null in the case of a binary mediator $M$ and a monotonicity assumption, several questions remain. First, we may be interested in testing that the treatment effect is explained by a non-binary $M$, or by a set of mechanisms---can the approach above be applied when $M$ is non-binary and potentially multi-dimensional? Second, in some applications we may be concerned about violations of the monotonicity assumption---can one test the sharp null of full mediation under relaxations of this assumption? Third, if we reject the sharp null then we know that mechanisms other than $M$ must matter, but how large is the contribution of the alternative mechanisms? 

In \Cref{sec: general theory}, we develop a general framework that enables us
to tackle all of these questions. We allow the mediator $M$ to take on multiple
values and to have multiple dimensions, so long as it has finite support
$\{m_{0}, \dots, m_{K-1}\}$. We
also allow the researcher to place arbitrary restrictions on $\theta_{lk} =
P(M(0) =m_{l}, M(1)=m_{k})$, the fraction of individuals with $M(0) =m_{l}$ and $M(1) = m_{k}$. The monotonicity assumption in the case with scalar $M$ then corresponds to the special case where one imposes that $\theta_{lk} =0$ if $m_l > m_k$. Our framework allows the researcher to impose weaker versions of this requirement---e.g. by allowing for up to $\bar{d}$ share of the population to be defiers---or to completely eliminate the monotonicity requirement altogether. Our framework also allows for various extensions of monotonicity to the setting with multi-dimensional $M$---e.g. a partial monotonicity assumption that imposes that each dimension of $M$ is increasing in $d$.

We derive testable implications of the sharp null of full mediation in this
general setting. These testable implications imply that for any set $A$ and any value of the mediator $m_k$, the treatment effect on the compound outcome $\tilde{Y} = 1\{Y \in A, M=m_k \}$ should be no larger than the number of ``compliers'' with $M(0) = m_{l}$ and
$M(1)=m_{k}$ for some $l \neq k $. The intuition for this is that under the sharp null, there should be no effect of the treatment on the outcome for ``always-takers'' with $M(1) = M(0)$. It follows that the treatment effect on $\tilde{Y}$ can only be driven by compliers, and thus the treatment effect on $\tilde{Y}$ must be weakly smaller than the number of compliers. When $M$ is non-binary, a complication arises because the vector of shares of always-takers and compliers, denoted by $\theta$, is only partially-identified. The testable implication is therefore that there exists \emph{some} shares $\tilde\theta$ consistent with the observable data such that the inequalities described above are satisfied. Since the identified set for $\theta$ is characterized by linear inequalities, it is simple to verify whether such a $\tilde\theta$ exists by solving a linear program; we also show that the solution to the linear program has a closed-form solution under monotonicity. We further show that these testable implications are sharp in the sense that they exhaust all of the testable information in the data: if they are satisfied, there exists a distribution of potential outcomes (and potential mediators) consistent with the observable data such that the sharp null holds. 

We also provide lower bounds on the extent to which the sharp null is violated. In particular, our results imply lower bounds on the fraction of the individuals who have $M=m_{k}$ under both treatments (the ``$k$-always-takers'') who are nevertheless affected by the treatment, $\nu_{k} = P(Y(1,m_{k}) \neq Y(0,m_{k}) \mid M(1) = M(0) = m_{k})$. The lower bounds on the $\nu_k$ are informative about the prevalence of alternative mechanisms: if the lower bound on $\nu_{k}$ is large, then alternative mechanisms matter for a high fraction of $k$-always-takers. In \Cref{subsec: bounds on ADE}, we also derive bounds on the average direct effect for $k$-always-takers, $ADE_{k} = E[Y(1,m_{k}) - Y(0,m_{k}) \mid M(1) = M(0) = m_{k}]$.

In \Cref{sec: inference}, we show how one can conduct inference on the sharp null of full mediation, exploiting results from the literature on moment inequalities \citep{andrews_inference_2023, cox_simple_2022, fang_inference_2023}. In Monte Carlo simulations calibrated to our applications, we find good performance for the approach of \citet{cox_simple_2022}, and thus recommend it in practice. Although for simplicity our main theoretical results focus on the case where $D$ is randomized, in \Cref{sec: nonexperimental} we show that our results extend to other non-experimental settings, including settings with instrumental variables, conditional unconfoundedness, and distributional difference-in-differences. 

\Copy{pgraph:use-cases}{We anticipate that our results will have several potential use-cases in applications, as highlighted by our empirical examples in \Cref{sec: applications}. First, in many settings, there is an obvious mechanism by which the treatment would be expected to affect the outcome---often referred to as a ``mechanical effect''---and it is of economic interest to know whether there are other mechanisms at play. Our motivating example of \citet{bursztyn_misperceived_2020} is one such case, where there is the mechanical effect of the information on job applications through the job-search service, and we are interested in whether the information treatment also has an effect on other behavior outside of the lab. Our tests of the sharp null directly address whether the effect of the treatment is driven entirely by the mechanical effect: in \citet{bursztyn_misperceived_2020}, we reject that the impact of the information treatment on job applications is driven entirely by the mechanical effect of job-search service sign-up. This conclusion is of economic interest, since it suggests that an information treatment that was not tied to a job-search service would also have some effect on labor market outcomes. Our results also help us to quantify the magnitude of the alternative mechanisms: our lower bounds suggest that at least $11$
\unskip percent of ``never-takers'' who would not enroll in the job-search service regardless of treatment status would nevertheless be induced to apply for jobs by receiving the information treatment (compared with an overall ATE of $0.12$
\unskip).

In other settings, there may not be a focal ``mechanical effect'', but the researcher may observe that the treatment affects a particular $M$ (or set of $M$'s), and conjecture that this mediator explains the treatment effect. Our tests of the sharp null, along with accompanying lower bounds on $\nu_k$ and $ADE_k$, help to quantify the completeness of such conjectures. This is illustrated in our second application to \citet{baranov_maternal_2020}, who find that cognitive behavioral therapy for new mothers has an impact on women's economic outcomes. They conjecture that this effect may operate through increased presence of a grandmother in the home and improved relationship quality with the husband. Our results help us to quantify the completeness of these conjectures. Our tests reject the null hypothesis that either of these mechanisms on its own fully explains the treatment, with our lower bounds suggesting that at least 10\% of always-takers are affected by the treatment for each mediator. On the other hand, we cannot statistically reject that the two mechanisms together explain the treatment effect. This, of course, does not imply that these are the only two mechanisms, but rather that the data is statistically consistent with the hypothesis that the combination of these mechanisms explains the effect.\footnote{The literature on using short-run surrogates for long-run outcomes often justifies the statistical surrogacy assumption (in part) by arguing that the treatment affects the long-run outcome only through the short-run outcome \citep[e.g.][]{athey_estimating_2024}. In settings where both short- and long-run outcomes are available, our tests of the sharp null may also be useful for assessing the plausibility of these arguments.}}

We have developed the \href{https://github.com/jonathandroth/TestMechs}{\texttt{TestMechs}} R package to facilitate implementation of the methods in this paper.

\paragraph{Related literature.} Our work relates to a large literature on mediation analysis. We briefly overview a few relevant strands of the literature, with a non-exhaustive list of citations, and refer the reader to \citet{vanderweele_mediation_2016} and \citet{huber_review_2019} for more comprehensive reviews. Much of the mediation analysis literature focuses on identification of average direct effects and indirect effects \citep[e.g.][]{robins_identifiability_1992,pearl_direct_2001}. A key challenge is that even if the treatment $D$ is randomized, it is typically the case that the mediator $M$ is not, and thus it is difficult to identify the effect of $M$ on $Y$ (conditional on $D$). Various strands of the literature have identified the effect of $M$ on $Y$ by assuming conditional unconfoundedness for $M$ \citep[e.g.][]{imai_identification_2010}, using an instrument for $M$ \citep[e.g.][]{frolich_direct_2017}, or adopting difference-in-differences strategies \citep[e.g.][]{deuchert_direct_2019, schenk_mediation_2023}. In contrast, we focus on learning about mechanisms without imposing assumptions that identify the effect of $M$ on $Y$. The question we try to answer is different from most of the existing literature, however: rather than focus on average direct and indirect effects, we start by testing the \emph{sharp null} that the effect of $D$ on $Y$ is fully explained by a particular mechanism (or set of mechanisms) $M$.\footnote{\citet{miles_causal_2023} also considers a sharp null. However, his sharp null is that either $Y(d,m)$ depends only on $d$ \emph{or} $M(d)$ does not depend on $d$, whereas we consider the sharp null that $Y(d,m)$ depends only on $m$. His focus is also different: rather than testing this sharp null, he considers which measures of the indirect effect are zero when his sharp null is satisfied.} We further provide lower-bounds on the extent to which $M$ does not fully explain the effect of $D$ on $Y$ by lower-bounding the treatment effects for always-takers who have the same value of $M$ regardless of treatment status. We view our work as complementary to much of the literature on mediation analysis, as we impose different assumptions but also address different questions.

A key observation in our paper is that under the sharp null of full mediation, $D$ is an instrument for the effect of $M$ on $Y$. Thus, in the setting where $M$ is binary, existing tools for testing instrument validity with binary endogenous treatment can be used ``off-the-shelf'' to test the sharp null, both with monotonicity \citep{kitagawa_test_2015, huber_testing_2015,mourifie_testing_2017} and without monotonicity \citep{balke_bounds_1997,wang_falsification_2017, kedagni_generalized_2020}.\footnote{\citet{wang_falsification_2017} consider tests of instrument validity when instrument $Z$, treatment $D$, and outcome $Y$ are all binary, and one does not impose monotonicity. They observe that the testable implications imply lower bounds on the average controlled direct effect (ACDE) of $Z$ on $Y$. Although their focus is testing instrument validity, they note in the conclusion that such lower bounds might also be used for ``explaining causal mechanisms'' in experiments. This observation is thus a precursor to the connections between tests for instrument validity and testing mechanisms derived in the more general setting in our paper.} One of the key technical contributions of our paper is to derive sharp testable implications of the sharp null in the setting where $M$ is potentially multi-valued or multi-dimensional, and where one places arbitrary restrictions on the type shares (e.g. monotonicity or relaxations thereof). Based on the equivalence between testing the sharp null and testing instrument validity described above, our results immediately imply sharp testable implications for settings with a binary instrument and multi-valued treatment, which may be of independent interest. Our testable implications build on the work of \citet{sun_instrument_2023}, who derived non-sharp testable implications of instrument validity with multi-valued treatments under monotonicity.\footnote{Another related paper is \citet{kedagni_generalized_2020}, who derive testable implications of instrument validity with potentially multi-valued treatments, without monotonicity. Their testable implications assume a weaker notion of independence, however, which when mapped to our context would imply that $D$ is independent of $Y(\cdot,\cdot)$ but not $M(\cdot)$. Under this weaker notion of independence, their testable implications are sharp in the special case of binary treatment and outcome, but may not be sharp otherwise.}   

Our paper also relates to the literature on principal stratification \citep[]{frangakis_principal_2002,zhang_estimation_2003,lee_training_2009,flores_nonparametric_2010}. In particular, note that the sub-population of $k$-always-takers corresponds to the so-called \emph{principal stratum} with $M(1)=M(0)=m_{k}$. Our bounds on $ADE_k$, the average direct effect for $k$-always-takers, match those in the aforementioned papers in the special case where $M$ is binary and one imposes monotonicity. Our bounds on $ADE_k$ extend the existing results to cover settings with non-binary $M$ and/or relaxations of monotonicity. Our primary focus, however, is on testing the sharp null of full mediation, which implies that the \emph{fraction} of always-takers affected should be zero (a Fisherian sharp null), which is stronger than the weak null of a zero average effect studied in the literature on principal stratification.

Finally, we note that in empirical economics, mechanisms are often studied more informally, rather than using the formal tools for mediation discussed above. One common approach is to show the effects of $D$ on a variety of intermediate outcomes, and to conjecture that a particular intermediate outcome $M$ may be an important mechanism if $D$ has an effect on $M$ (see our application to \citealp{baranov_maternal_2020} below for an example). The tools developed in this paper give formal methodology for testing the completeness of these conjectures: is the data consistent with the hypothesized $M$ fully explaining the treatment effect, and if not, how important are alternative mechanisms? A second common approach for evaluating mechanisms is heterogeneity analysis: is the treatment effect on $Y$ larger in observable subgroups of the population for which the effect of $D$ on $M$ is larger? Although heterogeneity is often analyzed informally, this approach is sometimes formalized with an over-identification test that evaluates the null that, across subgroups defined by covariate cells, the conditional average treatment effect of $D$ on $Y$ is linear in the conditional average treatment effect of $D$ on $M$ \citep[e.g.][]{angrist_choice_2023, angrist_instrumental_2023}. This approach provides a valid test of our sharp null under the additional assumption that the effect of $M$ on $Y$ for compliers is constant across sub-groups. By contrast, we derive testable implications of the sharp null that do not assume constant effects and do not require the presence of covariates.\footnote{Moreover, our results indicate that the typical over-identification test does not exploit all the information in the data even under the assumption of constant effects: not only can one test the relationship of the average effects across covariate cells, but under the sharp null the restrictions that we derive should also hold \emph{within} covariate cells.}

\paragraph{Set-up and notation.} Let $Y$ denote a scalar outcome, $D$ a binary treatment, and $M \in \reals^p$ a $p$-dimensional vector of mediators with $K$ support points, $m_0,...,m_{K-1}$. We denote by $Y(d,m)$ the potential outcome under treatment $d$ and mediator $m$. Likewise, $M(d)$ denotes the potential mediator under treatment $d$. The researcher observes $(Y,M,D) = (Y(D,M(D)), M(D), D) \sim P_{obs}$. We use $P$ to denote the joint probability measure over observed variables and potential outcomes and mediators.

\section{Special Case: Binary Mediator} \label{sec: binary case}

We first consider the special case with a binary mediator $M$, which helps us to develop intuition and illustrate connections to the existing literature on testing instrument validity. In the notation just introduced, this corresponds to $K=2$, with $m_0 =0$ and $m_1=1$, so that $M \in \{0,1\}$.

To fix ideas,
consider the setting of \citet{bursztyn_misperceived_2020}. The authors conduct
a randomized controlled trial (RCT) in Saudi Arabia focused on women's economic
outcomes. Their analysis is motivated by the descriptive fact that at baseline in their experiment, the vast majority of men in Saudi Arabia under-estimate how open other men are to allowing women to work outside the home. After eliciting beliefs, they randomly assign a treated group of men to receive information about the other men's opinions. At the
end of the experiment, both treated and untreated men choose between signing their wives up
for a job-search service or taking a gift
card. \citet{bursztyn_misperceived_2020} find that the treatment has a positive
effect on enrollment in the job-search service and on longer-run economic outcomes for
women, such as applying and interviewing for jobs.

An important question in interpreting these results is whether the treatment increased long-run labor market outcomes solely by increasing take-up of the job-search service, or whether the information led men to change behavior in other ways. This question is important for understanding what might happen if one were to provide men with information about others' beliefs without offering the opportunity to sign up for the job-search service. \citet{bursztyn_misperceived_2020} write (p. 3017):
\begin{quote}
    It is difficult to separate the extent to which the longer-term effects are driven by the higher rate of access to the job service versus a persistent change in perceptions of the stigma associated with women working outside the home. 
\end{quote}

\noindent The authors provide some indirect evidence that the effects may not operate entirely through the job-search service---for example, there are effects on men's opinions in a follow-up survey---but they cannot directly link these long-run changes in opinions to economic outcomes. In what follows, we will show that in fact there is information in the data that is directly informative about the question of whether the effects on long-run labor market outcomes are driven solely by the job-search service. 

For notation, let $D$ be a binary indicator for receiving the information treatment, $M$ a
binary variable indicating job-search service sign-up, and $Y$ a binary variable
indicating applying for jobs three to five months after the experiment (i.e., a
longer-term labor supply outcome). We let $Y(d,m)$ denote whether a woman would apply for jobs as a function of treatment status $d$ and job-search service sign-up $m$, and let $M(d)$ denote job-search service sign-up as a function of treatment status. Since treatment is randomly assigned, it is reasonable to assume that it is independent of the potential outcomes and mediators, i.e. $D \indep (Y(\cdot,\cdot),M(\cdot))$. For our analysis in this section, we will also impose the monotonicity assumption that receiving the information treatment weakly increases job-search service sign-up, so that $M(1) \geq M(0)$ (almost surely). To make this assumption reasonable, we restrict our analysis to the majority of men who prior to the experiment under-estimate other men's openness, so that all men are provided with information that other men are \emph{more} open than they initially expected, which we expect will increase job-search service sign-up. In the subsequent sections, we will show how this monotonicity assumption can be relaxed, but imposing it will make it easier to highlight the connections to instrumental variables. 

We now formalize the null hypothesis that the information treatment only affects long-run outcomes through its effect on job-search service sign-up. In particular, we say that the \textit{sharp null of full mediation} is satisfied if
\begin{equation}
  Y(0,m) = Y(1,m) \equiv Y(m) \text{ almost surely, for all } m \in \{0,1\},
\end{equation}
i.e. the treatment impacts the outcome only through its impact on $M$. If
the sharp null holds, signing up for the job-search service is the only mechanism
that matters for long-run job applications. On the other hand, if we reject the
sharp null, there is evidence that other mechanisms play a role for at least some people---i.e., there is some impact of changes in beliefs on long-run outcomes that does not operate purely through sign-up for the job-search service at the end of the experiment. 

Our first main observation is that if the sharp null holds (together with our assumptions of independence and monotonicity), then $D$ is a valid instrument for the LATE of $M$ on $Y$. This implies that testing the sharp null in this setting
is equivalent to testing the validity of the LATE assumptions when both the
treatment and instrument are binary. However, prior work has shown that in settings with a binary instrument and treatment, the LATE assumptions have testable implications \citep{kitagawa_test_2015,
  huber_testing_2015, mourifie_testing_2017}, and thus such tools can be used to test the sharp null.\footnote{More precisely, these tests are joint tests of the sharp null along with the independence and monotonicity assumptions. However, if we maintain that the latter two hold, then any violations must be due to violations of the sharp null. We explore relaxations of the monotonicity assumption in subsequent sections.} Applying the results in \cite{kitagawa_test_2015}, with $M$
playing the role of treatment and $D$ the role of instrument, we obtain the
following sharp testable implications:
\begin{align}\label{eq:bin_sharp_imp}
  \begin{aligned}
    P(Y \in A, M=0 \mid D=0) & \geq P(Y \in A, M=0 \mid D=1) \text{ and} \\
    P(Y \in A, M=1 \mid D=1) & \geq P(Y \in A, M=1 \mid D=0),
  \end{aligned}
\end{align}
for all Borel sets $A$.

To gain intuition for these testable restrictions, observe that under our monotonicity assumption we can divide the population into ``always-takers'' who enroll in the job-search service regardless of treatment ($M(0)=M(1)=1$), ``never-takers'' who do not enroll regardless of treatment ($M(0)=M(1)=0$), and ``compliers'' who enroll only if treated ($M(0)=0,M(1)=1$). Now, consider the compound outcome $\tilde{Y} = 1\{Y \in A, M=0\}$. For example, if $A= \{1\}$, then in our running example $\tilde{Y}$ is an indicator for the joint event of applying for a job and not signing up for the job-search service. Note that always-takers have $M=1$ under both treatments, and thus always have $\tilde{Y}=0$ regardless of treatment status. Next, observe that if the sharp null holds, then never-takers must also have the same value of $\tilde{Y}$ under both treatments: by definition, they have $M=0$ under both treatments, and hence under the sharp null that $D$ affects $Y$ only through $M$, they also have the same value of $Y$ under both treatments. Since $\tilde{Y}$ is just a compound outcome involving $Y$ and $M$, it follows that they have the same value of $\tilde{Y}$ under both treatments. Observe, further, that the treatment effect of $D$ on $\tilde{Y}$ for compliers must be weakly negative, since compliers have $M=1$ when they are treated, and thus when compliers are treated they have $\tilde{Y} = 0$. It follows that under the sharp null, there must be a weakly negative treatment effect of $D$ on $\tilde{Y}$, and hence
$$P(Y\in A, M=0 \mid D=1) - P(Y \in A, M=0 \mid D=0) \leq 0 ,$$
\noindent which gives the first testable implication in \eqref{eq:bin_sharp_imp}. If this implication is violated in the data, then we can conclude that---in violation of the sharp null---there must be an effect of the treatment for some never-takers. The second testable implication in \eqref{eq:bin_sharp_imp} can be derived analogously using the compound outcome of the form $1\{ Y \in A, M=1 \}$.

\begin{figure}[!ht]
  \includegraphics[width =
  0.6\linewidth]{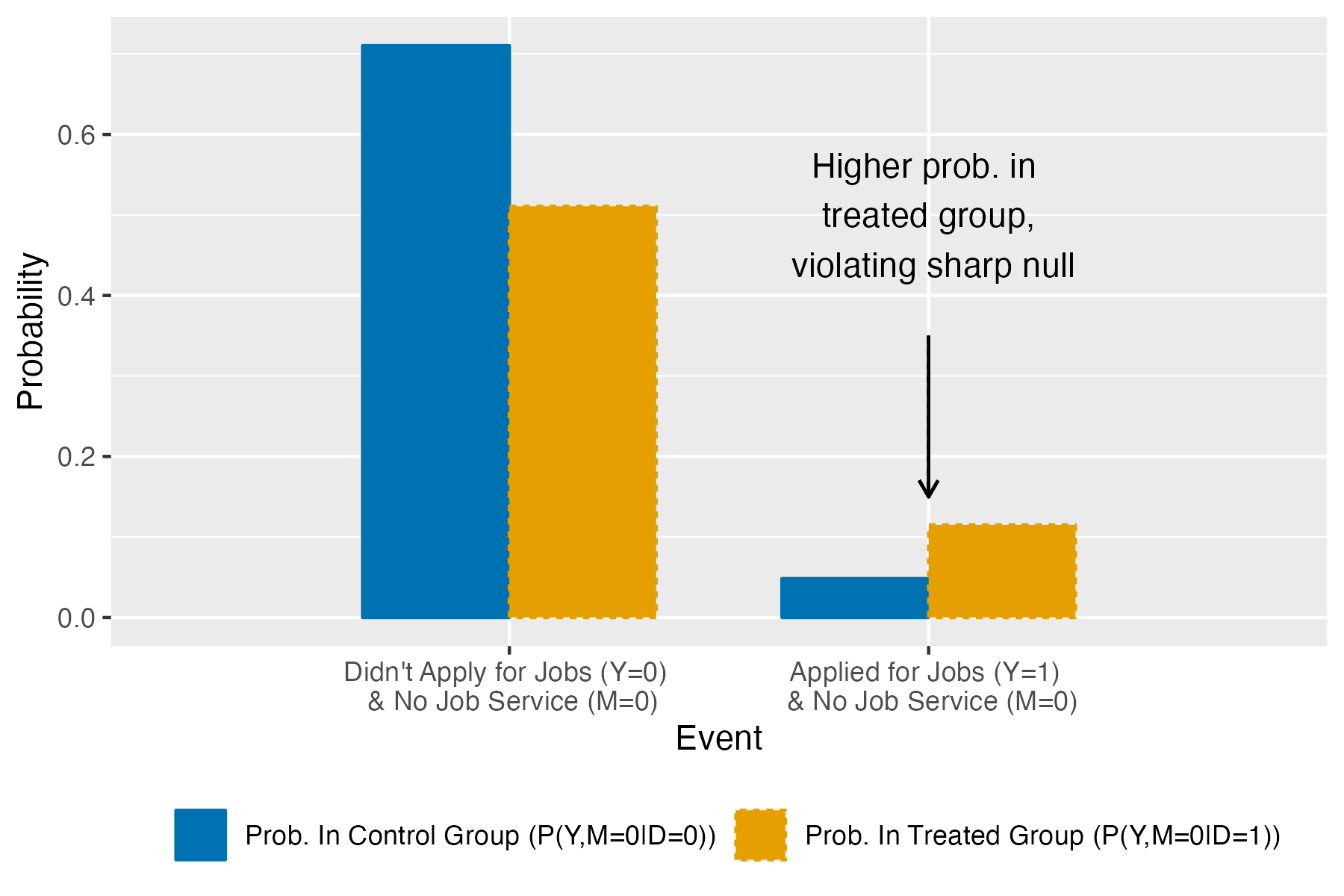}
    \caption{Illustration of Testable Implications in \citet{bursztyn_misperceived_2020}}\label{fig:bursztyn_bin}
    \floatfoot{Note: This figure shows estimates of the probabilities $P(Y=y, M=0 \mid D=d)$ for $d=0,1$ and $y=0,1$ in the application to \citet{bursztyn_misperceived_2020}. For example, $P(Y=1,M=0 \mid D=0)$ is the probability that one both applies for a job \emph{and} does not sign up for the job-search service conditional on being in the control group. Under the sharp null of full mediation, it should be that these probabilities are higher in the control group, i.e. $P(Y=y, M=0 \mid D=0) \geq P(Y=y, M=0 \mid D=1)$ for $y=0,1$. We see, however, that this inequality is violated in the empirical distribution for $y=1$: more women apply for jobs and don't use the job-search service in the treated group, as indicated by the black arrow.}
\end{figure}

The argument above implies that if the sharp null holds in \citet{bursztyn_misperceived_2020}, there should be a negative treatment effect on the compound outcome $1\{Y=1,M=0\}$, i.e. there should be fewer women in the treated group who both apply for jobs and don't use the job service. However, as shown in \Cref{fig:bursztyn_bin}, the empirical
distribution shows that the opposite is true: there are more women who apply for jobs and do not sign up for the job-search service in the treated group ($\hat{P}(Y = 1, M=0 \mid D=1) > \hat{P}(Y = 1, M=0 \mid D=0)$), indicating a
violation of the sharp null. This difference is statistically significant at the 5\% level, as we will describe in more detail in \Cref{sec: applications} after we describe methods for conducting inference. 

The data thus reject the sharp null hypothesis that the impact of the information treatment on job applications operates purely through job-search service sign-up. In particular, the data suggest that some never-takers must have their outcome affected by the treatment. We can thus conclude that there is some impact of changes in beliefs on job applications that does not operate mechanically through signing up for the job-search service. 

The analysis so far shows that tools originally developed for testing the LATE assumptions can be useful for testing hypotheses about mechanisms. However, several questions remain. First, our rejection of the null implies that the treatment affects the outcome through mechanisms other than job-search service sign-up, but how big are these alternative mechanisms? Second, our analysis relied on the monotonicity assumption that treatment increases job-search service sign-up, but what if we would like to relax this assumption? Third, while our motivating example had a binary $M$, in many cases we may be interested in testing that the treatment is explained by a non-binary mechanism, or by the combination of multiple mechanisms. Can the approach be extended to such cases?

In the subsequent section, we develop a general theoretical framework that allows us to address all of these questions. Our framework accommodates mechanisms $M$ that are potentially multi-valued or multi-dimensional, and allows for relaxations of the monotonicity assumption. Further, in addition to deriving testable implications of the sharp null, we also derive lower bounds on the extent to which the alternative mechanisms matter---in particular, we derive bounds on the fraction of always-takers (or never-takers) that are affected by the treatment, as well as the average effect of the treatment for these always-takers.

\section{Theory: General Case \label{sec: general theory}}

We now consider the general case where $M$ is a $p$-dimensional vector with
finite support $\{m_0,...,m_{K-1}\}$. We denote by $G=lk$ the event that $M(0) = m_l$ and $M(1)=m_k$. We refer to individuals with $G=kk$ as the $k$-always-takers, and individuals with $G = lk$ for $l \neq k$ as the $lk$-compliers. (Note that the terms ``always-taker'' and ``complier'' are used somewhat broadly here. For example, a ``never-taker'' in the case where $M$ is binary would be referred to as $0$-always taker, and likewise a defier would be a $10$-complier.) We denote by $\theta_{lk} := P(M(0)=m_l,M(1)=m_k)$ the fraction of the population of type $G=lk$, and let $\theta \in \reals^{K^2}$ be the vector in the $(K^2-1)$-dimensional simplex that collects the $\theta_{lk}$.

Extending the definition from the previous section, we say that the sharp null of full mediation holds if 
$$Y(0,m) = Y(1,m) \equiv Y(m) \text{ almost surely, for all } m \in \{m_0,...,m_{K-1}\}.$$
\noindent We note that if $M$ is multi-dimensional with, say, the first dimension corresponding to mechanism $A$ and the second corresponding to mechanism $B$, then the sharp null imposes that the treatment operates on $Y$ only through its joint effect on mechanisms $A$ and $B$.

For simplicity, we assume in this section that treatment assignment is independent of the potential outcomes and mediators. In \Cref{sec: nonexperimental}, we show how the results extend to other settings, such as instrumental variables, conditional unconfoundedness, and distributional difference-in-differences. 

\begin{asm}[Independence and overlap] \label{asm: independence}
The treatment is independent of the potential outcomes and mediators, $D \indep (Y(\cdot,\cdot), M(\cdot) ) $, with $0 < P(D=1) < 1$.
\end{asm}
For our identification results, we allow for the researcher to place arbitrary restrictions on the shares of each compliance type. 

\begin{asm}[Restrictions on type shares] \label{asm: restricted theta}
  $\theta \in R$ for $R \subseteq \Delta$ a closed non-empty set, where $\Delta$ denotes the $(K^2-1)$-dimensional simplex.
\end{asm}

\noindent We briefly review a few examples of restrictions on $\theta$ (i.e. choices of $R$) that may be natural in some applications. 

\begin{example}[Monotonicity and relaxations thereof] \label{ex: fully ordered M}
First, consider the case where $M$ is fully-ordered, so that $m_0 < m_1 <... < m_{K-1}$. This nests the binary example from the previous section as the special case where $K=2$. Then the monotonicity assumption that $M(1) \geq M(0)$ corresponds to the restriction \begin{equation} R = \{\theta \in \Delta : \theta_{lk} = 0 \text{ if } l>k \} . \label{eqn: delta for monotonicity}\end{equation}  
\noindent One could also weaken this assumption by, for example, allowing for up to $\bar{d}$ fraction of the population to be defiers, which corresponds to setting $R = \big\{\theta \in \Delta : \sum_{l,k: l>k} \theta_{lk} \leq \bar{d} \big\} .$
\end{example}

\begin{example}[Elementwise monotonicity] \label{ex: partial monotonicity}
Suppose that $M$ is a $p$-dimensional vector for $p>1$. It may sometimes be reasonable to impose that each element of $M(d)$ is increasing in $d$. This can be achieved by setting $R = \left\{\theta \in \Delta : \theta_{lk} = 0 \text{ if } m_l \not\preceq m_k \right\} ,$
where $m_l \preceq m_k$ if each element of $m_l$ is less-than-or-equal to the corresponding element of $m_k$.\footnote{Analogous logic could be used to impose that $M(0) \preceq M(1)$ in \emph{any} partial order, not just the elementwise one.} Similar to the previous example, one could also allow for up to $\bar{d}$ fraction of the population to have $M(0) \not\preceq M(1)$.   
\end{example}

\begin{example}[Bounded effect of $D$ on $M$]
In some settings, it may be reasonable to impose that the treatment does not have too large an effect on $M$, at least for most people. This could be formalized by setting
$$R = \big\{ \theta \in \Delta : \sum_{ l,k :||m_l -m_k|| > \kappa } \theta_{lk}  \leq \bar{d} \big\}.$$
\noindent This imposes that at most $\bar{d}$ fraction of the population has $||M(1) - M(0)|| > \kappa$. 
\end{example}

\begin{example}[No restrictions] \label{ex: no restrictions}
If the researcher is not willing to impose any restrictions on compliance types, then they can simply set $R = \Delta$.    
\end{example}

In contrast to the special case in \Cref{sec: binary case}, where the shares of each type were point-identified, in our general framework the vector of type shares $\theta$ may only be partially-identified. For example, if one relaxes the monotonicity imposed in \Cref{sec: binary case}, then analogous to the setting of instrumental variables without monotonicity \citep[e.g.][]{huber_sharp_2017}, the share of defiers $\theta_{10}$ will generically be partially identified.\footnote{As a concrete example, suppose that $P(M=1 \mid D=1) = 0.5$ and $P(M=1 \mid D=0) = 0.3$. Then the data is consistent with there being no defiers (by setting $\theta_{11} = 0.3$, $\theta_{01} = 0.2$, $\theta_{00} = 0.5$, and $\theta_{10} =0$) but it is also consistent with up to 0.3 fraction of the population being defiers (by setting $\theta_{11} = 0$, $\theta_{01} = 0.5$, $\theta_{00} = 0.2$, $\theta_{10} = 0.3$).} Partial identification of $\theta$ can also arise when $M$ is multi-valued even if one imposes monotonicity. To see this, suppose that $M \in \{0,1,2\}$ and the marginal distributions of $M \mid D$ are as given in \Cref{fig:type-shares-illustration}, panel (a). As can be seen in the figure, the treated group has a 0.2 higher probability that $M=2$ and a 0.2 lower probability that $M=0$ relative to the control group. This is consistent with 20\% of the population being $02$-compliers and there being no other complier types (i.e. $\theta_{02} = 0.2$, $\theta_{01}=\theta_{12}=0$), as shown in \Cref{fig:type-shares-illustration}, panel (b). However, it is also consistent with a ``cascade'' in which 20\% of the population is $01$-compliers, and another 20\% of the population is $12$-compliers (i.e. $\theta_{01}=\theta_{12}=0.2$, $\theta_{02}=0$), as shown in \Cref{fig:type-shares-illustration}, panel (c). 
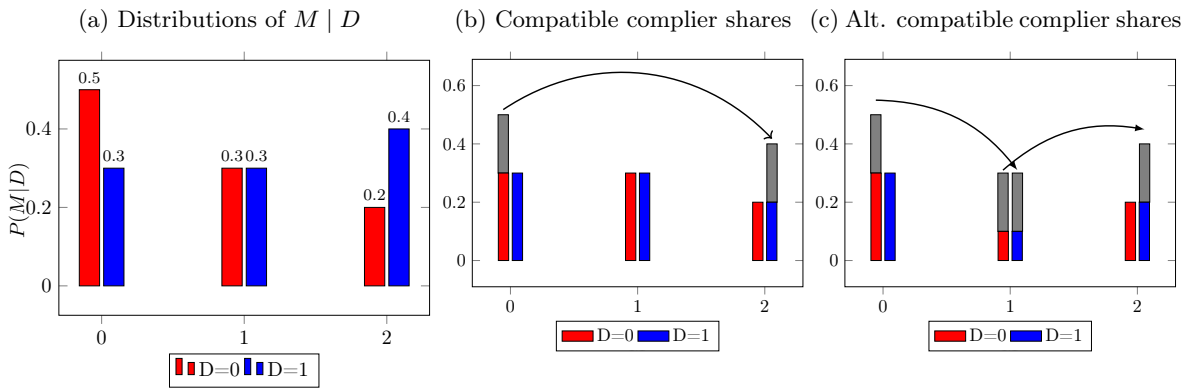
\begin{figure}[!hbt]
    \centering
    \subfloat[Distributions of $M \mid D$]{\resizebox{0.323\linewidth}{!}{\begin{tikzpicture}
\begin{axis}[
    ybar,
    width=8cm,
    height=6cm,
    enlargelimits=0.15,
    legend style={at={(0.5,1.02)},
    anchor=south,legend columns=-1},
    ylabel={$P(M | D)$},
    ylabel style={anchor=west, yshift=0.25cm, xshift=-1cm},
    symbolic x coords={0,1,2},
    xtick=data,
    nodes near coords,
    nodes near coords align={vertical},
    nodes near coords style={font=\footnotesize}, %
    x label style={at={(axis description cs:0.5,-0.2)},anchor=north},
    xlabel={M},
    ymin=0, %
    ymax = 0.6
    ]
\addplot [fill=cbOrange,draw=black] coordinates {
    (0,0.5)
    (1,0.3)
    (2,0.2)
};
\addplot [fill=cbSkyBlue,draw=black] coordinates {
    (0,0.3)
    (1,0.3)
    (2,0.4)
};
\legend{D=0,D=1}
\end{axis}
\end{tikzpicture}}} 
    \subfloat[Compatible complier \\ shares]{\resizebox{0.3\linewidth}{!}{\begin{tikzpicture}
\begin{axis}[
    ybar,
    width=8cm,
    height=6cm,
    enlargelimits=0.15,
    legend style={at={(0.5,1.02)},
    anchor=south,legend columns=-1},
    symbolic x coords={0,1,2},
    xtick=data,
    x label style={at={(axis description cs:0.5,-0.2)},anchor=north},
    xlabel={M},
    ymin=0,
    ymax=0.6,
    bar width=0.2cm,
    area legend
    ]

\draw[fill=gray,draw=black] ([xshift=-0.24cm]axis cs:0,0.3) rectangle ([xshift=-0.04cm]axis cs:0,0.5);

\draw[fill=gray,draw=black] ([xshift=0.04cm]axis cs:2,0.2) rectangle ([xshift=0.24cm]axis cs:2,0.4);

\addplot [fill=cbOrange,draw=black] coordinates {
    (0,0.3)
    (1,0.3)
    (2,0.2)
};

\addplot [fill=cbSkyBlue,draw=black] coordinates {
    (0,0.3)
    (1,0.3)
    (2,0.2)
};

\draw[->,thick] ([xshift=-0.14cm, yshift=0.1cm]axis cs:0,0.5) to[bend left=40] ([xshift=0.14cm, yshift=0.1cm]axis cs:2,0.4);

\legend{D=0,D=1}
\end{axis}
\end{tikzpicture}}} 
    \subfloat[Alternative compatible \\ complier shares]{\resizebox{0.3\linewidth}{!}{\begin{tikzpicture}
\begin{axis}[
    ybar,
    width=8cm,
    height=6cm,
    enlargelimits=0.15,
    legend style={at={(0.5,1.02)},
    anchor=south,legend columns=-1},
    symbolic x coords={0,1,2},
    xtick=data,
    x label style={at={(axis description cs:0.5,-0.2)},anchor=north},
    xlabel={M},
    ymin=0,
    ymax=0.6,
    bar width=0.2cm,
    area legend
    ]

\draw[fill=gray,draw=black] ([xshift=-0.24cm]axis cs:0,0.3) rectangle ([xshift=-0.04cm]axis cs:0,0.5);

\draw[fill=gray,draw=black] ([xshift=-0.24cm]axis cs:1,0.1) rectangle ([xshift=-0.04cm]axis cs:1,0.3);

\draw[fill=gray,draw=black] ([xshift=0.04cm]axis cs:1,0.1) rectangle ([xshift=0.24cm]axis cs:1,0.3);

\draw[fill=gray,draw=black] ([xshift=0.04cm]axis cs:2,0.2) rectangle ([xshift=0.24cm]axis cs:2,0.4);

\addplot [fill=cbOrange,draw=black] coordinates {
    (0,0.3)
    (1,0.1)
    (2,0.2)
};

\addplot [fill=cbSkyBlue,draw=black] coordinates {
    (0,0.3)
    (1,0.1)
    (2,0.2)
};

\draw[-latex, thick] ([xshift=-0.14cm]axis cs:0,0.55) to[bend left=25] ([xshift=0.14cm]axis cs:1,0.31);
\draw[-latex, thick] ([xshift=-0.14cm]axis cs:1,0.31) to[bend left=25] ([xshift=0.14cm]axis cs:2,0.45);

\legend{D=0,D=1}
\end{axis}
\end{tikzpicture}}} 
       \caption{Illustration of partial identification of type shares}
    \label{fig:type-shares-illustration}
\end{figure}

We will denote by $\Theta_I$ the set of possible values for $\theta$ (i.e. joint distributions on $(M(0),M(1))$) that are consistent with the observed distributions of $M \mid D$. Formally, we define the identified set $\Theta_I$ to be the set of values of $\tilde\theta$ such that 
\begin{align*}
  &\sum_l \tilde\theta_{kl} = P(M=m_k \mid D=0) \text{ for } k=0,...,K-1  \hspace{1cm} & \text{(Match marginals for $M \mid D=0$)} \\
  &\sum_l \tilde\theta_{lk} = P(M=m_k \mid D=1) \text{ for } k=0,...,K-1 \hspace{1cm} &
                                                                               \text{(Match
                                                                               marginals
                                                                               for
                                                                               $M \mid D=1$)}
  \\
  & \tilde\theta \in R \hspace{1cm} &
                                \text{(Type Share
                                Restrictions)}.
\end{align*}

\noindent For clarity of notation, we will use $\theta$ for the ``true'' shares and use $\tilde\theta$ to denote a generic element of the identified set $\Theta_I$. It is worth noting that the first two restrictions above are linear in $\tilde\theta$. Thus, if $R$ is characterized by linear restrictions (as is the case in Examples \ref{ex: fully ordered M}-\ref{ex: no restrictions} above), then $\Theta_I$ is characterized by linear constraints, and thus quantities such as $\max_{\tilde\theta \in \Theta_I} \tilde\theta_{kk}$ can be calculated by linear programming. This observation will be useful for practical implementation of the testable implications below, which involve optimizations over $\Theta_I$.

We now derive lower-bounds on the fraction of always-takers whose outcome is affected by the treatment despite having the same value of $M$ under both treatments.\footnote{In \Cref{subsec: bounds on ADE}, we derive bounds on the average effect of the treatment for the $k$-always-takers.} These lower bounds lead naturally to tests of the sharp null of full mediation, under which the fraction of always-takers affected should be zero. To be more precise, we define
$$\nu_k := P(Y(1,m_k) \neq Y(0,m_k) \mid G =kk)$$ 
\noindent to be the fraction of $k$-always-takers whose outcome is affected by the treatment despite always having $M=m_k$ under both treatments. The $\nu_k$ are a measure of the strength of mechanisms other than $M$: they tell us what fraction of the $k$-always-takers has a direct effect of the treatment. Under the sharp null of full mediation, $Y(1,m_k) = Y(0,m_k)$ with probability 1, and thus $\nu_k =0$ for all $k$. By contrast, if $\nu_k$ is close to 1 for a particular $k$, then alternative mechanisms other than $M$ matter for nearly all $k$-always-takers. 

Our first main result provides a lower bound on the $\nu_k$ as a function of the observable data and the type shares $\theta$. To simplify notation, let 
$$\Delta_k(A) := P(Y \in A, M=m_k \mid D=1) - P(Y \in A, M=m_k \mid D=0)$$
\noindent be the difference in the probability that $Y \in A$ and $M=m_k$ between the treated and control groups. Let $(x)_+ := \max\{x, 0\}$. We then have the following sharp lower bound on the fraction of $k$-always-takers affected by the treatment.

\begin{prop} \label{prop: testable implications for tv} 
\leavevmode
\begin{enumerate}
    \item 
    \textbf{(Lower bounds on $\nu_k$)} Suppose Assumptions \ref{asm: independence} and \ref{asm: restricted theta} hold. The true shares $\theta$ satisfy 
    \begin{align}
    \theta_{kk} \nu_{k} &\geq \left( \sup_A \Delta_k(A) - \sum_{l: l
  \neq k} \theta_{lk} \right)_+ \label{eqn: testable implication tv}
\end{align}
    \noindent for $k=0,...,K-1$. Since $\theta \in \Theta_I$, it follows that there exists some $\tilde\theta \in \Theta_I$ such that
    \begin{align}
    \tilde\theta_{kk} \nu_{k} &\geq \left( \sup_A \Delta_k(A) - \sum_{l: l
  \neq k} \tilde\theta_{lk} \right)_+ \label{eqn: testable implication tv - feasible}
\end{align}
   holds for all $k=0,...,K-1$, and hence $\nu_{k} \geq \inf \{ \nu_k \geq 0 \,:\, \exists \tilde\theta \in \Theta_{I} \text{ s.t. } \eqref{eqn: testable implication tv - feasible} \text{ holds}\}$ for all $k$.

    \item
    \textbf{(Sharpness)} The bound in \eqref{eqn: testable implication tv - feasible} is sharp: for any $\tilde\theta \in \Theta_I$, there exists a joint distribution $P^{\dagger}$ for $(Y(\cdot,\cdot),M(\cdot),D)$ consistent with the observable data\footnote{We say that $P^{\dagger}$ is consistent with the observable data if $(Y(D,M(D)), M(D), D) \sim P_{obs}$ under $P^\dagger$, i.e. the distribution of realized $(Y,M,D)$ under $P^{\dagger}$ matches the observed data distribution $P_{obs}$.} and Assumptions \ref{asm: independence} and \ref{asm: restricted theta} such that for all $l$ and $k$, $P^\dagger(G=lk) = \tilde{\theta}_{lk}$ and 
    $$\tilde{\theta}_{kk} \nu_k^\dagger = \left( \sup_A \Delta_k(A) - \sum_{l: l
  \neq k} \tilde\theta_{lk} \right)_+, \text{ where } \nu_k^\dagger = P^\dagger(Y(1,m_{k}) \neq Y(0,m_{k}) \mid G=kk) .$$ Moreover, $P^\dagger(Y(1,m) \neq Y(0,m) \mid G=lk) = 0$ if either $l \neq k$ or $m \not\in \{m_{l},m_{k}\}$. 
    
\end{enumerate}
\end{prop}

Equation \eqref{eqn: testable implication tv} gives a lower-bound on the fraction of $k$-always-takers whose outcome is affected by the treatment, $\nu_k$, involving the true type shares $\theta$ and functions of the observable data (the $\Delta_k$). The true $\theta$ will generically not be point-identified, and so \eqref{eqn: testable implication tv} cannot be used directly to give a feasible lower bound on $\nu_k$. However, we know that the true $\theta$ must lie in the identified set $\Theta_I$. This leads to the feasible implication given in \eqref{eqn: testable implication tv - feasible} which replaces $\theta$ in \eqref{eqn: testable implication tv} with \emph{some} $\tilde\theta \in \Theta_I$. The second part of \Cref{prop: testable implications for tv} shows that the bound given in \eqref{eqn: testable implication tv - feasible} is sharp in the sense that there exists a distribution of primitives consistent with the observable data such that the lower bound holds with equality. It further shows that under this distribution of primitives, there is no direct effect of treatment for complier types; hence, we can only obtain non-trivial lower bounds on direct effects for the always-takers.

 Recall that under the sharp null of full mediation, the fraction of always-takers whose outcome is affected by the treatment should be zero. We thus immediately obtain the following testable implications of the sharp null by setting $\nu_k = 0$ in \eqref{eqn: testable implication tv - feasible}.

\begin{cor}[Testable implications of sharp null] \label{cor: testable implications sharp null} \leavevmode
\begin{enumerate}
    \item 
    \textbf{(Testable implications)} Suppose Assumptions \ref{asm: independence} and \ref{asm: restricted theta} hold. If the sharp null of full mediation is satisfied, then there exists $\tilde\theta \in \Theta_I$ such that for all $k= 0,...,K-1$,
    \begin{equation}
    \sup_A \Delta_k(A) \leq \sum_{l: l
  \neq k} \tilde\theta_{lk}. \label{eqn: testable implications sharp null}
    \end{equation}
    \item
    \textbf{(Sharpness)} The testable implication in \eqref{eqn: testable implications sharp null} is sharp: if there exists $\tilde\theta \in \Theta_I$ such that \eqref{eqn: testable implications sharp null} holds for all $k$, then there exists a joint distribution $P^\dagger$ for $(Y(\cdot,\cdot),M(\cdot),D)$ consistent with the observable data and Assumptions \ref{asm: independence} and \ref{asm: restricted theta} such that the sharp null of full mediation holds.

\end{enumerate}     
\end{cor}

\paragraph{Intuition.} Observe that since the treatment is randomly assigned, $\Delta_k(A)$ is simply the average effect of the treatment $D$ on the compound outcome $\tilde{Y} = 1\{Y \in A, M=m_{k}\}$. Note that under the sharp null of full mediation, the treatment effect of $D$ on $\tilde{Y}$ should be zero for always-takers, since they have the same value of $M$ and $Y$ under both treatments. This implies that under the sharp null the effect of $D$ on $\tilde{Y}$ is driven only by compliers and thus cannot be ``too large''. This is precisely what is captured in \Cref{cor: testable implications sharp null}, which shows that under the sharp null, $\Delta_k(A)$ should be bounded above by the total mass of $lk$-compliers, $\sum_{l: l \neq k} \theta_{lk}$, regardless of the choice of $A$. If, in fact, the treatment effect on $\tilde{Y}$ is larger than the number of $lk$-compliers, then it must be that some $k$-always-takers had their outcome affected by the treatment, in violation of the sharp null. Indeed, the lower bound on the fraction of $k$-always-takers whose outcome is affected by the treatment ($\nu_k$) given in \eqref{eqn: testable implication tv} is proportional to the positive part of the difference between $\sup_A \Delta_k(A)$ and the number of $lk$-compliers.  

A slightly more formal sketch of the argument is as follows. We can write $\Delta_k(A) = E[\tilde{\tau}]$, where $\tilde{\tau} $ is the individual-level treatment effect of $D$ on $\tilde{Y}$.\footnote{Formally, $\tilde\tau = \tilde{Y}(1) - \tilde{Y}(0)$ for $\tilde{Y}(d) = 1\left\{Y(d,M(d)) \in A, M(d) =m_{k} \right\}$.} Since $\tilde{Y}$ is a binary outcome, the treatment effect $\tilde{\tau}$ must be in $\{-1,0,1\}$. We now argue that $\tilde\tau$ can equal 1 only if an individual is either an $lk$-complier, or a $k$-always taker with $Y(1,m_{k}) \neq Y(0,m_{k})$. To see why this is a case, note that for $\tilde{\tau}$ to be 1, an individual must have $M(1) =m_{k}$, and thus must be either an $lk$-complier or a $k$-always taker. However, a $k$-always taker can have a treatment effect on $\tilde{Y}$ of 1 only if $Y(1,m_{k}) \in A$ and $Y(0,m_{k}) \not\in A$, which implies that $Y(1,m_{k}) \neq Y(0,m_{k})$. It follows that
\begin{align*}
\Delta_k(A) & \leq \underbrace{P(G = kk, Y(1,m_{k}) \neq Y(0,m_{k})) }_{\text{Prob of $k$-AT w/ $Y(1,m_{k}) \neq Y(0,m_{k})$}} + \underbrace{\sum_{l: l \neq k} P(G=lk)}_{\text{Prob of $lk$-complier}}.
\end{align*}
\noindent Using the fact that $\theta_{lk} = P(G=lk)$ by definition, we can rewrite the inequality as
\begin{align*}
  \Delta_k(A) \leq \theta_{kk} \cdot P( Y(1,m_{k}) \neq Y(0,m_{k}) \mid G = kk) + \sum_{l : l \neq k} \theta_{lk}.
\end{align*}
\noindent Rearranging terms, we obtain that
$$\theta_{kk} P( Y(1,m_{k}) \neq Y(0,m_{k}) \mid G = kk)  \geq \Delta_k(A) - \sum_{l: l \neq k} \theta_{lk} ,$$
\noindent which together with the fact that probabilities are non-negative yields \eqref{eqn: testable implication tv}. $\blacktriangle$

\paragraph{Computation of bounds.} Suppose we are interested in computing the lower-bound on $\nu_k$ for a particular $k$. Recall that for any $\tilde\theta \in \Theta_I$, we have $P(M=m_{k} \mid D=1) = \tilde\theta_{kk} + \sum_{l: l \neq k} \tilde\theta_{lk}$. It follows that we can re-write \eqref{eqn: testable implication tv - feasible} as
$$\tilde{\theta}_{kk} \nu_k \geq \left(\sup_A \Delta_k(A) - (P(M=m_{k} \mid D=1) - \tilde\theta_{kk}) \right)_+ ,$$
where now the lower-bound depends on $\tilde\theta$ only through $\tilde\theta_{kk}$. To compute a lower-bound on $\nu_k$, we must minimize the lower-bound for $\nu_k$ given in the previous display over $\tilde\theta \in \Theta_I$. It can be shown (see \Cref{lem: min at thetakkmin}) that the minimum is actually obtained at the minimum possible value of $\tilde\theta_{kk}$, i.e. by plugging in $\tilde{\theta}_{kk}^{min} := \inf_{\tilde\theta \in \Theta_{I}} \tilde\theta_{kk}$ into the expression in the previous display. When $R$ is a polyhedron (as in our examples above), the identified set $\Theta_I$ is characterized by linear inequalities, and thus  $\tilde{\theta}_{kk}^{min}$ can be easily computed by solving a linear program. Assuming $\tilde\theta_{kk}^{min} > 0$, we then obtain the bound
$$\nu_k \geq \frac{1}{\tilde\theta_{kk}^{min}} \left(\sup_A \Delta_k(A) - (P(M=m_{k} \mid D=1) - \tilde\theta_{kk}^{min}) \right)_+ .$$
Similarly, to test whether the observable data is compatible with the sharp null we must verify whether there is any $\tilde\theta \in \Theta_I$ such that $\sup_A \Delta_k(A) \leq \sum_{l: l \neq k} \tilde\theta_{lk}$ for all $k$. By the same argument as in the previous paragraph, this is equivalent to testing whether there is any $\tilde\theta \in \Theta_I$ such that $\sup_A \Delta_k(A) \leq P(M=m_{k} \mid D=1) - \tilde\theta_{kk}$ for all $k$. Such a $\tilde\theta \in \Theta_I$ exists if and only if the solution to the linear program
\begin{equation} 
\min_{s \in \reals, \tilde\theta \in \Theta_I} s \text{ s.t. } \sup_A \Delta_k(A) \leq P(M=m_{k} \mid D=1) - \tilde\theta_{kk} + s \text{ for all } k \label{eqn: testable implications - lp}\end{equation}
is weakly negative, and so given knowledge of the distribution of the observable data, testing the implications of the sharp null is equivalent to solving a linear program.\footnote{We note that linear programming has been used for tractability in a variety of related but distinct partial identification settings; see, e.g. \citet{mogstad_policy_2024}, \citet{ji_model-agnostic_2024}, \citet{yap_sensitivity_2025} for some recent contributions.\label{fn:other-lp-papers}}

\begin{rem}[Closed-form solution with fully-ordered, monotone $M$] \label{rmk: closed form thetakk}
Consider the case where $M$ is fully-ordered and we impose monotonicity as in \Cref{ex: fully ordered M}. In this case, it turns out that there is a closed-form solution for $\tilde\theta_{kk}^{min}$. Intuitively, to minimize the number of always-takers, we wish to have as many compliers as possible. This can be achieved by maximizing the amount of ``cascading'', as in panel (c) of \Cref{fig:type-shares-illustration}. \Cref{lem: thetakk closed form} in the appendix formalizes this intuition, and shows that
\begin{equation}
\tilde\theta_{kk}^{min} = \max\{ \underbrace{ P(M=m_k \mid D=1) }_{ \text{Point mass at $M=m_k$ when $D=1$} } - \hspace{0.5cm} (\underbrace{ P(M \geq m_k \mid D=1) - P(M \geq m_k \mid D=0) }_{ \text{Treatment effect on survival fn of $M$ at $m_k$} }) \hspace{0.1cm} , \hspace{0.1cm} 0 \}. \label{eqn: thetakk min}
\end{equation}
\noindent Moreover, there exists $\tilde\theta \in \Theta_I$ such that $\tilde\theta_{kk} = \tilde\theta_{kk}^{min}$ simultaneously for all $k$. Thus, when $M$ is fully-ordered and we impose monotonicity, one need not use a linear program to lower bound $\nu_{k}$ or test the sharp null; one can simply plug in the value of $\tilde\theta_{kk}^{min}$ to the lower bounds and testable implications given above.
\end{rem}

\begin{rem}[Identifying power]  \label{rmk:when-tests-have-bite}
\Copy{idpower}{The testable implications we have derived for the sharp null are based on the fact that under the sharp null, there is no effect of the treatment on $k$-always-takers (i.e. $\nu_k=0$). If the data is compatible with there being no always-takers ($\tilde\theta_{kk}^{min}=0$ for all $k$), then our lower bounds on the fraction of always-takers affected by the treatment are trivially zero and there is no testable content of the sharp null of full mediation. Intuitively, we therefore have limited testable content when we are ``local'' to there being no always-takers, i.e. when $\tilde\theta_{kk}^{min} \approx 0$ for all $k$.} The expression for $\tilde\theta_{kk}^{min}$ in \eqref{eqn: thetakk min} is thus informative about when the testable implications will have bite. In particular, it shows that $\tilde\theta_{kk}^{min}$ will tend to be large when there is substantial point mass at $M=m_k$ in the treated group, and when the treatment effect on the survival function of $M$ is small at $M=m_k$. Thus, while our testable implications are valid for any $M$ with a finite number of support points, there will tend to be more identifying power when there is substantial point mass for at least some values of $M$. We therefore expect the testable implications to be most useful in settings where $M$ takes only a moderate number of values. It is also worth noting that if the treatment effect of $D$ on $M$ is so large that everyone could be a complier, then $\tilde\theta_{kk}^{min}=0$ for all $k$ and so there is no testable content. It follows that our testable implications will tend to have more bite when the treatment effect of $D$ on $M$ is small; this is intuitive, since the data suggest that $M$ is not the only mechanism if the treatment effect of $D$ on $M$ is small relative to the effect of $D$ on $Y$.
\end{rem}

\begin{rem}[Binning values of $M$] \label{rmk: binning M}
 In light of the previous remark, in settings where $M$ is continuous or discrete with many values, it may be tempting to discretize the original $M$ into a small number of bins before applying our results. Let $M^{disc}$ be a discretization of $M$ into $K$ bins (indexed by $k=0, \dots, K-1$). If we compute the lower bounds for $\nu_k$ given in \Cref{prop: testable implications for tv} using $M^{disc}$ as the mediator, we obtain valid lower bounds on $P(Y(1,M(1)) \neq Y(0,M(0)) \mid M^{disc}(1) = M^{disc}(0) = k)$, i.e. the fraction of people whose outcome depends on treatment status despite having $M$ in the $k$th bin under both treatments. Our tests for the sharp null thus remain valid if we assume that changes of $M$ within a bin do not affect the outcome, i.e. $Y(d,m) = Y(d,m')$ for all $m$ and $m'$ corresponding to $M^{disc} = k$. This is a strong assumption if taken literally. However, one might reasonably expect that a small change in $M$ should not affect the outcome for most people. This could be captured by the assumption that a change of $M$ within a bin affects no more than $\nu_{max}$ fraction of people, so that $P(Y(d,m) \neq Y(d,m') \mid M^{disc}(1) = M^{disc}(0) = k) \leq \nu_{max}$ for all $m,m'$ in the same bin. In this case, if the sharp null is satisfied, the lower bounds on $\nu_k$ given in \Cref{prop: testable implications for tv} using $M^{disc}$ as the mediator should all be below $\nu_{max}$. That is, there should exist $\tilde\theta \in \Theta_I$ such that $\tilde\theta_{kk} \nu_{max} \geq (\sup_A \Delta_k(A) - \sum_{l: l \neq k} \tilde\theta_{lk})_+$ for all $k$. 
\end{rem}

\begin{rem}[Functions of the $\nu_k$]
We may sometimes be interested in aggregations of the $\nu_k$ across $k$. For example, the total fraction of always-takers whose outcome is affected by treatment, pooling across $k$, is given by $$\bar{\nu} :=P(Y(1,M(1)) \neq Y(0,M(0)) \mid M(1) = M(0)) = \frac{ \sum_k \theta_{kk} \nu_k}{ \sum_k \theta_{kk} } .$$ To compute a lower bound on this quantity, we must find $\tilde\theta$ and $\nu$ to minimize $\frac{ \sum_k \tilde\theta_{kk} \nu_k}{ \sum_k \tilde\theta_{kk} } $ subject to the constraints that \eqref{eqn: testable implication tv - feasible} holds and $\tilde\theta \in \Theta_I$. If we reparameterize the problem in terms of $\tilde\theta$ and $\tilde{\nu}_k := \tilde\theta_{kk} \nu_k$, then both the numerator and denominator of the objective are linear in the parameters, and the constraints are also linear in the parameters if $R$ is a polyhedron. Thus, the problem of minimizing $\frac{ \sum_k \theta_{kk} \nu_k}{ \sum_k \theta_{kk} } $ over the identified set is a linear-fractional program, which can be recast as a simple linear program via the \citet{charnes_programming_1962} transformation. It is thus simple to solve for lower bounds on the total fraction of always-takers affected by treatment, pooling across $k$.
\end{rem}

\begin{rem}[Connections to IV testing] \label{rmk: connection to sun} Since testing the sharp null of full mediation is analogous to testing instrument validity---with $M$ playing the role of the endogenous variable and $D$ the instrument---\Cref{cor: testable implications sharp null} immediately implies sharp testable implications for instrument validity in settings with a binary instrument and multi-valued $M$.\footnote{Specifically, our results are relevant for testing instrument validity when one assumes the full randomization assumption that the instrument is independent of both potential outcomes and treatments. The implications we derive may not be valid under the weaker notion of independence considered in \citet{kedagni_generalized_2020}, which imposes only that the instrument is independent of potential outcomes but not potential treatments.}\fnsep\footnote{The case where $M$ is multi-dimensional does not have an obvious parallel in the literature on testing instrument validity, since this would correspond to an IV setting with a single instrument but multiple endogenous variables.} The sharp testable restrictions derived here thus may be of independent interest for the problem of testing instrument validity. \Copy{sun-discussion}{\citet{sun_instrument_2023} derived non-sharp testable implications of instrument validity in the setting where $M$ is multi-valued but fully-ordered and one imposes monotonicity. His testable restrictions involve only the observable distributions with the minimum and maximum value of $M$. By contrast, \Cref{cor: testable implications sharp null} shows that there are in fact testable restrictions coming from all possible values of $M$, and adding these additional restrictions makes the testable implications sharp. As an illustrative example, suppose that $M \in \{0,1,2\}$ and that we impose monotonicity. Suppose, further, that the treatment $D$ has no impact on the distribution of $M$. Intuitively, under the sharp null we should then expect the distribution of $(Y,M)$ to be independent of $D$. Indeed, the sharp testable implication given in \Cref{cor: testable implications sharp null} corresponds to the restriction that $P(Y \in A,M=k \mid D=1) = P(Y \in A, M=k \mid D=0)$ for all $A$ and $k \in \{0,1,2\}$, which is equivalent to $(Y,M) \indep D$. By contrast, Sun's implications only imply this equality for $k \in \{0,2\}$, which is weaker than full independence. Additionally, while \citet{sun_instrument_2023}'s results apply under a monotonicity assumption, our results also imply testable implications under relaxations of monotonicity via a suitable choice of $R$, as described in Examples \ref{ex: fully ordered M}-\ref{ex: no restrictions} above.}
\end{rem}

\begin{rem}[Experiments with missing outcomes]
\citet{li_randomization_2025} consider the setting where we have a randomized experiment that generates an outcome $Y^*(D)$. However, $Y^*$ may be missing not-at-random: the observed outcome is $Y = M(D) \cdot Y^*(D)$ where $M(D)$ is an observed indicator for whether the outcome is missing. \citet{li_randomization_2025} are interested in the sharp null that $Y^*(1) = Y^*(0)$ (a.s.). Observe that under this sharp null, there is no effect of $D$ on $Y$ for ``always-takers'' for whom $M(1) = M(0) =k$ for $k=0,1$. Hence, our tests of the sharp null of full mediation can be used directly to test the sharp null of no treatment effect with missing outcomes in \citet{li_randomization_2025}. 
\end{rem}

\begin{rem}[Mis-measured mediator] \label{rmk: mismeasured m}
Our analysis so far has assumed that the mediator of interest $M$ is observed. In some settings, however, we may only observe a noisy proxy $\tilde{M}$ for $M$. For example, $M$ could be actual employment and $\tilde{M}$ reported employment on a survey. Assume that $\tilde{M} \indep (Y,D) \mid M$, so that the measurement error is independent of the other variables in the model given the true measurement $M$. Suppose, further, that the researcher knows the distribution of measurement error, $\tilde{M} \mid M$---for example, the researcher may have access to an auxiliary dataset that contains both survey responses and administrative measures of employment (if not, one could conduct sensitivity analyses to conjectured measurement error distributions). For simplicity, suppose that $Y$ is discrete and $M$ and $\tilde{M}$ have the same support. We then have that $P(Y=y,\tilde{M} = \tilde{m} , D=d) = \sum_{m} P(Y=y,M=m , D=d) P(\tilde{M}=\tilde{m} \mid M=m)$. We can write this equality as $\tilde{p} = L p$, where $\tilde{p}$ is the vector collecting probabilities of the form $P(Y=y,\tilde{M} = \tilde{m} , D=d)$ for different $\tilde{m}$; $p$ is analogously the vector collecting the $P(Y=y,M=m , D=d)$ for different $m$; and $L$ is the $K \times K$ matrix collecting probabilities of the form $P(\tilde{M}=\tilde{m} \mid M=m)$. Provided that $L$ is full-rank, it follows that $p = L^{-1} \tilde{p}$, and thus the distribution of $(Y,M, D)$ is identified. Our results described above, which assume that $(Y,M,D)$ are directly observed, can thus be applied using the implied distributions of $(Y,M,D)$ under this measurement error structure. 
\end{rem}

\section{Inference} \label{sec: inference}
The previous section derived testable implications of the sharp null of full mediation, as well as measures of the extent to which it is violated, which involved the distribution of the observable data $(Y,M,D) \sim P_{obs}$. We now derive methods for inference on the sharp null given a sample of $N$ $iid$ observations (or clusters) drawn from $P_{obs}$, $(Y_i,M_i,D_i)_{i=1}^{N}$. For simplicity of notation, we focus on testing the sharp null, although a simple adaptation of the described approach can be used to test null hypotheses of the form $H_0: \nu_k \leq \nu^{ub}_k \hspace{.1cm} \forall k$ for any $\nu^{ub}_k$ (with the sharp null the special case with $\nu^{ub}_k = 0$ for all $k$.) 

We first comment on the non-standard nature of the inference problem. Recall that the testable implications of the sharp null are equivalent to whether the linear program \eqref{eqn: testable implications - lp} has a weakly negative solution. However, functions of the observable data enter the constraints of the linear program, and it is well-known that the solution to a linear program can be non-differentiable in the constraints (\citealp{shapiro_stochprog_1991}). Second, the function of the observable data in the constraints, $\sup_A \Delta_k(A)$, is itself potentially non-differentiable in the underlying data-generating process. If the outcome $Y$ is discrete, for example, then $\sup_A \Delta_k(A) = \sum_{y} (f_{Y,M=m_k \mid D=1}(y) - f_{Y,M=m_k \mid D=0}(y))_+ $, where $(x)_+ = \max\{x,0\}$, which is clearly non-differentiable in the partial probability mass functions $f_{Y,M=m_k \mid D=d}(y) := P(Y=y, M=m_k \mid D=d)$ if $f_{Y,M=m_k \mid D=1}(y) = f_{Y,M=m_k \mid D=0}(y)$ for any $y$. Since bootstrap methods are generally invalid when the target parameter is non-differentiable in the underlying data-generating process \citep{fang_inference_2019}, we cannot simply bootstrap the solution to \eqref{eqn: testable implications - lp}. 

We now show that methods from the moment inequality literature can be used to circumvent these issues. We focus on the case where the distribution of $Y$ is discrete, with support points $y_1,...,y_Q$. As we discuss in \Cref{rmk: discretizing y} below, if $Y$ is continuous, then the tests we derive remain valid if one uses a discretization of $Y$, although at the potential loss of sharpness. We also focus on the case where $R$ takes the polyhedral form $R = \{ \theta \in \Delta : B \theta \leq c\}$. To see the connection with moment inequalities, observe that with discrete $Y$, we have that $$\sup_A \Delta_k(A) = \sum_{q=1}^Q \left(P(Y=y_q, M=m_k \mid D=1) - P(Y=y_q, M=m_k \mid D=0) \right)_+$$

\noindent where again $(x)_+ = \max \left\{x,0 \right\}$. It follows that the inequality $\sup_A \Delta_k(A) \leq P(M=m_k \mid D=1) - \tilde\theta_{kk} $
\noindent holds if and only if there exist $\delta_{k1},...,\delta_{kQ}$ such that 
\begin{align}
& \sum_{q=1}^Q \delta_{kq} \leq P(M=m_k \mid D=1) - \tilde\theta_{kk} \label{eqn: tv constraint 1} \\
& \delta_{kq} \geq P(Y=y_q, M=m_k \mid D=1) - P(Y=y_q, M=m_k \mid D=0) \text{ for } q=1,...,Q \label{eqn: tv constraint 2} \\
& \delta_{kq} \geq 0 \text{ for } q=1,...,Q .\label{eqn: tv constraint 3}
\end{align}

\noindent Hence, the testable implications of the sharp null derived in \Cref{cor: testable implications sharp null} are equivalent to the statement that there exists some $\tilde\theta \in \Theta_I$ and $\delta$ such that \eqref{eqn: tv constraint 1}-\eqref{eqn: tv constraint 3} hold for all $k=0,...,K-1$. 

Observe, further, that $\delta,\tilde\theta$ and the observable probabilities enter the inequalities \eqref{eqn: tv constraint 1}-\eqref{eqn: tv constraint 3} linearly, and the same is true for the constraints that determine $\Theta_I$. Letting $\omega = (\tilde\theta',\delta')'$, it follows that we can write the testable implications of the sharp null as
\begin{equation}
H_0: \exists \, \omega \text{ s.t. } C_1 \omega - C_2p  \geq 0, \label{eqn: h0 w nuisance} 
\end{equation}
\noindent where $C_1,C_2$ are known matrices (not depending on the data) and $p$ is a vector that collects probabilities of the forms $P(Y=y_q, M=m_k \mid D=d)$ and $P(M=m_k \mid D=d)$. A recent literature on moment inequalities has considered testing hypotheses of the above form---in which the nuisance parameter $\omega$ enters linearly and with known coefficients $C_1$---given estimates $\hat{p}$ such that $\sqrt{N}(\hat{p}-p) \to N(0,\Sigma)$ \citep{andrews_inference_2023, cox_simple_2022, fang_inference_2023, cho_simple_2024}. Under mild conditions, the central limit theorem implies that the vector of conditional sample means $\hat{p}$ is asymptotically normal, and thus existing methods from the aforementioned papers can thus be used directly to test the sharp null of full mediation.

\begin{rem}[Discretizing continuous outcomes] \label{rmk: discretizing y}
Suppose that the outcome $Y$ is continuously distributed. Let $I_1,...,I_Q$ be disjoint intervals that partition the outcome space, and let $Y^{disc}$ be the discretization of $Y$ that equals $j$ when $Y \in I_j$. Let $\Delta_k^{disc}(A)$ be the analog to $\Delta_k(A)$ using $Y^{disc}$ instead of $Y$.  Observe that
\begin{align*}
\sup_{A} \Delta_k^{disc}(A) &= \sup_{A} \hspace{.1cm} P(Y^{disc} \in A, M=m_k \mid D=1) - P(Y^{disc} \in A, M=m_k \mid D=0) \\
&= \sup_{A \in \mathcal{A}_{disc}} P(Y \in A, M=m_k \mid D=1) - P(Y\in A, M=m_k \mid D=0)  = \sup_{A \in \mathcal{A}_{disc}} \Delta_k(A)
\end{align*}
\noindent where $\mathcal{A}_{disc}$ is the $\sigma$-algebra generated by $I_1,...,I_Q$. Since $\mathcal{A}_{disc}$ is a subset of the Borel $\sigma$-algebra, it follows that $\sup_A \Delta_k^{disc} (A) = \sup_{A \in \mathcal{A}_{disc}} \Delta_k(A) \leq \sup_A \Delta_k(A)$. Hence, the testable implications of the sharp null for $Y$ imply the testable implications of the sharp null for any discretization of $Y$. One can thus obtain valid inference on the sharp null by discretizing the outcome and then using the approach described above with $Y^{disc}$. Of course, to retain approximate sharpness of the testable implications, one would like to choose a discretization fine enough such that $\sup_A \Delta^{disc}_k(A) \approx \sup_A \Delta_k(A)$. Observe that with a continuous outcome, $\sup_A \Delta^{disc}_k(A) = \sup_A \Delta_k(A)$ if the sign of $f_{Y,M=m_k \mid D=1}(y) - f_{Y,M=m_k \mid D=0}(y)$ is constant at all $y$ within the same interval $I_j$. To obtain approximate sharpness of the testable implications, one would thus like to choose a discretization such that there is a cut-point close to any point where the partial densities cross. A practical tradeoff arises, however, because the validity of the methods described above to test moment inequalities relies on a central limit theorem for the sample probabilities $\hat{p}$, and hence requires that the number of observations per cell (i.e. $(Y^{disc},M,D)$ combination) not be too small. There is thus a trade-off whereby we expect that choosing a smaller number of bins is beneficial for size control but may lead to less sharp testable implications. Matters are further complicated by the fact that the finite-sample power of moment inequality methods may be non-monotonic in the number of bins (as we find in our Monte Carlo simulations below). In Appendix \ref{sec: disc-cont-outcome}, we discuss how the choice of bins is closely related to the choice of instrument functions in the literature on conditional moment inequalities \citep[e.g.][]{andrews_cmi_2013}, which is known to be a challenging problem. Although a formal treatment of the optimal bin choice is beyond the scope of this paper, we provide some practical heuristics following our Monte Carlo simulations in \Cref{sec: monte carlo}. We note, further, that having a modest number of bins may lead to a more interpretable parameter $\nu_k$. For example, if one uses a discretization using 5 bins based on the quintiles of the outcome, then $\nu_k$ corresponds to the fraction of $k$-always-takers whose outcome changes quintile when treated; this may be easier to interpret in some settings than the fraction of always-takers whose outcome is affected at all. 

\end{rem}

\subsection{Monte Carlo} \label{sec: monte carlo}
To evaluate the methods for inference described above, we conduct Monte Carlo simulations calibrated to our applications to \citet{bursztyn_misperceived_2020} and \citet{baranov_maternal_2020} in \Cref{sec: applications} below. For simplicity, we focus on testing the sharp null under a monotonicity assumption.

\paragraph{Treatment, outcome, and mediator.} The treatment, outcome, and mediator in our simulations match those in our empirical applications. For \citet{bursztyn_misperceived_2020}, the treatment is receiving information about other men's beliefs, the outcome is a binary indicator for applying for jobs outside of the home, and the mediator is a binary indicator for job-search service sign-up. For \citet{baranov_maternal_2020}, the treatment is cognitive behavioral therapy and the outcome is an index of financial empowerment. We consider two mediators, a binary indicator for the presence of a grandmother in the household, and a relationship-quality score, which is a score on a 1-5 scale.

\paragraph{Sample sizes.} The sample used for our main analysis of \citet{bursztyn_misperceived_2020} contains $284$
\unskip people, with treatment assignment randomized at the individual level (approximately half ($139$
\unskip) were treated). For the simulations calibrated to
\citet{bursztyn_misperceived_2020}, we draw $284$
\unskip $iid$ observations to match the
original sample size. In
\citet{baranov_maternal_2020}, treatment was assigned at the level of a cluster
(i.e. at the Union Council level), with a total of 40 clusters (20 treated, 20
control), and a total sample size of approximately 600 individuals ($568$
\unskip or $585$
\unskip depending on the choice of $M$). For simulations calibrated to \citet{baranov_maternal_2020}, we therefore draw 20
independent clusters from each treatment group. Given the small
number of clusters, we expect this to be a relatively challenging setting for
inference. To evaluate the impact of having a small number of clusters, we also
consider alternative simulation designs where we sample 40 or 100 clusters of
each treatment type, with a total of 80 and 200 clusters for each design.

\paragraph{Description of DGP.} In all of our simulations, we sample the distribution of $(Y,M)$ for control units (or clusters) from the empirical distribution of control units (or clusters) in our applications (i.e. from $(Y,M) \mid D=0$). For treated units in our simulations, we draw with probability $t$ from the empirical distribution of $(Y,M)$ for treated units, and with probability $1-t$ from the empirical distribution for control units, where $t \in \{0,0.5,1\}$ is a simulation parameter. Thus, when $t=1$, we are sampling both treated and control units in the simulation from the empirical distribution in the data, under which the sharp null is violated. This allows us to assess the power of the various tests. When $t=0$, on the other hand, the distribution of $(Y,M)$ for both treated and control units in the simulation is drawn from the empirical distribution for control units in the original data. This ensures that the testable implications of the sharp null and monotonicity are satisfied, which allows us to evaluate size control. (In fact, the design ensures that all of the implied moment inequalities hold with equality, which is generally a challenging setting for size control for moment inequality methods.) When $t=0.5$, the distribution of $(Y,M)$ for treated units is a mixture of the empirical distribution for treated and control units in the original data. Thus, the sharp null is violated, but the violation is smaller than under the case when $t=1$. Comparing across the cases $t=0.5$ and $t=1$ thereby allows us to evaluate how power changes with the size of the violation of the null.

\paragraph{Methods used.} To implement tests based on moment inequalities as
described above, we consider the hybrid test proposed by
\citet[][henceforth ARP]{andrews_inference_2023}, the conditional conditional
chi-squared test proposed by \citet[][henceforth CS]{cox_simple_2022},\footnote{More precisely, CS propose a conditional chi-squared test and a ``refined'' version of this test. The refinement addresses the fact that the baseline test can be conservative. Since the refinement is computationally costly with many moments, and only matters when one moment is binding, we only implement the refinement in DGPs with a binary outcome, for which there are fewer moments.\label{fn:cs-details}} and the
test proposed by \citet[][henceforth FSST]{fang_inference_2023}.\footnote{When $M$ is binary, we implement the formulation of the moment inequalities derived in \eqref{eq:bin_sharp_imp} without nuisance parameters. For non-binary $M$, we use the formulation in \eqref{eqn: h0 w nuisance}.} For comparison
to existing methods in the case where $M$ is binary, we consider the test for
instrument validity proposed by \citet[][henceforth
K]{kitagawa_test_2015}.\footnote{For the DGPs based on
  \cite{baranov_maternal_2020}, we use a modified version of
  \cite{kitagawa_test_2015} to account for clustering.} In the simulations calibrated to \citet{bursztyn_misperceived_2020}, the outcome is binary, and thus no discretization of the outcome is needed. For the simulations calibrated to \citet{baranov_maternal_2020}, where the outcome takes many values, for the moment inequality methods we consider a discretization of the outcome based on 5 bins in our main specification (see \Cref{rmk: discretizing y}). We also consider alternative specifications using 2 and 10 bins. Since the K test does not require a discrete outcome, we use the original continuous outcome when implementing the K test. Implementation of the FSST test requires specifying the moment-selection tuning parameter $\lambda$. We consider the two choices recommended by FSST in their Remark 4.2, one of which is data-driven and the other is not. We refer to the resulting tests as FSSTdd and FSSTndd (where `dd' denotes data-driven). For CS and ARP, we use analytic estimates of the variance of the moments, assuming the data are drawn $iid$ in the simulations calibrated to \citet{bursztyn_misperceived_2020}, or that clusters are drawn $iid$ in the simulations calibrated to \citet{baranov_maternal_2020}. Since the K and FSST tests require bootstrap replicates, we use a non-parametric bootstrap at either the individual or cluster level, as appropriate.\footnote{We have verified that ARP and CS return similar results if we use an analogous bootstrap estimate of the variance rather than the analytic estimates.} All tests impose monotonicity as defined in \Cref{eqn: delta for monotonicity}.\footnote{\Copy{mon-violated}{As described in the empirical section below, for the multi-valued $M$ in \citet{baranov_maternal_2020}, the empirical distribution for $M \mid D$ is inconsistent with monotonicity (although the violation is not statistically significant). Our simulation design ensures that the data are consistent with monotonicity under the null DGP ($t=0$). However, the alternative DGPs ($t \in \{0.5,1\}$) are based on the empirical distribution and are therefore inconsistent with monotonicity. Hence, the reported power of tests imposing monotonicity under these alternatives corresponds to their power to jointly detect a violation of the sharp null and a relatively small violation of the monotonicity assumption. We found the ranking of power across methods was identical when we modified the tests to allow for the minimal relaxation of monotonicity consistent with the data, and therefore present the results imposing monotonicity for simplicity and consistency with the other specifications.}\label{fn:mon-violated}} All tests are implemented with nominal size of 5\%.

  \begin{table}[!htbp] \centering
  \caption{Simulation results for binary $M$}
  \label{tab: main_sim_binM}
\begin{threeparttable}
\begin{tabular}{@{\extracolsep{5pt}} ccccccc} 
\\[-1.8ex]\hline 
\hline \\[-1.8ex] 
\\[-2.0ex] \multicolumn{7}{@{} l}{Panel A: Bursztyn et al}
 \\
 \\[-1.5ex]
 & $\bar{\nu}$ LB & ARP & CS & K & FSSTdd & FSSTndd \\ 
\cline{3-7} \\[-2.0ex]
t=0 & $0$ & $0.038$ & $0.032$ & $0.030$ & $0.078$ & $0.070$ \\ 
t=0.5 & $0.036$ & $0.196$ & $0.190$ & $0.116$ & $0.214$ & $0.194$ \\ 
t=1 & $0.077$ & $0.626$ & $0.632$ & $0.386$ & $0.620$ & $0.584$ \\ 
\\[-1.83ex] 
 \hline \\[-1.83ex]
\\[-2.0ex] \multicolumn{7}{@{} l}{Panel B: Baranov et al, 40 clusters}
 \\
 \\[-1.5ex]
 & $\bar{\nu}$ LB & ARP & CS & K & FSSTdd & FSSTndd \\ 
\cline{3-7} \\[-2.0ex]
t=0 & $0$ & $0.056$ & $0.154$ & $0.050$ & $0.232$ & $0.212$ \\ 
t=0.5 & $0.134$ & $0.194$ & $0.206$ & $0.064$ & $0.314$ & $0.270$ \\ 
t=1 & $0.283$ & $0.570$ & $0.668$ & $0.422$ & $0.750$ & $0.680$ \\ 
\\[-1.83ex] 
 \hline \\[-1.83ex]
\\[-2.0ex] \multicolumn{7}{@{} l}{Panel C: Baranov et al, 80 clusters}
 \\
 \\[-1.5ex]
 & $\bar{\nu}$ LB & ARP & CS & K & FSSTdd & FSSTndd \\ 
\cline{3-7} \\[-2.0ex]
t=0 & $0$ & $0.044$ & $0.064$ & $0.040$ & $0.132$ & $0.112$ \\ 
t=0.5 & $0.134$ & $0.322$ & $0.340$ & $0.160$ & $0.410$ & $0.322$ \\ 
t=1 & $0.283$ & $0.836$ & $0.936$ & $0.846$ & $0.956$ & $0.936$ \\ 
\\[-1.83ex] 
 \hline \\[-1.83ex]
\\[-2.0ex] \multicolumn{7}{@{} l}{Panel D: Baranov et al, 200 clusters}
 \\
 \\[-1.5ex]
 & $\bar{\nu}$ LB & ARP & CS & K & FSSTdd & FSSTndd \\ 
\cline{3-7} \\[-2.0ex]
t=0 & $0$ & $0.044$ & $0.054$ & $0.030$ & $0.120$ & $0.090$ \\ 
t=0.5 & $0.134$ & $0.686$ & $0.776$ & $0.618$ & $0.776$ & $0.734$ \\ 
t=1 & $0.283$ & $0.998$ & $1$ & $1$ & $1$ & $1$ \\ 
\\[-2.0ex]
\hline \\[-1.8ex] 
\end{tabular} 
    
\begin{tablenotes}[flushleft]\footnotesize
\item\emph{Notes}: This table contains simulation results for the DGPs where we
  have a binary mediator. The first column shows the value of $t$, which
  determines the distance from the null, as described in the main text. The second column shows the lower-bound on the fraction of always-takers affected by treatment, $\bar{\nu}$. The remaining columns contain the rejection probabilities for each of the methods considered. Panel A shows the results for the DGP based on
  \cite{bursztyn_misperceived_2020} and Panels B-D show the results for the DGPs
  based on \cite{baranov_maternal_2020}, with the binary grandmother mediator, under
  different numbers of clusters. In Panels B-D, we use a discretization of the outcome into 5 bins for all tests except the K test. Rejection probabilities are computed over 500
  simulation draws, under a 5\% nominal significance level.
  \end{tablenotes}
\end{threeparttable}
\end{table}

\begin{table}[!htbp]
  \centering
  \caption{Simulation results for non-binary $M$}
  \label{tab: main_sim_nonbinM}
  \begin{threeparttable}
    \begin{tabular}{@{\extracolsep{5pt}} cccccc} 
\\[-1.8ex]\hline 
\hline \\[-1.8ex] 
\\[-2.0ex] \multicolumn{6}{@{} l}{Panel A: Baranov et al, 40 clusters}
 \\
 \\[-1.5ex]
 & $\bar{\nu}$ LB & ARP & CS & FSSTdd & FSSTndd \\ 
\cline{3-6} \\[-2.0ex]
t=0 & $0$ & $0.052$ & $0.088$ & $0.274$ & $0.178$ \\ 
t=0.5 & $0.119$ & $0.066$ & $0.228$ & $0.438$ & $0.374$ \\ 
t=1 & $0.255$ & $0.166$ & $0.754$ & $0.864$ & $0.828$ \\ 
\\[-1.83ex] 
 \hline \\[-1.83ex]
\\[-2.0ex] \multicolumn{6}{@{} l}{Panel B: Baranov et al, 80 clusters}
 \\
 \\[-1.5ex]
 & $\bar{\nu}$ LB & ARP & CS & FSSTdd & FSSTndd \\ 
\cline{3-6} \\[-2.0ex]
t=0 & $0$ & $0.066$ & $0.048$ & $0.188$ & $0.128$ \\ 
t=0.5 & $0.119$ & $0.066$ & $0.314$ & $0.582$ & $0.500$ \\ 
t=1 & $0.255$ & $0.164$ & $0.962$ & $0.994$ & $0.990$ \\ 
\\[-1.83ex] 
 \hline \\[-1.83ex]
\\[-2.0ex] \multicolumn{6}{@{} l}{Panel C: Baranov et al, 200 clusters}
 \\
 \\[-1.5ex]
 & $\bar{\nu}$ LB & ARP & CS & FSSTdd & FSSTndd \\ 
\cline{3-6} \\[-2.0ex]
t=0 & $0$ & $0.046$ & $0.026$ & $0.144$ & $0.108$ \\ 
t=0.5 & $0.119$ & $0.076$ & $0.542$ & $0.862$ & $0.824$ \\ 
t=1 & $0.255$ & $0.286$ & $1$ & $1$ & $1$ \\ 
\\[-2.0ex]
\hline \\[-1.8ex] 
\end{tabular} 

    \begin{tablenotes}[flushleft]\footnotesize
    \item\emph{Notes}: This table contains simulation results for the DGPs where
      we have a non-binary mediator. The first column shows the value of $t$,
      which determines the distance from the null, as described in the main
      text. The second column shows the lower-bound on the fraction of always-takers affected by treatment, $\bar{\nu}$. The remaining columns contain the rejection probabilities for
      each of the inference methods considered. Each panel contains results for
      the DGPs based on \cite{baranov_maternal_2020}, where the non-binary relationship-quality
      mediator is considered, for different numbers of clusters. All tests use a discretization of the outcome based on 5 bins. Rejection
      probabilities are computed over 500 simulation draws, under a 5\% nominal significance level.
    \end{tablenotes}
  \end{threeparttable}
\end{table}

\paragraph{Simulation results.} \Cref{tab: main_sim_binM} reports the results for simulations
designs where we have a binary mediator. This includes the DGP based on
\cite{bursztyn_misperceived_2020} (Panel A), and the DGPs that are based on
\cite{baranov_maternal_2020} where the considered mediator is the binary
indicator for the presence of a grandmother (Panels B-D). \Cref{tab: main_sim_nonbinM} shows results calibrated to \citet{baranov_maternal_2020} using the non-binary relationship quality variable as the mediator. Both tables show the rejection probabilities for each of the methods described above under different simulation designs. To quantify the magnitude of the violations of the sharp null, the table also reports the lower-bound on the fraction of always-takers affected ($\bar{\nu}$).\footnote{For the simulations calibrated to \citet{baranov_maternal_2020} with multi-valued $M$, we compute the lower bound on $\bar{\nu}$ in the same way as described in \Cref{fn: monotonicity violated} in the application section below, which deals with the fact that the empirical distribution shows a small (but statistically insignificant) violation of monotonicity.}

We first evaluate size control. Recall that DGPs with $t=0$ impose the sharp null of full mediation. Across nearly all simulation designs, we find that the ARP, CS, and K tests have close to nominal size, with rejection probabilities no larger than 9\% for a 5\% test. The one notable exception is the simulations in Panel B of \Cref{tab: main_sim_binM}, where there are only 40 independent clusters, in which case CS is somewhat over-sized, with a null rejection probability of 0.15. Doubling the number of clusters to 80 (Panel C) restores approximate size control, however. We find that the FSST tests often have reasonable size control for settings with a large number of independent observations or clusters, but can be substantially over-sized in settings with a small or moderate number of clusters using the two default choices of tuning parameters, particularly with multi-valued $M$ (e.g. rejection probabilities of 0.274 and 0.178 in \Cref{tab: main_sim_nonbinM}, Panel A).

We next evaluate power, focusing on the simulations with $t=0.5$ and $t=1$ under which the null is violated. Across all of the simulation designs, the CS test has power similar to or greater than that of ARP. The differences can be substantial in some cases, particularly with multi-valued $M$ (e.g. power of $0.96$ vs $0.16$ in Panel B of \Cref{tab: main_sim_nonbinM}). Likewise, the power of the FSST tests is similar to or exceeds that of the CS test across nearly all simulation designs, although this comparison must be taken with some caution in cases where the FSST test appears to be over-sized. Finally, we note that in all of the simulations with binary $M$ (\Cref{tab: main_sim_binM}), the power of the three moment inequality tests (ARP, CS, FSST) is either similar to or exceeds that of the K test. This is the case both when the outcome is binary (Panel A), and when the outcome is continuous (Panels B-D). Recall that when the outcome is continuous, the moment inequality tests use a discretization of the outcome to 5 bins, whereas the K test does not use a discretization. The favorable power comparisons in Panels B-D thus suggest that discretization does not come at a large loss of power in this simulation design, although of course this conclusion may be specific to the particular DGP studied here.

\paragraph{Choice of bins.} In \Cref{tab: app_sim_baranov_binM,tab: app_sim_baranov_nonbinM}, we present results for simulations calibrated to \citet{baranov_maternal_2020} using a discretization with 2 or 10 bins, rather than the 5 considered above. The comparisons of size control and power across the methods are similar to the results reported above. However, increasing the number of bins from 5 to 10 exacerbates the size control issues seen with CS in the simulations based on \citet{baranov_maternal_2020} with only 40 clusters and binary $M$, while decreasing the number of bins to 2 improves size control. This is intuitive, since the number of moments used increases with the bin size, and thus we expect the quality of the central limit theorem approximation to be worse with more bins. In \Cref{tab: cell_count_nonbinM}, we report the median number of independent observations (unique clusters) per $(Y^{disc},M,D)$ cell in each of the simulation designs. We find that CS exhibits close to nominal coverage in all specifications with 15 or more independent observations per cell, but exhibits moderate size distortions in several (but not all) of the specifications with fewer than 15 observations per cell. In terms of power, we do not find an obvious pattern across bin sizes, with power increasing in the number of bins for some tests/DGPs and decreasing for others. This reflects the fact that although the testable implications become sharper the more bins are used (see \Cref{rmk: discretizing y}), the finite-sample power of moment inequality methods often decreases when increasing the number of moments. Based on our simulations, we heuristically recommend that researchers should try to have at least 15 independent observations per cell, acknowledging that there is a potential tradeoff between size control and power. This heuristic roughly aligns with that in \citet{andrews_cmi_2013}, who recommend having 10-20 observations per cell in settings with conditional moment inequalities.

\paragraph{Choice of test.} Based on our simulations, CS strikes us a reasonable default choice for most empirical settings, given that it has approximate size control across most of our simulation designs and favorable power relative to ARP. However, ARP performs somewhat better in terms of size control in settings with a small number of clusters, and thus may be an attractive alternative for researchers concerned about size control in such settings, albeit at the loss of some power (particularly with multi-valued $M$). Likewise, FSST may offer power improvements relative to CS in settings with a large number of independent observations, so that size control is less of a concern. In our applications below, we report results for CS in the main text; analogous results for ARP and FSST are given in \Cref{tbl:pvals}.

\section{Extension to Non-experimental Settings}
\label{sec: nonexperimental}

Our results so far have relied on the assumption that the treatment is as good
as randomly assigned (\Cref{asm: independence}). The role of randomization was
simply to identify the distribution of potential outcomes and potential
mediators under each treatment: randomization ensures that the observable
distribution $(Y,M) \mid D=d$ corresponds to the distribution of
$(Y^{\text{tot}}(d), M(d))$, where $Y^{\text{tot}}(\cdot) := Y(\cdot, M(\cdot))$. In this section, we show that analogous results go through if the distributions of $(Y^{\text{tot}}(d),M(d))$ are identified through some other strategy. One simply substitutes the expressions involving $(Y,M) \mid D=d$ in our earlier results with the formulas for $(Y^{\text{tot}}(d), M(d))$ under the alternative identifying assumptions. We first provide a general result extending our results to the case where $(Y^{\text{tot}}(d), M(d))$ is identified, then discuss how this applies to the common settings of instrumental variables, conditional unconfoundedness, and difference-in-differences.

To be more precise, define 
$$\Delta_k^*(A) := P(Y^{\text{tot}}(1) \in A, M(1) =m_{k}) - P(Y^{\text{tot}}(0) \in A, M(0) =m_{k})$$
\noindent to be the treatment effect on the compound outcome $1\{Y \in A, M=m_{k}\}$. Note that under randomization, $\Delta_k^*(A) = \Delta_k(A)$. Likewise, define 
\begin{equation}
  \Theta_I^* := \{ \tilde\theta^* \in R \,:\, \text{ for all } k=0,...,K-1, \hspace{.2cm} \textstyle\sum_{l} \tilde\theta^*_{lk} = P(M(1)=m_{k}), \sum_{l} \tilde\theta^*_{kl} = P(M(0)=m_{k}) \} \label{eqn: defn of ThetaI*}  
\end{equation}
\noindent and observe that $\Theta_I^* = \Theta_I$ under randomization. We then have the following result, analogous to \Cref{prop: testable implications for tv}, which gives lower bounds on the $\nu_k$ in terms of probabilities involving $(Y^{\text{tot}}(d),M(d))$. 

\begin{prop} \label{prop: lower bounds on tv - po version}
\leavevmode
    \begin{enumerate}
    \item 
    Suppose \Cref{asm: restricted theta} holds. Then the true shares $\theta$ satisfy
    \begin{align}
    \theta_{kk} \nu_k \geq \left( \sup_A \Delta_k^*(A)  - \sum_{l: l \neq k} \theta_{lk} \right)_+  \label{eqn: lower bound tv - true theta*}
    \end{align}
    
    \noindent for $k=0,...,K-1$. Since $\theta \in \Theta_I^*$, it follows that there exists some $\tilde\theta^* \in \Theta_I^*$ such that
    \begin{equation}
        \tilde\theta^*_{kk} \nu_k \geq \left( \sup_A \Delta_k^*(A)  - \sum_{l: l \neq k} \tilde\theta^*_{lk} \right)_+ =: \eta_k  \label{eqn: lower bound tv tildetheta*}
    \end{equation}
    \noindent for all $k = 0,...,K-1$, and hence $\nu_k \geq \inf \{\nu_k \geq 0 \,:\, \exists \tilde\theta^* \in \Theta^* \text{ s.t. } \eqref{eqn: lower bound tv tildetheta*} \text{ holds}\}.$ 

    \item
    The bound in \eqref{eqn: lower bound tv tildetheta*} is the sharp bound that only uses information from the marginals $(Y^{\text{tot}}(d),M(d))$ for $d=0,1$: if $\tilde\theta^* \in \Theta_I^*$, then there exists a distribution $P^*$ for $(Y(\cdot,\cdot),M(\cdot))$ that is consistent with the marginals of $(Y^{\text{tot}}(d),M(d))$ for $d=0,1$ such that $P^*(G=lk) = \tilde\theta_{lk}^*$ for all $l,k$, and
    \begin{equation}
      \tilde\theta_{kk}^* \cdot P^*(Y(1,m_{k}) \neq Y(0,m_{k}) \mid G = kk) = \eta_k \label{eqn: tv equality - po version}  
    \end{equation}
    \noindent for all $k$. Further, $P^*(Y(1,m) \neq Y(0,m) \mid G = lk) = 0$ if either $l \neq k$ or $m \not\in \{m_{l},m_{k}\}$. 
    \end{enumerate}
\end{prop}

We likewise obtain the following corollary regarding the sharp null of full mediation, analogous to \Cref{cor: testable implications sharp null}. 

\begin{cor} \label{cor: testing sharp null - po version}
\leavevmode
\begin{enumerate}
    \item 
    Suppose \Cref{asm: restricted theta} holds. If the sharp null of full mediation is satisfied, then there exists $\tilde\theta^* \in \Theta_I^*$ such that
    \begin{align}
     \sup_A \Delta_k^*(A)  \leq \sum_{l: l \neq k} \tilde\theta^*_{lk}  \label{eqn: testable implications sharp null - po version}
    \end{align}
   for $k=0,...,K-1$. 

   \item
   The testable implication in \eqref{eqn: testable implications sharp null - po version} is the sharp implication using only the marginals of $(Y^{\text{tot}}(d),M(d))$ for $d=0,1$: if there exists a $\tilde\theta^* \in \Theta_I^*$ such that \eqref{eqn: testable implications sharp null - po version} holds for all $k$, then there exists a distribution $P^*$ for $(Y(\cdot,\cdot),M(\cdot))$ consistent with the marginals for $(Y^{\text{tot}}(d),M(d))$ and the restriction that $\theta \in R$ such that the sharp null of full mediation holds.    
\end{enumerate}

\end{cor}

\Cref{prop: lower bounds on tv - po version} and \Cref{cor: testing sharp null - po version} show that we can bound the fraction of always-takers affected by treatment and test the sharp null so long as the marginal distributions of $(Y^{\text{tot}}(d),M(d))$ are identified. We now outline several settings where these distributions are identified, possibly for some sub-population of interest (e.g. for compliers with respect to an instrument).

\paragraph{Instrumental variables.} Suppose rather than $D$ being randomly assigned, we have a valid binary instrument $Z \in \{0,1\}$ for $D$. For example, in an experiment with imperfect compliance, $Z$ could be the randomized treatment assignment, and $D$ the realized treatment take-up. Suppose that $Z$ satisfies the standard instrument-monotonicity, relevance, exclusion, and independence assumptions \citep{imbens_identification_1994}. To be precise, we assume that $D=D(Z)$, where $D(1) \geq D(0)$ (a.s.) and $P(D(1) > D(0)) > 0$, and that $(Y,M) = (Y(D(Z),M(D(Z))), M(D(Z)))$ for $Z \indep (Y(\cdot,\cdot),M(\cdot),D(\cdot))$. These assumptions allow us to identify the LATE of $D$ on $Y$, i.e. the treatment effect of $D$ on $Y$ for instrument-compliers.\footnote{We use the phrase ``instrument-compliers'' to refer to compliers with respect to the instrument, i.e. with $D(1) > D(0)$, in contrast to the use of ``compliers'' elsewhere in the paper, which refers to individuals with $M(1) \neq M(0)$.} They likewise allow us to identify the LATE of $D$ on $M$. We might then be interested in the extent to which the effect of $D$ on $Y$ for instrument-compliers operates through $M$.  To apply the results in \Cref{prop: lower bounds on tv - po version} and \Cref{cor: testing sharp null - po version}, we need to identify the distributions of $(Y^{\text{tot}}(d),M(d))$ for instrument-compliers. It is well-known, however, that the distributions of potential outcomes for instrument-compliers are non-parametrically identified (see, e.g., \citealp{abadie_semiparametric_2003}). In particular, if we define $C^z = 1\{ D(1) > D(0)\}$ to be an indicator for being an instrument-complier, then $P( Y^{\text{tot}}(1) \in A, M(1) =m_{k} \mid C^z =1 )$ is identified as  \begin{equation}
    \frac{ E[D \cdot 1\{Y \in A, M=m_{k}\} \mid Z = 1 ] - E[D \cdot 1\{Y \in A, M=m_{k}\} \mid Z = 0 ]   }{ E[D \mid Z=1] - E[D \mid Z=0]} . \label{eqn: id in iv}
\end{equation}
In words, the probability that $Y^{\text{tot}}(1) \in A$ and $M(1) =m_{k}$ for instrument-compliers corresponds to the population IV estimand using the compound outcome $D \cdot 1\{Y \in A, M=m_{k} \}$. We can likewise identify $P(Y^{\text{tot}}(0) \in A, M(0) =m_{k})$ using the population IV estimand with the compound outcome $-(1-D) \cdot 1\{Y \in A, M=m_{k} \}$.\footnote{Note that the instrument exclusion restriction implies that $Z$ affects $Y$ only through $D$, and the sharp null implies that $D$ affects $Y$ only through $M$. Hence, under the sharp null, $Z$ affects $Y$ only through $M$. A simple approach to testing the sharp null in IV settings is thus to apply the results in \Cref{cor: testable implications sharp null} for experiments, relabeling the randomized treatment as $Z$ and ignoring the endogenous take-up $D$. In Section \ref{appendix sec: IV} we show that this is equivalent to applying \Cref{cor: testing sharp null - po version} when one imposes monotonicity of $M(d)$ in $d$, but does not exhaust all of the information in the data without this monotonicity condition.}

\paragraph{Conditional unconfoundedness.} Suppose that $D$ is as good as randomly assigned conditional on observable characteristics, $D \indep (Y(\cdot,\cdot), M(\cdot)) \mid X$. Under the overlap condition that $\eta < E[D \mid X] < 1-\eta$ for some $\eta>0$, the distributions of $(Y^{\text{tot}}(d),M(d))$ are non-parametrically identified by re-weighting the observed outcomes by the propensity score $p(X) := E[D \mid X]$,
\begin{align*}
& P(Y^{\text{tot}}(1) \in A, M(1) =m_{k}) = E\left[ \frac{D}{p(X)} 1\{ Y \in A, M =m_{k}\}  \right] \\
& P(Y^{\text{tot}}(0) \in A, M(0) =m_{k}) = E\left[ \frac{1-D}{1-p(X)} 1\{ Y \in A, M =m_{k}\}  \right] . \numberthis \label{eqn: id ipw}
\end{align*}

\noindent Hence, one can apply the results in \Cref{prop: lower bounds on tv - po version} and \Cref{cor: testing sharp null - po version} to obtain lower-bounds on the fraction of always-takers affected by the treatment, and to test the sharp null.$\blacktriangle$

\paragraph{Difference-in-differences.} Consider a two-period setting where no units are treated in the first period and units with $D=1$ are treated in the second period. Under a parallel trends assumption for $Y^{\text{tot}}(0)$, we can identify $E[Y_2^{\text{tot}}(1) - Y_2^{\text{tot}}(0) \mid D=1]$, the average treatment effect on the treated (ATT) in period 2 (where the 2 subscript denotes the second time period). Likewise, under a parallel trends assumption for $M(0)$, we can identify $E[M_2(1) - M_2(0) \mid D=1]$, the ATT on $M_2$. We may then be interested in the extent to which the ATT for $Y_2$ is driven by the effect of the treatment on $M_2$. However, the standard parallel trends assumptions for $Y^{\text{tot}}(0)$ and $M(0)$ identify only the counterfactual \emph{means} for the treated group and not the counterfactual \emph{distributions}, as would be required to apply the results in \Cref{prop: lower bounds on tv - po version} and \Cref{cor: testing sharp null - po version}. However, several papers have developed extensions of the standard difference-in-differences approach that allow one to infer the full counterfactual distribution for the treated group \citep{athey_identification_2006, callaway_quantile_2019, roth_when_2023}. These approaches could be applied to identify the distributions of $(Y^{\text{tot}}_{2}(d), M_2(d)) \mid D=1$, which could then be used in conjunction with \Cref{prop: lower bounds on tv - po version} and \Cref{cor: testing sharp null - po version} to examine the extent to which the ATT on $Y_2$ is driven by the effect on $M_2$. $\blacktriangle$ \\ %

The approach to inference described in \Cref{sec: inference} naturally extends to these settings as well. In \Cref{sec: inference}, we considered inference based on a vector of estimates $\hat{p}$, where each element of $\hat{p}$ corresponded to an estimate of a probability of the form $P(Y^{\text{tot}}(d) \in A, M=m_{k})$ under the assumption of randomly assigned treatment. To test the sharp null under the identifying strategies described above, one simply replaces the $\hat{p}$ in \Cref{sec: inference} with analogous estimates of $P(Y^{\text{tot}}(d) \in A, M=m_{k})$ derived under the alternative identifying assumptions. For example, in IV settings we could use two-stage least squares estimates based on the sample analog to the identification result in \eqref{eqn: id in iv}; under conditional unconfoundedness, we could use inverse-probability weighting estimates based on sample analogs to equation \eqref{eqn: id ipw} (one could likewise use outcome-modeling or doubly-robust methods); and in difference-in-differences settings we could use estimates obtained from any of the distributional difference-in-differences approaches mentioned above. 

We note that the approach introduced in this section uses only the information about the implied marginal distributions of $(Y^{\text{tot}}(d),M(d))$. Under conditional unconfoundedness, one could potentially use information on the conditional distributions $(Y^{\text{tot}}(d),M(d)) \mid X$.\footnote{The same applies to conditional exogeneity or conditional parallel trends in the IV and difference-in-differences settings, respectively.} Indeed, if treatment is as-good-as-randomly assigned conditional on $X$, then the implications of the sharp null in \Cref{cor: testable implications sharp null} should hold conditional on $X$ (almost surely). If $X$ is discrete and takes on a small number of values (e.g. an indicator for gender), then it is straightforward to apply the testable content separately for each possible value of $X$, and inference can be conducted by simply stacking the moment inequalities for each value of $X$. When $X$ is continuously distributed, however, fully exploiting the covariates becomes complicated since it requires estimating conditional distributions and involves the infinite-dimensional conditional type shares, $\theta(X)$. The approach we described here is simpler, since it does not involve estimating full conditional distributions and has a finite-dimensional parameter $\theta$. However, it may be conservative since it does not exploit all the information available.  \citet{carr_testing_2023} and \citet{farbmacher_instrument_2022} propose methods for testing IV validity conditional on covariates in the setting with binary endogenous treatments, which can be applied off-the-shelf in the setting in \Cref{sec: binary case} with binary $M$ and monotonicity. Whether these approaches can be extended to the setting with multi-valued $M$ and relaxations of monotonicity is an interesting question for future research.

\section{Empirical Applications}\label{sec: applications}

\subsection{\citet{bursztyn_misperceived_2020} revisited}

We now revisit our application to \citet{bursztyn_misperceived_2020} from \Cref{sec: binary case}. Recall that our treatment $D$ is random assignment to an information treatment about other men's beliefs about women working outside the home, $M$ is sign-up for the job-search service, and $Y$ is an indicator for whether the wife applies for jobs outside of the home. For our main specification, we restrict attention to the majority of men who at baseline under-estimate other men's beliefs, so that the monotonicity assumption that treatment weakly increases job-search service is plausible. (We find similar results when including all men; see \Cref{sec:empirical appendix}.)

\paragraph{Statistical significance.} Recall from \Cref{fig:bursztyn_bin} that the testable implications of the sharp null were rejected based on the empirical distribution. Using the approach to inference described above, we find these violations are in fact statistically significant, with a $p$-value of $0.02$
\unskip using the CS test.\footnote{Since the outcome is binary, no discretization is needed for this application. The $p$-value reported here is the smallest value of $\alpha$ for which the test rejects.} (We obtain similar results using the other tests; see \Cref{tbl:pvals}.) The data thus provides strong evidence that the impact of the information treatment on long-run labor market outcomes does not operate solely through the sign-up for the job-search service. In particular, there are some never-takers who would not sign up for the service under either treatment who are nevertheless induced to apply for jobs by the treatment. We thus see that, for at least some people, the information treatment has meaningful impact outside of the lab, beyond its impact on job-search service sign-up. 

\paragraph{Magnitudes of alternative mechanisms.} How large are the effects of the information treatment for those who are not induced to sign-up for the job-search service? \Cref{prop: testable implications for tv} gives us a lower bound on the fraction of the always-takers/never-takers who are affected by the treatment despite having no effect on job-search service signup. Our estimates of the lower bounds suggest that at least $11$
\unskip percent of ``never-takers'' who would not be signed up for the job-search service under either treatment are nevertheless affected by the treatment. (We obtain a trivial lower-bound of 0 for the ``always-takers''.) Applying the results in \Cref{prop: lee bounds with unknown theta}, we also estimate lower and upper bounds on the average effect for these never-takers of $0.11$
\unskip to $0.18$
\unskip.\footnote{In this simple setting with a binary outcome, the lower bound for the average effect corresponds exactly to our lower bound on the fraction of always-takers affected.} For comparison, our estimate of the overall average treatment effect is $0.12$
\unskip. The effect for never-takers is thus of a fairly similar magnitude to that of the total population, despite the fact that they have no change in job-search service signup. If we were willing to assume that the direct effects (i.e. effects not through the job-search service) were similar between always-takers, never-takers, and compliers (granted, a strong assumption), this would imply that the majority of the total effect operates through the information treatment. 

\paragraph{Robustness to monotonicity violations.} Our baseline results impose the monotonicity assumption that receiving the information that other men are more open to women working than one initially thought only increases job-search service sign-up. This could be violated if, for example, there is measurement error in the initial elicitation of beliefs, so that some men included in our sample actually initially over-estimated other men's beliefs. To explore robustness to violations of the monotonicity assumption, we re-compute our bounds on the fraction of never-takers affected allowing for up to $\bar{d}$ fraction of the population to be defiers. We find that the estimated lower-bound is positive for $\bar{d}$ up to $0.07$
\unskip, which corresponds to $7$
\unskip\% of the population being defiers, or put otherwise, $0.33$
\unskip defiers for every complier. 

\subsection{\citet{baranov_maternal_2020}}

We next examine the setting of \citet{baranov_maternal_2020}. They present long-run results on an RCT that randomized access to a cognitive behavioral therapy (CBT) program intended to reduce depression for pregnant women and recent mothers. In a seven-year followup, they find that the program substantially reduced depression and increased measures of women's financial empowerment, such as having control over finances and working outside of the home. They are then interested in the mechanisms by which treating depression increases financial empowerment. They therefore examine a variety of intermediate outcomes. Two of the outcomes for which they find positive effects of the treatment are the presence of a grandmother in the household (a proxy for family support) and the women's self-reported relationship quality with the husband (on a 1-5 scale). They write (p. 849):
\begin{quote}
\small{
  These results suggest that improved social support within the household, either through a relationship with the husband or asking grandmothers for help, might be a mechanism underlying the effectiveness of this CBT intervention. }
\end{quote}
\noindent The tools developed above allow us to test the completeness of these conjectured mechanisms. Can the presence of a grandmother or improved relationship quality, either individually or together, explain the impact on financial empowerment, or must there be other mechanisms at play as well? We begin by analyzing each of these mechanisms separately, and then turn to studying the combination of the two. 

Since the outcome in \citet{baranov_maternal_2020} is a continuous index, we rely on a discretization using 5 bins, which we found in our Monte Carlo simulations calibrated to this application led to a reasonable tradeoff between size control and power (although with moderate size distortions for the setting with binary $M$). This also roughly aligns with our heuristic for choosing the bin size in the setting with binary $M$, as it yields a median cell count of 14. In the setting with non-binary $M$, this leads to a median cell count of 8, which is somewhat below our heuristic threshold of 15; using 2 bins would deliver a count of 15. Nevertheless, in the main text we present results using 5 bins to maintain consistency in the outcome variable across different specifications, and because the Monte Carlo results suggest that this choice performs well in this application with multi-valued $M$ despite the small cell size. This also gives the $\nu_k$ naturally-interpretable units as the fraction of always-takers whose outcome changes quintile when treated. In \Cref{tbl:pvals-bins}, we find qualitatively similar results using 2 or 10 bins. 

\paragraph{Grandmother mechanism.} We first examine whether the effects of the intervention can be explained through the binary mechanism of whether a grandmother is present in the household (measured at the 7-year follow-up).  \Cref{fig: grandma} shows estimates of $P(Y=y,M=0 \mid D=d)$ for both $d=1$ and $d=0$, similar to \Cref{fig:bursztyn_bin} for our previous application. If one imposes monotonicity, then as derived in \Cref{sec: binary case} we should have that $P(Y=y, M=0 \mid D=1) \leq P(Y=y, M=0 \mid D=0)$ for all values of $y$. As shown in the figure, however, this inequality appears to be violated at large values of $y$, suggesting that the outcome for some never-takers improved when receiving the treatment. These violations of the sharp null are statistically significant (CS $p=$ $0.02$
\unskip). Our estimates of the lower bound derived in \Cref{prop: testable implications for tv} imply that at least $19$
\unskip percent of never-takers are affected by the treatment. Thus, we can reject that the entirety of the treatment effect operates through increased grandmother presence in the home.\footnote{Specifically, we can reject that the effect operates through increasing \emph{long-run} grandmother presence, as measured at the 7-year follow-up. The results are less conclusive using the presence of a grandmother at the 1-year follow-up: we obtain $p=$$0.19$
\unskip, although the point estimates suggest that $14$
\unskip percent of never-takers are affected by treatment.} These conclusions rely on the monotonicity assumption that receiving CBT weakly increases the presence of the grandmother; this could be violated if, for example, some grandmothers were present when the mother was struggling but decided they were no longer needed as the mother improved. As before, we can explore robustness to allowing for defiers: our estimated lower bounds on the fraction of never-takers affected remain positive unless we allow for at least $11$
\unskip percent of the population to be defiers, or equivalently, $0.51$
\unskip defiers per complier. 

\paragraph{Relationship quality mechanism.} We next examine relationship quality (as of the 7-year follow-up) as the mechanism, which is measured on a 1-5 scale. We can thus apply the methods for multi-valued $M$ developed in \Cref{sec: general theory}. Under the monotonicity assumption that CBT improves the relationship with the husband, we reject the sharp null using CS ($p=$ $0.03$
\unskip); we obtain a point estimate of the lower bound on the fraction of always-takers affected (pooling across different values of $M$) of $10$
\unskip\%.\footnote{The monotonicity assumption requires that the population CDF of $M \mid D=1$ is everywhere smaller than the population CDF of $M \mid D=0$. This is satisfied at three of the four support points of the empirical distribution. However, the empirical CDF in the treated group is $0.015$
\unskip larger at $M=4$, although this difference is not statistically significant from zero ($p$=$0.75$
\unskip). Thus, the empirical distribution violates monotonicity, although we cannot reject that monotonicity holds in the population. To compute our estimate of the lower bound on the fraction of always-takers affected using the empirical distribution, we therefore allow for the minimum number of defiers compatible with the empirical distribution of $M \mid D$ ($0.015$
\unskip). We apply an analogous approach when considering the grandmother and relationship-quality mechanisms jointly (using the minimal relaxation of the elementwise version of monotonicity).\label{fn: monotonicity violated}} There is thus some evidence that the effect of CBT on financial empowerment does not operate entirely through improvements in relationship quality. The lower bound on the fraction of always-takers affected remains positive allowing for up to $8$
\unskip\% of the population to be defiers.

\begin{figure}
  \centering
  \includegraphics[width = 0.75\linewidth]{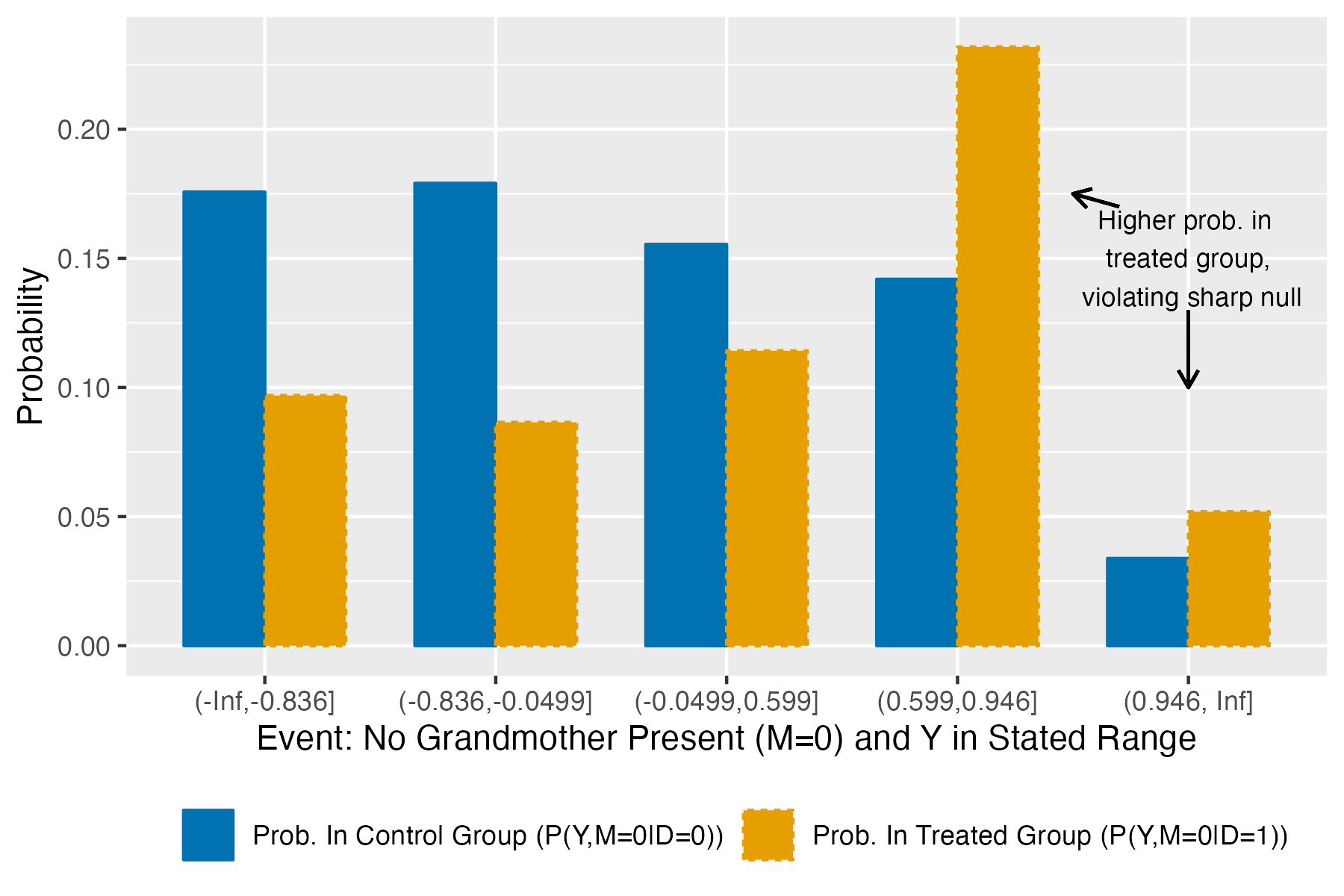}
  \caption{Testable Implications of the Sharp Null for the Grandmother Mediator in \citet{baranov_maternal_2020}}
  \label{fig: grandma}
  \floatfoot{Note: This figure shows testable implications of the sharp null of full mediation in \citet{baranov_maternal_2020}, similar to \Cref{fig:bursztyn_bin}. The mediator is presence of a grandmother in the home. The bars show estimates of probabilities of the form $P(Y^{disc} =y, M=0 \mid D=d)$, where $Y^{disc}$ is a discretization of the outcome (an index of mother's financial empowerment) into 5 bins. Under the sharp null of full mediation, we should have that $P(Y^{disc} =y, M=0 \mid D=0) \geq P(Y^{disc} =y, M=0 \mid D=1)$, but this appears to be violated for large values of $y$, as indicated with the black arrows.}
\end{figure}

\paragraph{Combinations of mechanisms.} Can the combination of the grandmother and relationship-quality mechanisms explain the improvement in financial empowerment? To evaluate this, we consider the case where $M$ is a vector containing both candidate mechanisms. Our estimated lower bound on the fraction of always-takers affected is $7$
\unskip\%. However, our test of the sharp null, imposing the monotonicity assumption that treatment increases each of the elements of $M$, is not statistically significant at conventional levels (CS $p = $ $0.65$
\unskip). Thus, while the point estimate suggests a moderate violation of the sharp null, we do not significantly reject the null hypothesis that the combination of these two mechanisms, which the authors interpret broadly as proxies for ``social support within the household'', can explain the effect of CBT on financial empowerment. This of course does not establish that no other mechanisms are at play, but rather that the data are statistically consistent with this null hypothesis at conventional levels.

\section{Conclusion} \label{sec: conclusion}

This paper develops tests for the ``sharp null of full mediation'' that the effect of a treatment $D$ on an outcome $Y$ operates only through a conjectured set of mediators $M$. A key observation is that when $M$ is binary, existing tools for testing the validity of the LATE assumptions can be used for testing the null. We develop sharp testable implications in a more general setting that allows for multi-valued and multi-dimensional $M$, and allows for relaxations of the monotonicity assumption. Our results also provide lower bounds on the size of the alternative mechanisms for always-takers. We illustrate the usefulness of these tests in two empirical applications. 

Future work might extend the analysis in this paper in several directions. First, our analysis focuses on the case where $M$ is discrete. Although one can discretize $M$ under the assumptions described in \Cref{rmk: binning M}, an interesting question for future work is whether one can impose alternative assumptions that allow for testing the sharp null directly when $M$ is continuous. One potentially fruitful direction is to explore whether methods for testing instrument validity with a continuous treatment \citep[e.g.][]{dhaultfoeuille_testing_2024} can be adapted to this setting. Second, our current analysis allows the potential outcomes to depend arbitrarily on $M$, and does not impose any assumptions on how $M$ is assigned. In some settings, however, it may be reasonable to restrict the magnitude of the effect of $M$ on $Y$, or to restrict the degree of endogeneity of $M$. Incorporating such restrictions may lead to sharper testable implications. Finally, it may be interesting to extend our results to settings with non-binary treatments.

\FloatBarrier 
\bibliography{Bibliography}

\newpage
\appendix 

\setcounter{page}{1}
 \renewcommand{\figurename}{Appendix Figure}
\setcounter{figure}{0}
\renewcommand{\tablename}{Appendix Table}
\setcounter{table}{0}
\noindent \textbf{For Online Publication}

\section{Proofs of Results in Main Text}
\label{sec: proofs}
We prove \Cref{prop: lower bounds on tv - po version} before \Cref{prop:
  testable implications for tv}, since \Cref{prop: testable implications for tv}
follows as a corollary to \Cref{prop: lower bounds on tv - po version}. In all
proofs presented in Appendix \ref{sec: proofs}, we relabel $m_{k}$ as $k$ for
notational convenience.
\paragraph{Proof of \Cref{prop: lower bounds on tv - po version}}

\paragraph{Proof of Part 1.}
By the law of total probability, we have that
\begin{align*}
P(Y^{\text{tot}}(1) \in A, M(1) = k) &= \sum_{l} P(M(0) =l,M(1)=k, Y(1,k) \in A)  \\
&=\sum_{l} P(M(0) =l,M(1)=k) \cdot P(Y(1,k) \in A \mid M(0) =l,M(1)=k) \\ 
&=\sum_{l} \theta_{lk} \cdot P(Y(1,k) \in A \mid G=lk)    
\end{align*}
 
\noindent and similarly,
$$P(Y^{\text{tot}}(0) \in A, M(0) = k) = \sum_{l} \theta_{kl} P(Y(0,k) \in A \mid G=kl). $$
\noindent Combining the previous two displays, it follows that
\begin{align*}
&\theta_{kk} \left( P(Y(1,k) \in A \mid G=kk) - P(Y(0,k) \in A \mid G=kk)  \right) \\
&= P(Y^{\text{tot}}(1) \in A, M(1) = k) - P(Y^{\text{tot}}(0) \in A, M(0) = k)  - \\
& \hspace{1cm} \sum_{l:l \neq k} \theta_{lk} P(Y(1,k) \in A \mid G=lk) + \sum_{l:l \neq k} \theta_{kl}  P(Y(0,k) \in A \mid G=kl)  \\
&\geq P(Y^{\text{tot}}(1) \in A, M(1) = k) - P(Y^{\text{tot}}(0) \in A, M(0) = k) - \sum_{l:l \neq k} \theta_{lk}
\end{align*}
\noindent where the inequality uses the fact that probabilities are bounded between $0$ and $1$. Taking a $\sup$ over Borel sets $A$, we obtain that
\begin{align*}
\theta_{kk} TV_{kk} &\geq \sup_A \left\{ P(Y^{\text{tot}}(1) \in A, M(1) = k) - P(Y^{\text{tot}}(0) \in A, M(0) = k) \right\} - \sum_{l: l \neq k} \theta_{lk}\\
&= \sup_A \, \Delta^*_k(A) - \sum_{l: l \neq k} \theta_{lk} ,
\end{align*}
where
$$TV_{kk} : = \sup_A \left\{ P(Y(1,k) \in A \mid G=kk) - P(Y(0,k) \in A \mid G=kk)  \right\}$$
\noindent is the total variation (TV) distance between $Y(1,k) \mid G=kk$ and $Y(0,k) \mid G=kk$. To establish the first claim, it thus suffices to show that $P(Y(1,k) \neq Y(0,k) \mid G=kk) \geq TV_{kk}$. Recall, however, that the TV distance is the Wasserstein-0 distance \citep[see, e.g.,][]{villani_optimal_2009}, and thus
$$TV_{kk} = \inf_{\substack{Q \text{ s.t. } \\ (Y(1,k),Y(0,k)) \sim Q \\ Y(1,k) \sim P_{Y(1,k) \mid G=kk} \\ Y(0,k) \sim P_{Y(0,k) \mid G=kk}} } E_Q[1\{Y(1,k) \neq Y(0,k)\}],$$
\noindent where $P_{Y(d,k) \mid G=kk}$ is the marginal distribution of $Y(d,k) \mid G =kk$. Since $E_Q[1\{Y(1,k) \neq Y(0,k)\}] = P_Q(Y(1,k) \neq Y(0,k))$, it follows (from the definition of the $\inf$) that $P(Y(1,k) \neq Y(0,k) \mid G=kk) \geq TV_{kk}$, which completes the proof of the first claim. The second claim is immediate from the fact that $\theta \in \Theta_I^*$ by construction.

\paragraph{Proof of Part 2.} Fix $\tilde\theta^* \in \Theta_I^*$. Since $M(d)$ has finite support, there exists a dominating, positive $\sigma$-finite measure $\mu$\footnote{Specifically, we can take $\mu(\cdot) = \sum_{m,d} \mu_{Y^{\text{tot}}(d) \mid M(d) =m}(\cdot)$, where $\mu_{Y^{\text{tot}}(d) \mid M(d) =m}$ is the probability measure for $Y^{\text{tot}}(d) \mid M(d) =m$ if $P(M(d) =m) > 0$ and zero otherwise. By construction $\mu$ is a positive $\sigma$-finite dominating measure, and hence the densities exist by the Radon-Nikodym theorem for $m,d$ such that $P(M(d) =m) >0$. For $m,d$ such that $P(M(d)=m) =0$, we can trivially set $f_{Y^{\text{tot}}(d) \mid M(d)=m}$ to be any probability density wrt $\mu$.} and densities $f_{Y^{\text{tot}}(d) \mid M(d)}$ such that for $d=0,1$ and all $k$,
$$P(Y^{\text{tot}}(d) \in A, M(d) = k) = P(M(d)=k) \cdot \int_A f_{Y^{\text{tot}}(d) \mid M(d)=k} \, d\mu.$$
\noindent We define the partial density $f_{Y^{\text{tot}}(d),M(d)=k}(y) := P(M(d)=k) \cdot  f_{Y^{\text{tot}}(d) \mid M(d)=k}(y)$. Note that
$$\sup_A \Delta_k^*(A) = \int_{\mathcal{Y}} \left( f_{Y^{\text{tot}}(1),M(1)=k} - f_{Y^{\text{tot}}(0),M(0)=k} \right)_+ \, d\mu.$$

We start by providing an intuitive sketch of the proof.\footnote{We thank an anonymous referee for the suggestion.} 
For a mediator value $k$, let \[f_{min}(y) := \min\{ f_{Y^{\text{tot}}(1),M(1)=k}(y), f_{Y^{\text{tot}}(0),M(0)=k}(y)\}\] be the minimum of the partial densities. Intuitively, $f_{min}(y)$ is an upper bound on the density of individuals who can have $(Y^{\text{tot}}(1),M(1))= (Y^{\text{tot}}(0),M(0)) = (y,k)$, i.e. an upper bound on the share of individuals who are $k$-always takers with $Y(1)=Y(0)=y$. It follows that $\int f_{\mathrm{min}}(y) \, d\mu(y)$ is an upper bound on the fraction of individuals who are $k$-always-takers with $Y(1) = Y(0)$. It turns out that the quantity $\eta_k$ given in \eqref{eqn: lower bound tv tildetheta*} is equal to $(\tilde\theta_{kk}^* - \int f_{\mathrm{min}}(y) \, d\mu(y))_+ $, the amount by which the share of $k$-always-takers, $\tilde\theta_{kk}^*$, exceeds the upper bound on the share of individuals who can be $k$-always-takers with $Y(1)=Y(0)$. If $\eta_k = 0$, then the share of $k$-always-takers is smaller than the upper bound. We can then construct distributions for the potential outcomes and mediators such that the sharp null holds: we specify all $k$-always-takers to have $Y(1)=Y(0)$, with the marginal density of their potential outcomes proportional to $f_{min}$. We then allocate the outcomes of the remaining types to match the marginals of $(Y^{\text{tot}}(d),M(d))$. If $\eta_k > 0$, then we cannot possibly allocate all $k$-always-takers to have $Y(1) = Y(0)$. Instead, we construct distributions of potential outcomes and mediators to achieve the maximum possible share of always-takers with $Y(1) = Y(0)$, $\int f_{\mathrm{min}}(y) d\mu(y)$. We do this by assigning $\int f_{\mathrm{min}}(y) \, d\mu(y)$ share of the population to be always-takers with $Y(1)=Y(0)$, with the marginal distribution of their potential outcomes proportional to $f_{min}$. We then assign the remaining mass of $k$-always-takers, $\eta_k$, to have $Y(1) \neq Y(0)$ in a way that is compatible with the observable data. 

We now continue with the formal proof of sharpness, which shows exactly how to construct these densities. We begin by proving the following lemma. 
\begin{lem}\label{lem: sharpness step}
Suppose that there exist collections of probability densities (measurable wrt $\mu$), $\mathcal{F}_1 := (f^*_{Y(1,k)\mid G=lk} )_{l=0}^{K-1}$ and $\mathcal{F}_0 := (f^*_{Y(0,k) \mid G=kl})_{l=0}^{K-1}$, such that for all $k$,
\begin{align}
 & f_{Y^{\text{tot}}(1),M(1)=k} = \sum_{l} \tilde\theta^*_{lk} f^*_{Y(1,k) \mid G = lk} \label{eqn: data matching 1} \\
  &f_{Y^{\text{tot}}(0),M(0)=k} = \sum_{l} \tilde\theta^*_{kl} f^*_{Y(0,k) \mid G = kl}, \label{eqn: data matching 2}
\end{align}
\noindent and for all $k$ with $\tilde\theta^*_{kk} > 0$,
\begin{equation}
 \tilde\theta^*_{kk} TV_{kk} = \eta_k, \text{ where } TV_{kk} := \int_{\mathcal{Y}} \left(f^*_{Y(1,k) \mid G=kk} - f^*_{Y(0,k) \mid G=kk} \right)_+ \, d\mu. \label{eqn: match tv} 
\end{equation}
\noindent Then there exists a joint distribution $P^*$ for $(Y(\cdot,\cdot),M(\cdot))$ satisfying the conditions of \Cref{prop: lower bounds on tv - po version} part 2, i.e. such that $P^*(G=lk) = \tilde\theta_{lk}^*$ for all $l,k$; $\tilde\theta_{kk}^* P^*(Y(1,k) \neq Y(0,k) \mid G = kk) = \eta_k$ for all $k$; and $P^*(Y(1,m) \neq Y(0,m) \mid G = lk) = 0$ if either $l \neq k$ or $m \not\in \{l,k\}$. 
\end{lem}

\noindent\textit{Proof of Lemma.} We will construct a distribution $P^*$ under which $(Y(1,k),Y(0,k)) \indep (Y(1,k'),Y(0,k')) \mid G$ for $k \neq k'$. That is, we will construct a $P^*$ that takes the form
$$P^*((Y(0,0),Y(1,0)) \in B_0,..., (Y(0,K-1),Y(1,K-1)) \in B_{K-1}, G=lk) = \tilde\theta_{lk}^* \cdot \prod_{m = 0}^{K-1} \int_{B_m} f^*_{(Y(1,m),Y(0,m)) \mid G=lk} \, \, d\tilde{\mu}_{m,lk},$$
\noindent where the $\tilde{\mu}_{m,lk}$ are measures on $\mathcal{Y}^2$ with one-dimensional marginals dominated by $\mu$. Note that this construction implies that $P^*(G=lk) = \tilde\theta^*_{lk}$. Let $f^{**}_{Y(d,k) \mid G=g}$ denote the implied marginal distribution over $Y(d,k) \mid G=g$ under $P^*$. Then $P^*$ matches the marginals $(Y^{\text{tot}}(d),M(d))$ if and only if 
\begin{align}
& P^*(M(1) =k ) = P(M(1) =k) \text{ for all } k \label{eqn: P* matches M1}\\
& P^*(M(0) =k ) = P(M(0) =k) \text{ for all } k \label{eqn: P* matches M0} \\
&  \sum_{l} \tilde\theta^*_{lk} f^{**}_{Y(1,k) \mid G=lk} = f_{Y^{\text{tot}}(1), M(1) =k} \text{ ($\mu$-a.s., for all $k$)} \label{eqn: P* matches pd 1}\\
&  \sum_{l} \tilde\theta^*_{kl} f^{**}_{Y(0,k) \mid G=kl} = f_{Y^{\text{tot}}(0), M(0) =k}  \text{ ($\mu$-a.s., for all $k$)} \label{eqn: P* matches pd 0}
\end{align}
\noindent Note that by construction, $P^*(M(1)=k) = \sum_{l} \tilde\theta^*_{lk}$. However, since $\tilde\theta^* \in \Theta_I^{*}$, we have that $\sum_{l} \tilde\theta^*_{lk} = P(M(1) =k)$, and hence \eqref{eqn: P* matches M1} holds. Likewise, we have that $P^*(M(0)=k) = \sum_{l} \tilde\theta^*_{kl} = P(M(0) =k)$, so \eqref{eqn: P* matches M0} holds. Observe further that if $f^{**}_{Y(1,k)\mid G=lk} = f^{*}_{Y(1,k)\mid G=lk}$ and $f^{**}_{Y(0,k)\mid G=kl} = f^{*}_{Y(0,k)\mid G=kl}$ for all $l,k$ such that $\tilde\theta^*_{lk}>0$, then \eqref{eqn: P* matches pd 1} and \eqref{eqn: P* matches pd 0} hold by assumption. To complete the proof of the lemma, it thus suffices to show that we can construct joint densities $f^*_{Y(1,m),Y(0,m) \mid G}$ satisfying the conditions of \Cref{prop: lower bounds on tv - po version} part 2 with marginals matching $f^{*}_{Y(1,k)\mid G=lk}$ and $f^{*}_{Y(0,k)\mid G=kl}$ for all $l,k$ such that $\tilde\theta_{lk}^*>0$.

Note that for $l \neq k$, the requirement that $P^*$ has marginals $f^{*}_{Y(1,k)\mid G=lk}$ and $f^{*}_{Y(0,k)\mid G=kl}$ for all $l,k$ such that $\tilde\theta_{lk}^*>0$ does not restrict the joint density for $(Y(1,m),Y(0,m)) \mid G=lk$ for $m \not \in \{l,k\}$. Hence, for $m \not\in \{l,k\}$ we may pick any joint distribution for $(Y(1,m),Y(0,m)) \mid G=lk$ (with marginals dominated by $\mu$) satisfying the sharp-null restriction $P^*(Y(1,m)=Y(0,m) \mid G=lk)=1$ (e.g., if $\mu$ is counting measure, we may choose the degenerate distribution, $(Y(1,m),Y(0,m))=(0,0)$ w.p. 1). It then suffices to specify the remaining two joint densities as
\begin{align*}
 & f^*_{Y(1,k), Y(0,k) \mid G =lk }(y_1,y_0) \propto f^*_{Y(1,k) \mid G=lk}(y_1) \cdot 1\{y_1 = y_0\} \\
 & f^*_{Y(1,l), Y(0,l) \mid G =lk }(y_1,y_0) \propto f^*_{Y(0,l) \mid G=lk}(y_0) \cdot 1 \{ y_1 = y_0\},
\end{align*}
which yields the desired marginals and ensures that $P^*(Y(1,m) = Y(0,m) \mid G=lk) = 1$ for all $m$.\footnote{The densities in the previous display are measurable with respect to the dominating measure $\tilde\mu$ such that for $B \subset \mathcal{Y}^2$, $\tilde{\mu}(B) = \mu(\{ y \,:\, (y,y) \in B \})$.} 

Next, consider the case where $G=kk$. Observe that the requirement that $P^*$ has marginals $f^{*}_{Y(1,k)\mid G=lk}$ and $f^{*}_{Y(0,k)\mid G=kl}$ for all $l,k$ such that $\tilde\theta_{lk}^*>0$ does not restrict the joint distribution of $(Y(1,m),Y(0,m))\mid G=lk$ for $m \neq k$, and hence for $m \neq k$ we can choose $f^*_{Y(1,m),Y(0,m) \mid G=kk} $ corresponding to any arbitrary density (with marginals measurable wrt $\mu$) such that $P^*(Y(1,m) = Y(0,m) \mid G=kk) =1$. Now, suppose first that $\tilde\theta_{kk}^* > 0$. Recall that for scalar random variables $H_1$ and $H_2$ with densities $h_1$ and $h_2$ (measurable wrt $\mu$), respectively, there exists a coupling $(H_1,H_2) \sim Q$ with marginals matching $H_1$ and $H_2$ such that $P_Q(H_1 \neq H_2) = tv$, where $tv = \int_{\mathcal{Y}} (h_1-h_2)_+ \, d\mu$ is the total variation distance between $H_1$ and $H_2$. However, by assumption, $\int (f^*_{Y(1,k) \mid G=kk} - f^*_{Y(0,k) \mid G=kk})_+ \, d\mu = \eta_k / \tilde\theta_{kk}^*$, and hence there exists a joint density $f^{*}_{Y(0,k),Y(1,k) \mid G=kk}$ with marginals $f^*_{Y(1,k) \mid G=kk}$ and $f^*_{Y(0,k) \mid G=kk}$ such that $P^*(Y(1,k) \neq Y(0,k) \mid G=kk) = \eta_k / \tilde\theta_{kk}^*$, and hence \eqref{eqn: tv equality - po version} holds. Finally, suppose that $\tilde\theta_{kk}^* =0$. We claim that in this case $\eta_k =0$. Observe that since $\tilde\theta^* \in \Theta_I^*$ and $\tilde\theta_{kk}^* =0$, we have that $\sum_{l: l \neq k} \tilde\theta^*_{lk} = P(M(1) =k)$. However, 
\begin{align*}
 \sup_A \left\{P(Y^{\text{tot}}(1) \in A, M(1) =k) - P(Y^{\text{tot}}(0) \in A, M(0) =k)\right\} \leq \sup_A \left\{P(Y^{\text{tot}}(1) \in A, M(1) =k) \right\}  = P(M(1) =k),
\end{align*}
\noindent and hence $$\sup_A \left\{P(Y^{\text{tot}}(1) \in A, M(1) =k) - P(Y^{\text{tot}}(0) \in A, M(0) =k) \right\} - \sum_{l: l \neq k} \tilde\theta_{lk}^* \leq 0, $$
\noindent which implies that $\eta_k =0$. It follows that \eqref{eqn: tv equality - po version} holds (with both sides of the equation equal to zero) regardless of the value of $ P^*(Y(1,m_{k}) \neq Y(0,m_{k}) \mid G = kk)$. It thus suffices to choose $f^*_{Y(0,k),Y(1,k) \mid G=kk}$ to be any joint distribution with marginals $f^*_{Y(0,k) \mid G=kk}$ and $f^*_{Y(1,k) \mid G=kk}$ (e.g. using the perfect-dependence copula). $\blacktriangle$

We now show that there exist densities satisfying the conditions of the Lemma. (For the remainder of the proof, all densities are measurable wrt $\mu$, and integrals are taken wrt $\mu$, so we omit the dependence for ease of notation.) Fix $k$. There are four possible cases, which we consider in turn below.

\paragraph{Case 1 ($\tilde\theta_{kk}^* = 0$).} Suppose $\tilde\theta_{kk}^* = 0$. Choose
\begin{align*}
  &  f^*_{Y(1,k) \mid G=lk} = f_{Y^{\text{tot}}(1) \mid M(1) =k} \text{ for all } l  \\
  &  f^*_{Y(0,k) \mid G=kl} = f_{Y^{\text{tot}}(0) \mid M(0) =k} \text{ for all } l .
\end{align*}

\noindent Note that 
$$ \sum_{l} \tilde\theta^*_{lk} f^*_{Y(1,k) \mid G = lk} =  \left( \sum_{l} \tilde\theta^*_{lk} \right) f_{Y^{\text{tot}}(1) \mid M(1) = k} = P(M(1) =k) \cdot f_{Y^{\text{tot}}(1) \mid M(1) = k}  = f_{Y^{\text{tot}}(1),M(1)=k},$$
\noindent where the first equality uses the construction of $f^*$, the second equality uses the fact that $\tilde\theta^* \in \Theta_I^*$, and the final equality uses the definition of the partial density. It follows that \eqref{eqn: data matching 1} holds. Equation \eqref{eqn: data matching 2} can be verified analogously. Since $\tilde\theta_{kk}^* =0$, we do not need to verify \eqref{eqn: match tv}. (As shown at the end of the proof of \Cref{lem: sharpness step}, $\eta_k=0$ in this case, so \eqref{eqn: match tv} holds trivially.)

For the remainder of the proof, we suppose that $\tilde\theta_{kk}^*>0$. We define $f_{min} := \min\{f_{Y^{\text{tot}}(1),M(1)=k}, f_{Y^{\text{tot}}(0),M(0)=k} \}$.

\paragraph{Case 2 ($\tilde\theta_{kk}^*>0$, $\eta_k>0$, and $f_{min}=0$ ($\mu$-a.e.)).}
 Assume that $\eta_k > 0$ and that $f_{min} =0$ ($\mu$-a.e.). Consider the densities of the potential outcomes
\begin{align*}
 &f^*_{Y(1,k) \mid G =g} = f_{Y^{\text{tot}}(1),M(1)=k} / P(M(1)=k) \text{ for all } g\\
 &f^*_{Y(0,k) \mid G =g} = f_{Y^{\text{tot}}(0),M(0)=k} / P(M(0) =k) \text{ for all } g.
\end{align*}
\noindent By construction, the densities are non-negative and integrate to 1, and thus are valid densities. Since the definition of $\Theta_I^*$ implies that $\sum_{l} \tilde\theta^*_{lk} = P(M(1)=k)$ and $\sum_{l} \tilde\theta^*_{kl} = P(M(0)=k)$, it is immediate that \eqref{eqn: data matching 1} and \eqref{eqn: data matching 2} hold. Moreover, since $f_{min} =0$, it follows that $f_{Y^{\text{tot}}(0),M(0)=k} =0$ whenever $f_{Y^{\text{tot}}(1),M(1)=k} >0$, and consequently $(f^*_{Y(1,k) \mid G=kk} -f^*_{Y(0,k) \mid G=kk})_+ = f^*_{Y(1,k) \mid G=kk}$. It follows that $$\tilde\theta^*_{kk} \int_{\mathcal{Y}} (f^*_{Y(1,k) \mid G = kk} - f^*_{Y(0,k) \mid G = kk})_+ = \tilde\theta^*_{kk} \int_{\mathcal{Y}} f^*_{Y(1,k) \mid G=kk} = \tilde\theta^*_{kk}.$$ Note, however, that 
\begin{align*}
    \eta_k = &\int_{\mathcal{Y}} \left(f_{Y^{\text{tot}}(1),M(1)=k} - f_{Y^{\text{tot}}(0),M(0)=k} \right)_+ - \sum_{l:l \neq k} \tilde\theta^*_{lk} \\
    =  & \int_{\mathcal{Y}} f_{Y^{\text{tot}}(1),M(1)=k} - \sum_{l:l \neq k} \tilde\theta^*_{lk} \\ = & P(M(1) =k) - \sum_{l:l \neq k} \tilde\theta^*_{lk} \\ = & \tilde\theta^*_{kk}
\end{align*}
\noindent where the first equality uses the fact that $\sup_A \Delta_k^*(A) = \int_{\mathcal{Y}} \left(f_{Y^{\text{tot}}(1),M(1)=k} - f_{Y^{\text{tot}}(0),M(0)=k} \right)_+$ and the assumption that $\eta_k >0$; the second equality uses the fact that $f_{min} =0$, and thus $f_{Y^{\text{tot}}(0),M(0)=k}$ is zero whenever $f_{Y^{\text{tot}}(1),M(1)=k}>0$ as argued above; and the final equality uses the fact that $P(M(1)=k) = \tilde\theta^*_{kk} + \sum_{l:l \neq k} \tilde\theta^*_{lk}$ by the definition of the identified set $\Theta_I^*$. It follows from the previous two displays that \eqref{eqn: match tv} holds.

\paragraph{Case 3 ($\tilde\theta_{kk}^*>0$, $\eta_k>0$, and $f_{min}>0$ on a set of positive measure).}
Suppose next that $\eta_k >0$ and that $f_{min}>0$ on a set of positive measure (wrt $\mu$). Then $\int_{\mathcal{Y}} f_{min}>0$, and since $f_{min} \geq 0$ by construction, it follows that $\tilde{f}_{min} = f_{min} / \int_{\mathcal{Y}} f_{min}$ is a valid density. Define $f_d := f_{Y^{\text{tot}}(d),M(d)=k} - f_{min}$ and $\tilde{f}_d := f_d / \int_{\mathcal{Y}} f_d$. We claim that the $\tilde{f}_d$ are valid densities. First, observe from the definition of $f_{min}$ that $f_d \geq 0$ everywhere. To show that $\tilde{f}_d$ is a valid density, it thus remains to show that $\int_{\mathcal{Y}} f_d > 0$, in which case $\tilde{f}_d$ is well-defined and integrates to 1 by construction. Observe, however, that since $\eta_k > 0$,
\begin{align*}
0 < \sup_{A} \Delta_k^*(A) = \int_{\mathcal{Y}} (f_{Y^{\text{tot}}(1),M(1)=k} - f_{Y^{\text{tot}}(0),M(0)=k})_+ = \int_{\mathcal{Y}} f_1  
\end{align*}
\noindent where the second equality follows from the fact that $(A-B)_+ = A - \min\{A,B\}$ and the definition of $f_1$. We thus see that $\int_{\mathcal{Y}} f_1 >0$. Next, observe that if $f_0 = 0$ ($\mu$-a.e.), then $(f_{Y^{\text{tot}}(1),M(1)=k} - f_{Y^{\text{tot}}(0),M(0)=k})_+ = f_{Y^{\text{tot}}(1),M(1)=k} - f_{Y^{\text{tot}}(0),M(0)=k}$ ($\mu$-a.e.), and thus
\begin{align*}
    \sup_{A} \Delta_k^*(A) = & \int_{\mathcal{Y}} f_{Y^{\text{tot}}(1),M(1)=k} - f_{Y^{\text{tot}}(0),M(0)=k} \\ = & P(M(1)=k ) - P(M(0)=k) \\ = & \sum_{l:l \neq k} \tilde\theta^*_{lk} - \sum_{l:l \neq k} \tilde\theta^*_{kl} \\ \leq &  \sum_{l:l \neq k} \tilde\theta^*_{lk},
\end{align*}
\noindent which implies that $\sup_{A} \Delta_k^*(A)  -  \sum_{l:l \neq k} \tilde\theta^*_{lk} \leq 0$, which contradicts the assumption that $\eta_k > 0$. Hence, we see that $f_0 >0$ on a set of positive measure (wrt $\mu$), and thus $\int_{\mathcal{Y}} f_0 > 0$, completing the proof that the $\tilde{f}_d$ are valid densities. Now, let $\nu_k^* = \eta_k / \tilde\theta^*_{kk}$, and construct the densities as follows: 
\begin{align*}
& f^*_{Y(d,k) \mid G=kk} = (1 - \nu_k^*) \tilde{f}_{min} + \nu_k^*  \tilde{f}_d \text{ for } d=0,1\\
& f^*_{Y(1,k) \mid G=g} = \tilde{f}_1 \text{ for } g \in \{lk \,:\, l \neq k\} \\
& f^*_{Y(0,k) \mid G=g} = \tilde{f}_0 \text{ for } g \in \{kl \,:\, l \neq k\}.
\end{align*}

\noindent To verify that $f^*_{Y(d,k) \mid G=kk}$ is a valid density, we will show that $\nu_k^* \in [0,1]$, in which case $f_{Y(d,k) \mid G=kk}$ is a convex combination of valid densities and hence a valid density. Note that $\nu_k^* = \eta_{k} / \tilde\theta^*_{kk}$, where $\eta_k>0$ and $\tilde\theta^*_{kk}>0$ by assumption, from which we see that $\nu_k^* \geq 0$. To show that $\nu_k^* \leq 1$, observe that
\begin{align*}
 \nu_k^* &= \frac{ \sup_{A} \left\{ P(Y^{\text{tot}}(1) \in A, M(1)=k ) - P(Y^{\text{tot}}(0) \in A, M(0)=k) \right\} - \sum_{l: l
  \neq k} \tilde\theta^*_{lk} }{ \tilde\theta^*_{kk} }\\ 
  &\leq \frac{ \sup_{A} \left\{ P(Y^{\text{tot}}(1) \in A, M(1)=k ) \right\} - \sum_{l: l
  \neq k} \tilde\theta^*_{lk} }{ \tilde\theta^*_{kk} } \\
  & = \frac{ P(M(1)=k) - \sum_{l: l
  \neq k} \tilde\theta^*_{lk} }{ \tilde\theta^*_{kk} } \\&= \frac{ \tilde\theta^*_{kk} }{ \tilde\theta^*_{kk} }.   
\end{align*}
\noindent We have thus verified that the density for $f_{Y(d,k) \mid G=kk}$ is valid.

We now verify that the specified densities satisfy \eqref{eqn: data matching 1}. Note that 
$$\sum_{l} \tilde\theta^*_{lk} f^*_{Y(1,k) \mid G=lk} = \left(\sum_{l:l \neq k} \tilde\theta^*_{lk} + \tilde\theta^*_{kk} \nu_k^* \right) \frac{f_1}{ \int_{\mathcal{Y}} f_1 } + \tilde\theta^*_{kk} (1-\nu_k^*) \frac{f_{min}}{\int_{\mathcal{Y}} f_{min}}.$$

\noindent Since $f_1 + f_{min} = f_{Y^{\text{tot}}(1), M(1)=k}$ by the definition of $f_1$, to verify \eqref{eqn: data matching 1} it suffices to verify that $\left(\sum_{l:l \neq k} \tilde\theta^*_{lk} + \tilde\theta^*_{kk}\nu_k^*  \right) / \int_{\mathcal{Y}} f_{1} = 1$ and $\tilde\theta^*_{kk} (1-\nu_k^*) / \int_{\mathcal{Y}} f_{min}=1$. Observe, however, that
\begin{align*}
 \nu_k^* &= \frac{1}{\tilde\theta^*_{kk}} \left(\sup_A \Delta_k^*(A) - \sum_{l:l \neq k} \tilde\theta^*_{lk} \right)  \\
 &= \frac{1}{\tilde\theta^*_{kk}} \left( \int_{\mathcal{Y}} (f_{Y^{\text{tot}}(1),M(1)=k} - f_{min}) -  \sum_{l:l \neq k} \tilde\theta^*_{lk}\right) \\
 &= \frac{1}{\tilde\theta^*_{kk}} \left( P(M(1)=k ) - \int_{\mathcal{Y}} f_{min} - \sum_{l:l \neq k} \tilde\theta^*_{lk} \right) \\
 &= \frac{1}{\tilde\theta^*_{kk}} \left( \tilde\theta^*_{kk} - \int_{\mathcal{Y}} f_{min} \right)  = 1 - \frac{\int_{\mathcal{Y}} f_{min}}{\tilde\theta^*_{kk}}
\end{align*}

\noindent where the first equality uses the definition of $\nu_k^*$ and the assumption that $\eta_k >0$; the second equality uses the fact that $\int_{\mathcal{Y}} (f-g)_+ = \int_{\mathcal{Y}} (f - \min\{f,g\})$; the third equality uses basic properties of densities; and the fourth equality uses the fact that $P(M(1)=k ) = \sum_{l:l \neq k} \tilde\theta^*_{lk} + \tilde\theta^*_{kk}$ since $\tilde\theta^* \in \Theta^*_I$. It is then immediate from the previous display that $\tilde\theta^*_{kk} (1-\nu_k^*) / \int_{\mathcal{Y}} f_{min}=1$. To show that $\left(\sum_{l:l \neq k} \tilde\theta^*_{lk} + \tilde\theta^*_{kk}\nu_k^*  \right) / \int_{\mathcal{Y}} f_{1} = 1$, we again use the fact that $P(M(1)=k) = \sum_{l:l \neq k} \tilde\theta^*_{lk} + \tilde\theta^*_{kk}$, to obtain that $\sum_{l:l \neq k} \tilde\theta^*_{lk} + \tilde\theta^*_{kk}\nu_k^* = P(M(1)=k) - (1-\nu_k^*) \tilde\theta^*_{kk} = P(M(1)=k) - \int_{\mathcal{Y}} f_{min}$, where the second equality uses the result in the previous display. However, from the definition of $f_1$, $\int_{\mathcal{Y}} f_1 = \int_{\mathcal{Y}} f_{Y^{\text{tot}}(1), M(1)=k} - f_{min} = P(M(1)=k ) - \int_{\mathcal{Y}} f_{min}$, and we thus see that  $\left(\sum_{l:l \neq k} \tilde\theta^*_{lk} + \tilde\theta^*_{kk}\nu_k^*  \right) / \int_{\mathcal{Y}} f_{1} = 1$, as needed to verify \eqref{eqn: data matching 1}. An analogous argument can be used to verify \eqref{eqn: data matching 2}. 

To show that the specified densities match \eqref{eqn: match tv}, note that the construction of $$\tilde{f}_d \, \propto \, f_{Y^{\text{tot}}(d), M(d)=k} - f_{min}$$ implies that $\tilde{f}_0 =0$ whenever $(\tilde{f}_1 - \tilde{f}_0)_+ > 0$ and likewise $\tilde f_1 =0$ whenever $(\tilde f_1-\tilde f_0)_+=0$. It follows that $\int_{\mathcal{Y}} (\tilde{f}_1 - \tilde{f}_0)_+ = \int_{\mathcal{Y}} \tilde{f}_1 = 1$. Hence,
\begin{align*}
\tilde\theta^*_{kk} \int_{\mathcal{Y}} (f^*_{Y(1,k) \mid G = kk} - f^*_{Y(0,k) \mid G = kk})_+ = \tilde\theta^*_{kk} \int_{\mathcal{Y}} \nu_k^* (\tilde{f}_1 - \tilde{f}_0)_+   = \tilde\theta^*_{kk} \nu_k^* = \eta_k,
\end{align*}
\noindent as needed.

\paragraph{Case 4 ($\tilde\theta_{kk}^*>0$ and $\eta_k=0$).}
Next, consider the case where $\eta_k =0$ and hence $\sup_A \Delta_k^*(A) - \sum_{l:l \neq k} \tilde\theta^*_{lk} \leq 0$. Consider the densities
\begin{align*}
& f^*_{Y(1,k) \mid G =kk} = f^*_{Y(0,k) \mid G =kk} = f_{min}/ \int_{\mathcal{Y}} f_{min} \\
& f^*_{Y(1,k) \mid G =g} = \frac{1}{\sum_{l: l \neq k} \tilde\theta^*_{lk}} \left( f_{Y^{\text{tot}}(1), M(1)=k} - \tilde\theta^*_{kk} \frac{f_{min}}{\int_{\mathcal{Y}} f_{min}} \right) \text{ for all } g \in \{lk \,:\, l \neq k\} \\
& f^*_{Y(0,k) \mid G =g} = \frac{1}{\sum_{l: l \neq k} \tilde\theta^*_{kl}} \left( f_{Y^{\text{tot}}(0), M(0)=k } - \tilde\theta^*_{kk} \frac{f_{min}}{\int_{\mathcal{Y}} f_{min}} \right) \text{ for all } g \in \{kl \,:\, l \neq k\}. 
\end{align*}

We now verify that the specified densities are in fact proper. First, we showed above that if $f_{min} = 0$ ($\mu$-a.e.), then $\eta_k = \tilde\theta^*_{kk} >0$. Hence, since $\eta_k =0$, it must be that $\int_{\mathcal{Y}} f_{min} > 0$, so that $f_{min} / \int_{\mathcal{Y}} f_{min}$ is a proper density. Next, we verify that the specified densities for $g \neq kk$ are non-negative. Recall that by assumption $\sup_A \Delta_k^*(A) - \sum_{l:l \neq k} \tilde\theta^*_{lk} \leq 0$. Note, further, that 
$$\sup_A \Delta_k^*(A) = \int_{\mathcal{Y}} f_{Y^{\text{tot}}(1),M(1)=k} - f_{min} = P(M(1)=k) - \int_{\mathcal{Y}} f_{min},$$ and hence
$$P(M(1)=k ) - \int_{\mathcal{Y}} f_{min} - \sum_{l: l \neq k} \tilde\theta^*_{lk}\leq 0 .$$
\noindent However, since $P(M(1)=k ) - \sum_{l: l \neq k} \tilde\theta^*_{lk} = \tilde\theta^*_{kk}$ by the definition of $\Theta^*_I$, we see from the previous display that $\int_{\mathcal{Y}} f_{min} \geq \tilde\theta^*_{kk}$, and thus $\frac{\tilde\theta^*_{kk}}{\int_{\mathcal{Y}} f_{min}} \leq 1$. But since $f_{Y^{\text{tot}}(d),M(d)=k} \geq f_{min}$ by construction, it follows that $f_{Y^{\text{tot}}(d),M(d)=k} - \frac{\tilde\theta^*_{kk}}{\int_{\mathcal{Y}} f_{min}} f_{min} \geq 0$, and hence the specified densities for $f_{Y(d,k) \mid G=g}$ for $g \neq kk$ are non-negative. To see that these densities integrate to 1, observe that $$\int_{\mathcal{Y}} \left( f_{Y^{\text{tot}}(1),M(1)=k} - \frac{\tilde\theta^*_{kk}}{\int_{\mathcal{Y}} f_{min} } f_{min}\right) = P(M(1)=k) - \tilde\theta^*_{kk} = \sum_{l: l \neq k} \tilde\theta^*_{lk}$$
and similarly
$$\int_{\mathcal{Y}} \left( f_{Y^{\text{tot}}(0),M(0)=k} -  \frac{\tilde\theta^*_{kk}}{\int_{\mathcal{Y}} f_{min} } f_{min} \right) = P(M(0)=k) - \tilde\theta^*_{kk} = \sum_{l: l \neq k} \tilde\theta^*_{kl}.$$

\noindent Finally, it is trivial to verify from the construction of the densities above that equations \eqref{eqn: data matching 1}, \eqref{eqn: data matching 2}, and \eqref{eqn: match tv} hold. $\blacktriangle$

\paragraph{Proof of \Cref{prop: testable implications for tv}}
Under \Cref{asm: independence}, $(Y,M) \mid D=d \sim (Y^{\text{tot}}(d),M(d))$ for $d=0,1$ (where recall $Y^{\text{tot}}(d) := Y(d,M(d))$). The result then follows immediately from \Cref{prop: lower bounds on tv - po version}. Specifically, as noted in the main text, under \Cref{asm: independence}, $\Delta_k(A) = \Delta_k^*(A)$ and $\Theta_I = \Theta_I^*$. Hence, \eqref{eqn: lower bound tv - true theta*} from \Cref{prop: lower bounds on tv - po version} implies \eqref{eqn: testable implication tv}, and likewise \eqref{eqn: lower bound tv tildetheta*} implies \eqref{eqn: testable implication tv - feasible}, which yields the first part of the Proposition. For the second part, for any $\tilde\theta \in \Theta_I = \Theta_I^*$, part 2 of \Cref{prop: lower bounds on tv - po version} implies that there exists a distribution $P^{*}$ for $(Y(\cdot,\cdot),M(\cdot))$ that is consistent with the marginals $(Y^{\text{tot}}(d),M(d)) \sim (Y,M) \mid D=d$ for $d=0,1$ such that $$\tilde\theta_{kk} P^*(Y(1,k) \neq Y(0,k) \mid G=kk) = \left(\sup_A \Delta_k^*(A) - \sum_{l:l\neq k} \tilde\theta_{lk}\right)_+ = \left(\sup_A \Delta_k(A) - \sum_{l:l\neq k} \tilde\theta_{lk}\right)_+$$ and $P^*(Y(1,m) \neq Y(0,m) \mid G=lk) = 0$ if either $l \neq k$ or $m \not\in \{l,k\}$. Part 2 of the Proposition is thus satisfied for $P^{\dagger}$ the distribution over $(Y(\cdot,\cdot),M(\cdot),D)$ defined such that $(Y(\cdot,\cdot), M(\cdot)) \sim P^*$, $D \sim P_{D}$ (where $P_D$ is the marginal distribution of $D$ under $P_{obs}$), and $D \indep (Y(\cdot,\cdot),M(\cdot))$. $\blacktriangle$

\paragraph{Proof of \Cref{cor: testable implications sharp null}}
\begin{proof}
Under the sharp null of full mediation, $\nu_k=0$ for all $k$, so \eqref{eqn: testable implication tv - feasible} gives
$0\ge\big(\sup_A \Delta_k(A)-\sum_{l\neq k}\tilde\theta_{lk}\big)_+$ for some $\tilde\theta\in\Theta_I$, which is equivalent to \eqref{eqn: testable implications sharp null}.
Conversely, if \eqref{eqn: testable implications sharp null} holds for some $\tilde\theta\in\Theta_I$, then by the second part of \Cref{prop: testable implications for tv} there exists a distribution $P^\dagger$ consistent with the observable data and Assumptions \ref{asm: independence} and \ref{asm: restricted theta} such that $P^\dagger (Y(1,m) \neq Y(0,m) \mid G = lk) P^\dagger(G=lk) = 0$ for all $l, k$ and $m$. By the law of total probability, $P^\dagger(Y(1,m) \neq Y(0,m)) = \sum_{l,k} P^\dagger (Y(1,m) \neq Y(0,m) \mid G = lk) P^\dagger(G=lk) = 0$ for all $m$, so that the sharp null of full mediation holds under $P^\dagger$.
\end{proof}

\paragraph{Proof of \Cref{cor: testing sharp null - po version}}
\begin{proof}
The proof is nearly identical to that for \Cref{cor: testable implications sharp null}, except citing \Cref{prop: lower bounds on tv - po version} instead of \Cref{prop: testable implications for tv}, but we provide it for completeness. Under the sharp null of full mediation, $\nu_k=0$ for all $k$, so \eqref{eqn: lower bound tv tildetheta*} gives
$0\ge\big(\sup_A \Delta_k^*(A)-\sum_{l\neq k}\tilde\theta_{lk}^*\big)_+$ for some $\tilde\theta^*\in\Theta_I^*$, which is equivalent to \eqref{eqn: testable implications sharp null - po version}. Conversely, if \eqref{eqn: testable implications sharp null - po version} holds for some $\tilde\theta^*\in\Theta_I^*$, then by the second part of \Cref{prop: lower bounds on tv - po version} there exists a distribution $P^*$ consistent with the marginals for $(Y^{\text{tot}}(d),M(d))$ and the restriction that $\theta \in R$ such that $P^* (Y(1,m) \neq Y(0,m) \mid G = lk) P^*(G=lk) = 0$ for all $l, k$ and $m$. By the law of total probability, $P^*(Y(1,m) \neq Y(0,m)) = \sum_{l,k} P^*(Y(1,m) \neq Y(0,m) \mid G = lk) P^*(G=lk) = 0$ for all $m$, so that the sharp null of full mediation holds under $P^*$.
    
\end{proof}

\section{Additional Theoretical Results}
\subsection{Bounds on average effects for always-takers} \label{subsec: bounds on ADE}

The quantity $\nu_k$ provides a measure of what fraction of $k$-always-takers are affected by alternative mechanisms. However, in some settings we may also be interested in the average \emph{magnitude} of the alternative mechanisms for the always-takers. In this section, we derive bounds on 
$$ADE_{k} := E[Y(1,m_k) - Y(0,m_k) \mid G =kk ] ,$$
\noindent the average direct effect of the treatment on the outcome for the $k$-always-takers. This provides an alternative measure of the size of the alternative mechanisms for the always-takers. 

To derive bounds for $ADE_{k}$, we first derive bounds on $E[Y(1,m_k) \mid G=kk]$. Observe that individuals with $M=m_k$ and $D=1$ must be either $k$-always-takers or $lk$-compliers. The share of $k$-always-takers among this population is given by $\check\theta_{kk}^1 := P(G=kk \mid D=1, M=m_k) = \frac{\theta_{kk}}{ P(M=m_k \mid D=1) }$. It follows that the observable distribution of $Y \mid D=1, M=m_k$ is a mixture with weight $\check{\theta}_{kk}^1$ on $Y(1,m_k) \mid G=kk$ and weight $(1-\check{\theta}_{kk}^1)$ on the distribution of $Y(1,m_k)$ for $lk$-compliers. We can thus obtain bounds on $E[Y(1,m_k) \mid G =kk]$ by considering the worst-case scenario where the $k$-always-takers compose the bottom $\check{\theta}_{kk}^{1}$ fraction of the $Y \mid D=1, M=m_k$ distribution, and the best-case scenario where they compose the top $\check{\theta}_{kk}^{1}$ fraction. 

The following lemma formalizes this intuition for obtaining bounds on $E[Y(1,m_k) \mid G=kk]$, and applies analogous logic to obtain bounds on $E[Y(0,m_k) \mid G=kk]$. For notation, for a distribution $F$, let $F^{-1}(u) = \inf \{y \,:\, F(y) \geq u \}$ be the $u$th quantile of $F$.
\begin{lem} \label{lem: lee bounds w known theta}
Suppose \Cref{asm: independence} holds. Then if $\check\theta_{kk}^1>0$,
$$ \frac{1}{\check\theta_{kk}^1} \int_0^{\check\theta_{kk}^1}  F_{Y \mid D=1, M=k}^{-1}(u) \, du \leq E[Y(1,k) \mid G =kk ] \leq \frac{1}{\check\theta_{kk}^1} \int_{1-\check\theta_{kk}^1}^1  F_{Y \mid D=1, M=k}^{-1}(u) \, du$$
\noindent and if $\check\theta_{kk}^0>0$,
$$ \frac{1}{\check\theta_{kk}^0} \int_0^{\check\theta_{kk}^0}  F_{Y \mid D=0, M=k}^{-1}(u) \, du \leq  E[Y(0,k) \mid G =kk ] \leq \frac{1}{\check\theta_{kk}^0}\int_{1-\check\theta_{kk}^0}^1  F_{Y \mid D=0, M=k}^{-1}(u) \, du.$$
The bounds are sharp in the sense that there exists a distribution $P^\dagger$ for $(Y(\cdot,\cdot),M(\cdot), D)$ consistent with \Cref{asm: independence} and the observable data and with $\theta_{lk} = P^\dagger(G=lk)$ such that the bounds hold with equality. If the distributions of $Y \mid D=d,M=k$ are continuous, then the bounds can equivalently be written as 
$$ E[Y \mid M=k, D=1, Y \leq y^1_{\check\theta_{kk}^1}] \leq E[Y(1,k) \mid G =kk ] \leq E[Y \mid M=k, D=1, Y \geq y^1_{1-\check\theta_{kk}^1}]$$
and 
$$ E[Y \mid M=k, D=0, Y \leq y^0_{\check\theta_{kk}^0}] \leq E[Y(0,k) \mid G =kk ] \leq E[Y \mid M=k, D=0, Y \geq y^0_{1-\check\theta_{kk}^0}],$$
\noindent where $y^d_{q} := F^{-1}_{Y \mid D=d, M=k}(q)$ is the $q$th quantile of $Y \mid D=d, M=k$. 
\end{lem}

\begin{proof}
We begin by deriving the bounds for $E[Y(1,k) \mid G=kk]$. Observe that under \Cref{asm: independence}, $$F_{Y \mid D=1, M=k} = \check\theta^1_{kk} F_{Y(1,k) \mid G=kk} + (1- \check\theta^1_{kk}) H,$$   
\noindent where $H = \frac{1}{ \sum_{l: l \neq k} \theta_{lk} } \sum_{l: l \neq lk} \theta_{lk} F_{Y(1,k) \mid G=lk}$ is a valid CDF (corresponding to a mixture of the distributions of $Y(1,k) \mid G =g$ for types $g=lk$, $l \neq k$). Hence, 
$$F_{Y(1,k) \mid G =kk} = \frac{1}{\check\theta_{kk}^1} F_{Y \mid D=1, M=k} - \frac{1 - \check\theta_{kk}^1}{\check\theta_{kk}^1} H.$$
\noindent From the fact that CDFs are bounded between 0 and 1, it follows that
$$ \max\left\{\frac{1}{\check\theta_{kk}^1} F_{Y \mid D=1, M=k} - \frac{1 - \check\theta_{kk}^1}{\check\theta_{kk}^1}, 0  \right\} \leq  F_{Y(1,k) \mid G=kk} \leq \min \left\{ \frac{1}{\check\theta_{kk}^1} F_{Y \mid D=1, M=k} , 1 \right\} $$
\noindent Recall that if $F_1 \leq F_2$ everywhere for CDFs $F_1$ and $F_2$, the $F_1$ distribution first-order stochastically dominates the $F_2$ distribution, and thus $E_{F_1}[Y] \geq E_{F_2}[Y]$. Hence, we have that $E_{F_{ub}}[Y(1,k)] \leq E[Y(1,k) \mid G=kk] \leq E_{F_{lb}}[Y(1,k)]$, where $F_{lb},F_{ub}$ are respectively the lower and upper bounds on the CDF given in the previous display.  

Now, let $U$ be uniform on $[0,1]$, and consider the random variable $Y_{ub} \sim F_{Y \mid D=1, M=k}^{-1}(U) \mid U \in [0,\check\theta_{kk}^1]$. Observe that 
\begin{align*}
F_{Y_{ub}}(y) &= P(F_{Y \mid D=1, M=k}^{-1}(U) \leq y \mid U \in [0,\check\theta_{kk}^1]  ) \\
&= P(F_{Y \mid D=1, M=k}(y) \geq U \mid U \in [0,\check\theta_{kk}^1]) \\
&= \min \left\{ \frac{1}{\check\theta_{kk}^1} F_{Y \mid D=1, M=k}(y) , 1 \right\} = F_{ub}(y).
\end{align*}
\noindent It follows that $E_{F_{ub}}[Y(1,k)] = E[F_{Y \mid D=1, M=k}^{-1}(U) \mid U \in [0,\check\theta_{kk}^1]] =  \frac{1}{\check\theta_{kk}^1} \int_0^{\check\theta_{kk}^1}  F_{Y \mid D=1, M=k}^{-1}(u) \, du$, which gives the lower-bound on $E[Y(1,k) \mid G=kk]$ given in the lemma. When $Y$ is continuously distributed, note that $Y_{ub} \sim \left(Y \mid D=1,M=k, Y \leq y_{\check\theta_{kk}^1} \right)$, and thus we can also write the lower-bound as $E[Y \mid D=1, M=k, Y \leq y_{\check\theta_{kk}^1}]$. Analogously, we can verify that the random variable $Y_{lb} \sim F_{Y \mid D=1, M=k}^{-1}(U) \mid U \in [1-\check\theta_{kk}^1,1]$ has the CDF $F_{lb}$, which gives the upper bound on $E[Y(1,k) \mid G =kk]$ given in the proposition. 

To show that the lower bound is sharp, consider $P^\dagger$ such that $D \indep Y(\cdot,\cdot), M(\cdot)$ and $P^\dagger(M(0)=l,M(1)=k)= \theta_{lk}$, and the marginal distributions of the potential outcomes are such that
$Y(1,k) \mid G =kk \sim Y_{lb}$ and $Y(1,k) \mid G = lk \sim F_{Y \mid D=1, M=k}^{-1}(U) \mid U \in [\check\theta_{kk}^1,1]$ for all $g =lk$ with $l \neq k$. Then the distribution of $Y \mid M=k, D=1$ is given by the mixture: 
$$\check\theta_{kk}^1 \left(F_{Y \mid D=1, M=k}^{-1}(U) \mid U \in [0,\check\theta_{kk}^1] \right) + (1- \check\theta_{kk}^1) \left(F_{Y \mid D=1, M=k}^{-1}(U) \mid U \in [\check\theta_{kk}^1,1] \right) \sim F_{Y \mid D=1, M=k}^{-1}(U). $$
\noindent Recalling that if $Y$ has CDF $F$, then $Y \sim F^{-1}(U)$, we see that the implied distribution of $Y \mid M=k, D=1$ under $P^\dagger$ matches the observable data. (The marginals $Y(0,m) \mid G$ under $P^\dagger$ can be chosen to be any set of distributions matching the observable data; likewise the copula of potential outcomes can be chosen arbitrarily.) The sharpness of the upper bound can be shown analogously. Sharp bounds for $E[Y(0,k) \mid G =kk ]$ can be shown analogously to those for $E[Y(1,k) \mid G =kk]$. 
\end{proof}

\noindent \Cref{lem: lee bounds w known theta} immediately implies bounds on $ADE_{k}$ by differencing the inequalities for the expectations of $Y(1,m_k)$ and $Y(0,m_k)$. Note, however, that the bounds in \Cref{lem: lee bounds w known theta} involve the always-taker share $\check{\theta}_{kk}^d = \frac{\theta_{kk}}{P(M=m_{k} \mid D=d)}$, which may only be partially identified. It is straightforward to see, however, that the bounds become wider the smaller is $\check{\theta}_{kk}^d$. Intuitively, this is because the most-favorable subdistribution of fraction $\check{\theta}_{kk}^d$ is more favorable the smaller is $\check{\theta}_{kk}^d$, and likewise for the least-favorable subdistribution. Sharp bounds on $ADE_{k}$ can thus be obtained by plugging $\tilde\theta_{kk}^{min}$ into the bounds given in \Cref{lem: lee bounds w known theta}, where recall $\tilde\theta_{kk}^{min} := \inf_{\tilde\theta \in \Theta_I} \tilde\theta_{kk}$ is the minimum value of $\tilde\theta_{kk}$ consistent with $\tilde\theta \in \Theta_I$. For notation, let $LB_d(\check{\theta}^d_{kk})$ and $UB_d(\check{\theta}^d_{kk})$ denote the lower- and upper-bounds on $E[Y(d,m_k) \mid G=kk]$ given in \Cref{lem: lee bounds w known theta} as a function of $\check{\theta}^d_{kk}$. We then have the following bounds on $ADE_{k}$. While we present the result for the case where $Y$ is continuous, the analogous result holds if one replaces $LB_d$ and $UB_d$ with the expressions given in \Cref{lem: lee bounds w known theta} for the case where $Y$ is not assumed to be continuous.

\begin{prop} \label{prop: lee bounds with unknown theta}
Suppose \Cref{asm: independence,asm: restricted theta} hold and $Y$ is continuously distributed. If $\tilde\theta_{kk}^{min} = \inf_{\tilde\theta \in \Theta_I} \tilde\theta_{kk} > 0$, then bounds on $ADE_{k}$ are given as follows: 
\begin{align*}
&LB_{1}(\check{\theta}_{kk}^{1,min}) - UB_0(\check{\theta}_{kk}^{0,min}) \leq ADE_{k} \leq UB_{1}(\check{\theta}_{kk}^{1,min}) - LB_{0}(\check{\theta}_{kk}^{0,min})
\end{align*}
\noindent where $\check{\theta}_{kk}^{d,min} := \frac{\tilde\theta_{kk}^{min}}{P(M=m_k \mid D=d)}$. The lower and upper bounds are sharp in the sense that there exists a distribution $P^\dagger$ for $(Y(\cdot,\cdot),M(\cdot), D)$ consistent with the observable data and \Cref{asm: independence,asm: restricted theta} such that the bound holds with equality.
\end{prop}

\begin{proof}
From \Cref{lem: lee bounds w known theta}, we have that 
\begin{align*}
    &\inf_{\check\theta_{kk}^d \in \check\Theta_{I,d}} E[ F^{-1}_{Y \mid D=d, M=k}(U) \mid U \in [0,\check\theta_{kk}^d] ] \\ \leq & E[Y(1,k) \mid G =kk]  \\ \leq & \sup_{\check\theta_{kk}^d \in \check\Theta_{I,d}} E[F^{-1}_{Y \mid D=d, M=k}(U) \mid U \in [1-\check\theta_{kk}^d,1] ]
\end{align*}
for $U$ uniformly distributed and $\check\Theta_{I,d}$ the set of values for $\check\theta_{kk}^d = \tilde\theta_{kk}/ P(M=k \mid D=d)$ consistent with $\tilde\theta \in \Theta_I$. Since $F^{-1}_{Y \mid D=d, M=k}(U)$ is increasing in $U$, it follows that the $\inf$ and $\sup$ are both obtained at $\check\theta_{kk}^{min}$. 
The bounds for $ADE_{k} =E[Y(1,k) - Y(0,k) \mid G =kk]$ follow simply from differencing the bounds for the two potential outcomes in the previous display. Sharpness for the bounds for $ADE_{k}$ follows from the fact that, as shown in the proof to \Cref{lem: lee bounds w known theta}, for each $d=0,1$, the bounds for $E[Y(d,k) \mid G=kk]$ can be achieved only by specifying the marginals $Y(d,k) \mid G=kk$ (and choosing the remaining potential outcomes in any arbitrary way that matches the data) and thus the bounds for $d=0,1$ can be achieved simultaneously.    
\end{proof}

It is worth noting that in the simple case where $M$ is binary and one imposes monotonicity, the bounds on $ADE_{k}$ correspond to \citet{lee_training_2009}'s bounds, where $D$ is viewed as the treatment and $M$ as the ``sample selection''. In the binary case, the $ADE_{k}$ can also be viewed as what the statistics literature refers to as principal strata direct effects for the principal strata with $M(1) = M(0) = m_{k}$ \citep{frangakis_principal_2002, zhang_estimation_2003}.\footnote{\citet{vanderweele_comments_2012} argues that one should not interpret the principal stratum effect for compliers as an indirect effect, but rather a combination of the direct and indirect effects (a total effect). This critique does not apply to our analysis of the principal stratum effects for always-takers, since their value of $M$ is unaffected by $D$, and thus any effects for this subgroup must be direct effects.} \citet{flores_nonparametric_2010} observed that such bounds could be used for mediation analysis in the case of binary $M$---their Proposition 1 matches the bounds given in \Cref{lem: lee bounds w known theta} for the special case where $M$ is binary under monotonicity---although they use this primarily as an intermediate step to derive bounds on the full-population average direct effect of treatment. Our result extends these existing results for the binary case to settings where $M$ may be multi-valued (and where monotonicity may fail). 

It is also worth emphasizing that the sharp null of full mediation considered earlier is distinct from the null hypothesis that $ADE_{k}=0$ for all $k$. In particular, the sharp null imposes that the treatment does not have an effect on the outcome for any always-taker, whereas the null that $ADE_{k}=0$ imposes that the treatment does not affect the $k$-always-takers on average. This is analogous to the distinction between the sharp null considered by Fisher and the weak null considered by Neyman, applied to the sub-population of always-takers. Thus, we may be able to reject the sharp null in settings where we cannot reject the weaker null that the $ADE_{k}$ are zero. 

\subsection{Minimum value of $\nu_k$ achieved at \texorpdfstring{$\tilde{\theta}_{kk}^{min}$}{theta_kk^min}}

The following result formalizes the sense in which the lower bound on $\nu_k$ implied by \eqref{eqn: testable implication tv - feasible} from \Cref{prop: testable implications for tv} is achieved at the minimum possible value of $\tilde\theta_{kk}$ in the identified set, $\tilde\theta_{kk}^{min} := \inf_{\tilde\theta \in \Theta_I} \tilde\theta_{kk}$.

\begin{lem} \label{lem: min at thetakkmin}
 If $\tilde\theta_{kk}^{min} > 0$, then
 $$\inf_{\substack{\nu_k \geq 0, \tilde\theta \in \Theta_I \\ \text{ s.t. }  \eqref{eqn: testable implication tv - feasible}  \text{ holds} }} \nu_k=
     \frac{1}{ \tilde\theta_{kk}^{min} } \left( \sup_A \Delta_k(A) - P(M=m_k \mid D=1) + \tilde\theta_{kk}^{min} \right)_+ .$$
     On the other hand, if $\tilde\theta_{kk}^{min} = 0$, then $$\inf_{\substack{\nu_k \geq 0, \tilde\theta \in \Theta_I \\ \text{ s.t. }  \eqref{eqn: testable implication tv - feasible}  \text{ holds} }} \nu_k= 0.$$ 
\end{lem}

\begin{proof}
    For simplicity of notation, without loss of generality let $m_k = k$. We argued in the main text that \eqref{eqn: testable implication tv - feasible} can be equivalently written as 
    $$\tilde{\theta}_{kk} \nu_k \geq \left(\sup_A \Delta_k(A) - (P(M=k \mid D=1) - \tilde\theta_{kk}) \right)_+ .$$
    Hence, when $\tilde\theta_{kk}^{min}>0$, we have that
    \begin{align*}
        \inf_{\substack{\nu_k \geq 0, \tilde\theta \in \Theta_I \\ \text{ s.t. }  \eqref{eqn: testable implication tv - feasible}  \text{ holds} }} \nu_k & = \inf_{\tilde\theta \in \Theta_I} \frac{1}{\tilde\theta_{kk}} \left(\sup_A \Delta_k(A) - (P(M=k \mid D=1) - \tilde\theta_{kk}) \right)_+ \\
        &= \inf_{\tilde\theta \in \Theta_I}  \left( \frac{1}{\tilde\theta_{kk}} \left(\sup_A \Delta_k(A) - P(M=k \mid D=1) \right) + 1 \right)_+.
    \end{align*}
    \noindent To show that the $\inf$ is achieved at $\tilde\theta_{kk}^{min}$, it thus suffices to show that $$\sup_A \Delta_k(A) - P(M=k \mid D=1) \leq 0,$$ in which case the expression inside the $\inf$ is weakly increasing in $\tilde\theta_{kk}$. Note, however, that
    \begin{align*}
     \sup_A \Delta_k(A) &= \sup_A \left\{ P(Y \in A, M=k \mid D=1) - P(Y \in A, M=k \mid D=0)  \right\}  \\
     & \leq  \sup_A P(Y \in A, M=k \mid D=1) \\
     & = P(M=k \mid D=1),
    \end{align*}
    which completes the proof for the case where $\theta^{kk}_{min} > 0$. 

    Next, suppose that $\tilde\theta^{kk}_{min} = 0$. It is straightforward to verify that since $R$ is closed by \Cref{asm: restricted theta}, $\Theta_I$ is also closed. It follows that there exists $\tilde\theta^* \in \Theta_I$ such that $\tilde\theta_{kk}^* = 0$. Since $\tilde\theta^* \in \Theta_I$, we have that $\sum_{l} \tilde\theta^*_{lk} = P(M=k \mid D=1)$. Combined with the fact that $\tilde\theta_{kk}^* = 0$, we thus have that $\sum_{l:l\neq k} \tilde\theta_{lk}^* = P(M=k\mid D=1)$. This combined with the inequality in the previous display implies that $(\sup_A \Delta_k(A) - \sum_{l:l\neq k} \tilde\theta_{lk}^*)_+ = 0$, and hence \eqref{eqn: testable implication tv - feasible} holds with $\tilde\theta = \tilde\theta^*$ and $\nu_k = 0$. 
\end{proof}

\subsection{Closed-form solution for \texorpdfstring{$\tilde\theta_{kk}$}{theta_kk} with fully-ordered $M$}
The following result formalizes the closed-form solution for $\tilde\theta_{kk}^{min}$ when $M$ is fully-ordered and we impose monotonicity, as discussed in \Cref{rmk: closed form thetakk}.

\begin{lem} \label{lem: thetakk closed form}
Suppose $M$ is fully-ordered, so that $m_0 < m_1< ... < m_{K-1}$. Suppose Assumptions \ref{asm: independence} and \ref{asm: restricted theta} are satisfied, where $R = \{\theta \in \Delta : \theta_{lk} = 0 \text{ if } m_l>m_k \}$ imposes the monotonicity assumption that $M(1) \geq M(0)$. Then 
$$\tilde\theta_{kk} \geq P(M=m_k \mid D=1) - \min\{P(M=m_k \mid D=1), P(M \geq m_k \mid D=1) - P(M \geq m_k \mid D=0)\} $$
\noindent for all $\tilde\theta \in \Theta_I$, and there exists $\tilde\theta \in \Theta_I$ such that inequality holds with equality simultaneously for all $k$. 
\end{lem}

\begin{proof}
 The result is an immediate corollary of \Cref{lem: thetakk closed form - po version} below, since $\Theta_I = \Theta_I^*$ under \Cref{asm: independence}.   
\end{proof}

\begin{lem} \label{lem: thetakk closed form - po version}
Suppose $M$ is fully-ordered, so that $m_0 < m_1< ... < m_{K-1}$. Suppose \Cref{asm: restricted theta} holds with $R = \{\theta \in \Delta : \theta_{lk} = 0 \text{ if } m_l>m_k \}$, which imposes the monotonicity assumption that $M(1) \geq M(0)$ (almost surely). Let $\Theta_I^*$ be as defined in \eqref{eqn: defn of ThetaI*}. Then 
\begin{equation}
  \tilde\theta^*_{kk} \geq P(M(1)=m_k) - \min\{P(M(1)=m_k), P(M(1) \geq m_k) - P(M(0) \geq m_k)\} \label{eqn: lower bound on thetakk - po version}  
\end{equation}
\noindent for all $\tilde\theta^* \in \Theta_I^*$, and there exists $\tilde\theta^* \in \Theta_I^*$ such that the inequality holds with equality simultaneously for all $k$. 
\end{lem}

\begin{proof}
For simplicity of notation, without loss of generality let $m_k = k$. Observe that the definition of $\Theta_I^*$ together with the assumed form for $R$ implies that for any $\tilde\theta^* \in \Theta_I^*$,
\begin{equation} P(M(1)=k) = \tilde\theta^*_{kk} + \sum_{l: l<k} \tilde\theta^*_{lk} \label{eqn: moment matchign m1} \end{equation}
\noindent and hence the conclusion of the proposition holds if and only if
\begin{align}
\sum_{l:l<k} \tilde\theta^*_{lk} &\leq \min\{P(M(1)=k), P(M(1)\geq k) - P(M(0) \geq k)\}, \label{eqn: inequality ordered theta}
\end{align}
\noindent for $k=0,...,K-1$, and there exists some $\tilde\theta^* \in \Theta_I^*$ such that the inequality holds with equality for all $k$.

We first show the inequality in \eqref{eqn: inequality ordered theta}. It is immediate from \eqref{eqn: moment matchign m1} that $$\sum_{l:l<k} \tilde\theta^*_{lk} \leq P(M(1)=k).$$ Moreover, using the restriction that $\tilde\theta_{lk}^* =0$ if $l>k$, we have that
$$P(M(1) \geq k) - P(M(0) \geq k) =  \sum_{l: l<k} \hspace{.1cm} \sum_{k':k' \geq k} \tilde\theta^*_{lk'} \geq \sum_{l:l <k} \tilde\theta^*_{lk},$$
\noindent which together with the previous display gives the inequality in \eqref{eqn: inequality ordered theta}.

We next show there exists a $\tilde\theta^* \in \Theta_I^*$ such that \eqref{eqn: inequality ordered theta} holds with equality for all $k$. To obey monotonicity, we set $\tilde\theta^*_{lk}=0$ whenever $k<l$. 

We now recursively set the remaining $\tilde\theta^*_{lk}$. Start with $k=0$. Set $\tilde\theta^*_{00} = P(M(1)=0)$. Note that the monotonicity assumption that $M(1) \geq M(0)$ almost surely implies that $P(M(1)=0) \leq P(M(0)=0)$. It is then straightforward to verify that the following properties hold for $\bar{k} =0$ (in what follows, we interpret sums over empty sets as zero): 

\begin{enumerate}[(i)]
    \item 
    $\sum_{l:l<j} \tilde\theta^*_{lj} = \min\{P(M(1)=j), P(M(1)\geq j) - P(M(0) \geq j)\}$ \text{for all} $j \leq \bar{k}$

    \item
    $\sum_{l:l \leq j} \tilde\theta^*_{lj} = P(M(1)=j)$ \text{for all} $j \leq \bar{k}$
    \item
    $\sum_{l:l \leq \bar{k} } \tilde\theta^*_{jl} \leq P(M(0)=j)$ \text{for all} $j \leq \bar{k}$.
\end{enumerate}

Now, suppose that for some $k \geq 1$, $\tilde\theta^*_{lj}$ has been determined for all $l=0,...,K-1$ and all $j=0,...,k-1$, and properties (i)-(iii) hold for $\bar{k} = k-1$. (We showed above that this holds in the base case $k=1$.) Set $\tilde\theta^*_{kk} = P(M(1)=k) - \min\{P(M(1)=k), P(M(1)\geq k) - P(M(0) \geq k)\}$. For $l = 0,...,k-1$, proceed as follows

\begin{enumerate}
    \item 
    If $\sum_{l':l'<l} \tilde\theta^*_{l'k} = P(M(1)\geq k) - P(M(0) \geq k)$, then set $\tilde\theta^*_{lk} =0$.

    \item
    Otherwise, set $$\tilde\theta^*_{lk} = \min\left\{P(M(1)\geq k) - P(M(0) \geq k) - \sum_{l':l'<l} \tilde\theta^*_{l'k} \hspace{.1cm} ,\hspace{.1cm} P(M(0)=l) - \sum_{k':k'<k} \tilde\theta^*_{lk'} \right\}.$$
\end{enumerate}

\noindent Note that the first term in the minimum is weakly positive by construction while property (iii) ensures that the second term in the minimum is non-negative, so that $\tilde\theta^*_{lk} \geq 0$. We claim that the construction above implies that 
$$\sum_{l: l<k} \tilde\theta^*_{lk} = \min\{P(M(1)=k), P(M(1)\geq k) - P(M(0) \geq k)\}.$$
\noindent To see why this is the case, suppose towards contradiction that 
$$\sum_{l: l<k} \tilde\theta^*_{lk} < \min\{P(M(1)=k), P(M(1)\geq k) - P(M(0) \geq k)\}.$$
\noindent Then $\tilde\theta^*_{lk}$ is always set via step 2 in the procedure above. However, the construction of $\tilde\theta^*_{lk}$ in step 2 combined with the fact that $\sum_{l:l<k} \tilde\theta^*_{lk} < P(M(1) \geq k) - P(M(0) \geq k)$ implies that for all $l=0,...,k-1$, we have that 
$$\tilde\theta^*_{lk} = P(M(0)=l) - \sum_{j:j<k} \tilde\theta^*_{lj} .$$
\noindent Summing over $l<k$, we obtain that 
\begin{align*}
\sum_{l: l<k} \tilde\theta^*_{lk} &= \sum_{l: l<k} P(M(0)=l) - \sum_{l: l<k} \sum_{j: j<k} \tilde\theta^*_{lj} \\
&= \sum_{l: l<k} P(M(0)=l) - \sum_{j: j<k}  \sum_{l: l<k} \tilde\theta^*_{lj} \hspace{1cm} \text{(Reversing order of sums)} \\
&= \sum_{l: l<k} P(M(0)=l) - \sum_{j: j<k}  \sum_{l: l \leq j} \tilde\theta^*_{lj} \hspace{1cm} \text{(Using monotonicity)} \\
&= \sum_{l: l<k} P(M(0)=l) - \sum_{j:j<k}  P(M(1)=j) \hspace{1cm} \text{(Using property (ii))}\\
&= P(M(0)<k) - P(M(1)<k) \\
& = P(M(1) \geq k) - P(M(0) \geq k)
\end{align*}
\noindent which is a contradiction. 

\noindent It follows that property (i) holds also for $\bar{k} = k$. Likewise, the construction of $\tilde\theta^*_{kk}$ combined with property (i) implies that property (ii) holds for $\bar{k} =k$. Finally, the construction of $\tilde\theta^*_{lk}$ (particularly step 2) guarantees that property (iii) holds for $\bar{k}=k$ as well.

By induction we can obtain $\tilde\theta^*$ satisfying properties (i) through (iii) for $\bar{k}=K-1$. The resulting $\tilde\theta^*$ satisfies monotonicity and is bounded between 0 and 1 by construction. Property (ii) guarantees that $\tilde\theta^*$ matches the marginal of $M(1)$, i.e. $\sum_l \tilde\theta^*_{lk} = P(M(1)=k)$. 

It thus remains only to establish that $\tilde\theta^*$ matches the marginal distribution of $M(0)$. Property (ii) implies that $\sum_{l} \tilde\theta^*_{jl} \leq P(M(0) =j)$. To establish equality for all $j$, it thus suffices to show that $\sum_{j} \sum_{l} \tilde\theta^*_{jl} \geq \sum_j P(M(0)=j) = 1$. Note, however, that from property (ii) and monotonicity, we have
$$\sum_{j} \sum_{l} \tilde\theta^*_{jl} = \sum_{j} \left( \sum_{l: l \leq j} \tilde\theta^*_{lj} \right) = \sum_j P(M(1) =j ) = 1 , $$

\noindent which completes the proof.
\end{proof}

\subsection{Additional results for IV}
\label{appendix sec: IV}

Consider the setting in \Cref{sec: nonexperimental} in which we have a binary instrument $Z$ for $D$ that satisfies the \citet{imbens_identification_1994} assumptions. \Cref{cor: testing sharp null - po version} provides sharp testable implications of the sharp null that $D$ affects $Y$ only through $M$ for instrument-compliers. Note, however, that the instrument exclusion restriction implies that $Z$ affects $Y$ only through $D$, and the sharp null implies that $D$ affects $Y$ only through $M$. Hence, under the sharp null, it follows that $Z$ affects $Y$ only through $M$. A simple alternative approach to testing the sharp null would therefore be to apply the testable implications derived under random assignment in \Cref{cor: testable implications sharp null} viewing $Z$ as the randomized treatment and ignoring $D$. In this section, we show that the two approaches coincide when one imposes that $M(d)$ is monotonic in $d$ (i.e. $R$ is as given in \eqref{eqn: delta for monotonicity}). Thus, if one is willing to impose monotonicity, one can simply apply our approach for experiments using $Z$ as the treatment variable. If one does not impose monotonicity of $M(d)$, then this approach remains valid but potentially loses information relative to using \Cref{cor: testable implications sharp null}. 

To see why this is the case, recall from \Cref{sec: nonexperimental} that we can identify the marginal distributions $(Y^{\text{tot}}(d),M(d)) \mid C^z=1$, where $C^z$ is an indicator for instrument-compliers. We let $\Theta^*_{I,C^z}$ denote the analog to $\Theta^*_I$ defined in \eqref{eqn: defn of ThetaI*}, with all probabilities conditional on $C^z=1$. Observe that \Cref{cor: testing sharp null - po version} together with the formulas derived in \Cref{sec: nonexperimental} for $(Y^{\text{tot}}(d),M(d)) \mid C^z=1$ imply that under the sharp null, there exists some $\tilde\theta^* \in \Theta^*_{I,C^z}$ such that for all $k=0,...,K-1$,
\begin{align}
&\sup_A \frac{ \Delta^Z_k(A) }{\alpha_C} \leq \frac{ E[D \cdot 1\{M=m_k\} \mid Z=1 ] - E[D \cdot 1\{M=m_k\} \mid Z =0] }{ \alpha_C }  - \tilde\theta_{kk}^* \label{eqn: sharp null w iv using pos - general}
\end{align}
\noindent where $$\Delta^Z_k(A) := P(Y \in A, M=m_k \mid Z = 1) - P(Y \in A, M=m_k \mid Z =0)$$
and $$\alpha_C := E[D \mid Z=1] - E[D \mid Z=0]$$
is the identified share of instrument-compliers. 
Note that if $M$ is fully-ordered and $R$ imposes monotonicity, then by \Cref{lem: thetakk closed form - po version}, 
\begin{align*}
\inf_{\tilde\theta \in \Theta^*_{I,C^z}} \tilde\theta_{kk}^* = \left(P(M(1) =m_k \mid C^z =1) - \left( P(M(1) \geq m_k \mid C^z =1 ) - P(M(0) \geq m_k \mid C^z=1 \right)  \right)_+  =: \tilde\theta_{kk}^{min,*}  
\end{align*}
\noindent and there exists $\tilde\theta^* \in \Theta^*_{I,C^z}$ such that $\tilde\theta^*_{kk} = \tilde\theta_{kk}^{min,*}$ for all $k$. Using this result along with the expressions derived in \Cref{sec: nonexperimental} for $(Y^{\text{tot}}(d),M(d)) \mid C^z=1$, it follows that under monotonicity, \eqref{eqn: sharp null w iv using pos - general} is equivalent to
\begin{align*}
&\sup_A \frac{ \Delta_k^Z(A) }{\alpha_C} \leq \frac{ E[D \cdot 1\{M=m_k\} \mid Z=1 ] - E[D \cdot 1\{M=m_k\} \mid Z =0] }{ \alpha_C } \\ & - \left(\frac{ E[D \cdot 1\{M=m_k\} \mid Z=1 ] - E[D \cdot 1\{M=m_k\} \mid Z =0] }{ \alpha_C } - \frac{ P(M \geq m_k \mid Z =1) - P(M \geq m_k \mid Z =0) }{ \alpha_C } \right)_+ 
\end{align*}
\noindent Multiplying through by $\alpha_C$ and using the fact that $a-(a-b)_+ = \min\{a,b\}$ for any $a,b$, we obtain the equivalent implication
\begin{align}
\sup_A  \Delta_k^Z(A) \leq \min\{ &E[D \cdot 1\{M=m_k\} \mid Z=1 ] - E[D \cdot 1\{M=m_k\} \mid Z =0] \,,\, \nonumber \\  & \hspace{0.5cm} P(M \geq m_k \mid Z =1) - P(M \geq m_k \mid Z =0) \}. \label{eqn: sharp null w iv using pos - monotonicity}
\end{align}

On the other hand, we could ignore $D$ and simply use \Cref{cor: testable implications sharp null} to test whether $Z$ affects $Y$ only through $M$ (relabeling treatment `$D$' with treatment `$Z$' in \Cref{cor: testable implications sharp null}). The testable implication is that there exists some $\tilde\theta \in \Theta_I$ such that for all $k=0,...,K-1$
\begin{equation}
  \sup_A \Delta^Z_k(A) \leq P(M=m_k \mid Z=1) - \tilde\theta_{kk}. \label{eqn: testable imp ignoring d}  
\end{equation}
\noindent Using the expression for $\tilde\theta_{kk}^{min}$ under monotonicity given in \Cref{lem: thetakk closed form}, this is equivalent to, for all $k=0,...,K-1$,
\begin{align}
\sup_A \Delta_k^Z(A) &\leq P(M=m_k \mid Z=1) - \left(P(M=m_k \mid Z=1) - \left( P(M \geq m_k \mid Z =1) - P(M \geq m_k \mid Z =0) \right) \right)_+ \nonumber \\
&= \min\{P(M=m_k \mid Z=1) \,,\, P(M \geq m_k \mid Z =1) - P(M \geq m_k \mid Z =0) \}. \label{eqn: sharp null w iv ignoring d - monotonicity} 
\end{align}

We now claim that \eqref{eqn: sharp null w iv using pos - monotonicity} and \eqref{eqn: sharp null w iv ignoring d - monotonicity} are equivalent under our assumption that $Z$ satisfies the \citet{imbens_identification_1994} assumptions. Note that both inequalities take the form 
\begin{align}
\sup_A \Delta_k^Z(A) \leq \min\{ q \,,\, P(M \geq m_k \mid Z =1) - P(M \geq m_k \mid Z =0) \}  
\end{align}
\noindent but differ in the choice of $q$. Note that the upper bound is weakly increasing in $q$. Moreover, since $D \in \{0,1\}$,
$$E[D \cdot 1\{M=m_k\} \mid Z=1] \leq E[ 1\{M=m_k\} \mid Z=1] = P(M=m_k \mid Z=1).$$
\noindent We thus see that the upper bound in \eqref{eqn: sharp null w iv using pos - monotonicity} is weakly tighter than that in \eqref{eqn: sharp null w iv ignoring d - monotonicity}, so \eqref{eqn: sharp null w iv using pos - monotonicity} implies \eqref{eqn: sharp null w iv ignoring d - monotonicity}. 
Conversely, we will show that if \eqref{eqn: sharp null w iv using pos - monotonicity} is violated, then \eqref{eqn: sharp null w iv ignoring d - monotonicity} is also violated. Note that since $Z$ satisfies the \citet{imbens_identification_1994} assumptions, $\Delta_k^Z(A) / \alpha_C$ is the LATE of $D$ on the compound outcome $1\{Y, M = m_k\}$, i.e.
$$ \frac{\Delta_k^Z(A)}{\alpha_C} = P(Y^{\text{tot}}(1) \in A, M(1) = m_k \mid C^z =1) -  P(Y^{\text{tot}}(0) \in A, M(0) = m_k \mid C^z =1).$$

It follows that
\begin{align*}
  \sup_A \Delta_k^Z(A)&= \alpha_C \sup_A \left\{ P(Y^{\text{tot}}(1) \in A, M(1) = m_k \mid C^z=1) - P(Y^{\text{tot}}(0) \in A, M(0) =m_k \mid C^z = 1) \right\}  \\
  & \leq \alpha_C P(M(1) = m_k \mid C^z = 1) \\
  & =  E[D \cdot 1\{M=m_k\} \mid Z=1 ] - E[D \cdot 1\{M=m_k\} \mid Z =0] \numberthis \label{eqn: intermed ineq iv equiv}
\end{align*}
\noindent where the inequality uses the fact that $P(Y^{\text{tot}}(1) \in A, M(1) =m_k \mid C^z=1) \leq P(M(1) =m_k \mid C^z=1)$; and the final equality again uses the fact that $Z$ satisfies the \citet{imbens_identification_1994} assumptions to obtain that $$P(M(1) = m_k \mid C^z=1) = \frac{E[D \cdot 1\{M=m_k\} \mid Z=1] - E[D \cdot 1\{M=m_k\} \mid Z=0] }{\alpha_C}.$$

\noindent It follows from \eqref{eqn: intermed ineq iv equiv} that \eqref{eqn: sharp null w iv using pos - monotonicity} can be violated only if there exists some $k$ such that
\begin{align*}
\sup_A \Delta_k(A) > P(M \geq m_k \mid Z=1) - P(M \geq m_k \mid Z=0),    
\end{align*}
\noindent in which case \eqref{eqn: sharp null w iv ignoring d - monotonicity} is also violated. This completes the proof that the testable implications in \Cref{cor: testable implications sharp null} and \Cref{cor: testing sharp null - po version} are equivalent under monotonicity of $M(d)$. 

This equivalence breaks down if one does not impose monotonicity of $M(d)$. Intuitively, if we do not impose monotonicity, then by ignoring $D$ we lose information about the type shares $\theta$ among instrument-compliers. A simple example is as follows. Suppose there are two groups in the population occurring with equal probability. The first group are instrument-compliers ($D(1)=1,D(0)=0$) and always-takers with respect to $M$ ($M(1)=M(0)=1$) with $Y(d,m) = d$. Since $Y(d,m)$ depends on $d$, the sharp null is violated. The second group are instrument never-takers ($D(1) = D(0) = 0$) and never-takers with respect to $M$ ($M(1)=M(0)=0)$ with $Y(d,m) = 0$. Then it is straightforward to verify that \eqref{eqn: sharp null w iv using pos - general} is violated whereas \eqref{eqn: testable imp ignoring d} is not. The reason for this is that the type shares using the potential outcomes are point-identified for instrument-compliers since $P(M(1) =1 \mid C^z=1) = P(M(0) = 1 \mid C^z=1) = 1$, and hence $\Theta^*_{I,C^z}$ is a singleton with unique element $\theta^*$ such that $\theta^*_{11} = 1$ and $\theta^*_{lk} = 0$ for $lk \neq 11$. Thus \eqref{eqn: sharp null w iv using pos - general} reduces to the restriction that $\sup_A \Delta_k^Z(A) = 0$, which is violated for $k=1$. On the other hand, since $P(M=1 \mid Z=1) = P(M=1 \mid Z=0) = 0.5$, the identified set $\Theta_I$ based on \Cref{cor: testable implications sharp null} includes $\tilde\theta$ such that $\tilde\theta_{01}=\tilde\theta_{10} = 0.5$ and $\tilde\theta_{kk} =0$ for all $k$, in which case there are no always-takers and thus \eqref{eqn: testable imp ignoring d} is trivially satisfied.

\section{Additional Results on Discretizing Continuous Outcomes}
\label{sec: disc-cont-outcome}

\paragraph{Connection to conditional moment inequalities.} \Copy{andrews-shi-discussion}{As discussed in Remark \ref{rmk: discretizing y}, implementing the proposed test after discretizing the outcome variable continues to yield a valid inference procedure under the sharp null. A drawback is that the resulting testable implications are potentially no longer sharp. Discretizing the outcome is  analogous to the selection of instrument functions in tests of conditional moment inequalities, as in \citet{andrews_cmi_2013}. Indeed, in the simplest setting where $M$ is binary and monotonicity holds, \citet{mourifie_testing_2017} show that the sharp testable implications can be reformulated as conditional moment inequalities of the form $E[g(M,D) \mid Y] \geq 0$, for $g$ a function of $M,D$. \citet{andrews_cmi_2013} propose tests of inequalities of the form $E[g(M,D) \mid Y] \geq 0$ that are based on unconditional inequalities of the form $E[h(Y) \cdot g(M,D) ] \geq 0$, where $h$ is a non-negative function of $Y$. Their recommended approach is to use hyper-cubes for $h$, which for a one-dimensional $Y$ (as in our setting) corresponds to indicators for $Y$ lying in particular intervals. Thus, applying the \citet{andrews_cmi_2013} approach to the implications derived in \citet{mourifie_testing_2017} is equivalent to testing a version of the \citet{mourifie_testing_2017} implications with a discretized outcome.\footnote{In the simple case of binary $M$ and monotonicity, one could apply other tests for conditional moment inequalities that are not analogous to discretizing the outcome, but which typically rely on other tuning parameters. For example, \citet{mourifie_testing_2017} suggest an approach based on \citet{chernozhukov_intersection_2013}'s test of conditional moment inequalities, which requires the choice of a bandwidth and kernel for non-parametric mean estimation. Whether such approaches could be extended to our more general setting with multi-valued $M$ or relaxations of monotonicity strikes us an interesting question for future work.} The choice of instrument functions for conditional moment inequalities is known to be a theoretically challenging question. As a result, existing recommendations are often heuristic in nature, motivated by simulation evidence.\footnote{There are some formal results on \emph{rate-optimal} choices of test statistic for conditional moment inequalities \citep{armstrong_weighted_2014, chetverikov_adaptive_2018}, although these results do not appear to immediately apply to our general setting in which there are nuisance parameters (and thus we are interested in subvector inference).} For instance, based on Monte Carlo simulations, \citet{andrews_cmi_2013} suggest choosing instrument functions such that the expected number of observations per cell lies between 10 and 20.}

\paragraph{Monte Carlos.} In a similar spirit, as mentioned in the main text, we conduct Monte Carlo simulations calibrated to \citet{baranov_maternal_2020} using 2 or 10 bins instead of the 5 bins used in the results reported in the main text (\Cref{tab: app_sim_baranov_binM,tab: app_sim_baranov_nonbinM}). In \Cref{tab: cell_count_nonbinM}, we report how the number of independent observations per cell varies across these specifications. Since the observations are clustered, rather than counting the raw observations, we count the number of independent clusters that have at least one observation in a given cell. Here, a cell corresponds to a support point of the vector $(D, M, Y^{disc})$, where $Y^{disc}$ denotes the discretized outcome variable. For each simulated dataset, we compute the number of independent clusters per cell and record the median cell count. We then report the average of these median cell counts across simulation replications. See \Cref{sec: monte carlo} for discussion of these results.

\begin{table}[!htbp] \centering
  \caption{Simulation results for \citet{baranov_maternal_2020} with binary $M$ and different discretizations of the outcome}
    \label[appendixtable]{tab: app_sim_baranov_binM} 
    \begin{threeparttable}
\small{
\begin{tabular}{@{\extracolsep{5pt}} ccccccc} 
\hline 
\hline
\\[-2.0ex] 
\multicolumn{7}{@{} l}{Panel A: Baranov et al, 40 clusters, 2 bins}
 \\
 & $\bar{\nu}$ LB & ARP & CS & K & FSSTdd & FSSTndd \\ 
\cline{3-7} 
t=0 & $0$ & $0.086$ & $0.078$ & $0.050$ & $0.136$ & $0.126$ \\ 
t=0.5 & $0.134$ & $0.264$ & $0.256$ & $0.064$ & $0.314$ & $0.280$ \\ 
t=1 & $0.283$ & $0.828$ & $0.822$ & $0.422$ & $0.844$ & $0.830$ \\ 
 \hline 
\\[-2.0ex] \multicolumn{7}{@{} l}{Panel B: Baranov et al, 80 clusters, 2 bins}
 \\
 & $\bar{\nu}$ LB & ARP & CS & K & FSSTdd & FSSTndd \\ 
\cline{3-7} 
t=0 & $0$ & $0.046$ & $0.040$ & $0.040$ & $0.098$ & $0.090$ \\ 
t=0.5 & $0.134$ & $0.444$ & $0.430$ & $0.160$ & $0.456$ & $0.434$ \\ 
t=1 & $0.283$ & $0.978$ & $0.976$ & $0.846$ & $0.976$ & $0.976$ \\  
 \hline
\\[-2.0ex] \multicolumn{7}{@{} l}{Panel C: Baranov et al, 200 clusters, 2 bins}
 \\
 & $\bar{\nu}$ LB & ARP & CS & K & FSSTdd & FSSTndd \\ 
\cline{3-7}t=0 & $0$ & $0.052$ & $0.044$ & $0.030$ & $0.082$ & $0.078$ \\ 
t=0.5 & $0.134$ & $0.822$ & $0.816$ & $0.618$ & $0.818$ & $0.796$ \\ 
t=1 & $0.283$ & $1$ & $1$ & $1$ & $1$ & $1$ \\ 
 \hline 
\\[-2.0ex] \multicolumn{7}{@{} l}{Panel D: Baranov et al, 40 clusters, 10 bins}
 \\
 & $\bar{\nu}$ LB & ARP & CS & K & FSSTdd & FSSTndd \\ 
\cline{3-7}
t=0 & $0$ & $0.072$ & $0.188$ & $0.050$ & $0.324$ & $0.262$ \\ 
t=0.5 & $0.134$ & $0.164$ & $0.246$ & $0.064$ & $0.340$ & $0.308$ \\ 
t=1 & $0.283$ & $0.530$ & $0.658$ & $0.422$ & $0.774$ & $0.720$ \\ 
 \hline
\\[-2.0ex] \multicolumn{7}{@{} l}{Panel E: Baranov et al, 80 clusters, 10 bins}
 \\
 & $\bar{\nu}$ LB & ARP & CS & K & FSSTdd & FSSTndd \\ 
\cline{3-7}
t=0 & $0$ & $0.052$ & $0.086$ & $0.040$ & $0.208$ & $0.158$ \\ 
t=0.5 & $0.134$ & $0.272$ & $0.314$ & $0.160$ & $0.436$ & $0.368$ \\ 
t=1 & $0.283$ & $0.798$ & $0.924$ & $0.846$ & $0.960$ & $0.942$ \\ 
 \hline 
\\[-2.0ex] \multicolumn{7}{@{} l}{Panel F: Baranov et al, 200 clusters, 10 bins}
 \\
 & $\bar{\nu}$ LB & ARP & CS & K & FSSTdd & FSSTndd \\ 
\cline{3-7}
t=0 & $0$ & $0.042$ & $0.048$ & $0.030$ & $0.122$ & $0.100$ \\ 
t=0.5 & $0.134$ & $0.636$ & $0.742$ & $0.618$ & $0.804$ & $0.754$ \\ 
t=1 & $0.283$ & $0.998$ & $1$ & $1$ & $1$ & $1$ \\ 
\hline
\end{tabular} 
}

\begin{tablenotes}[flushleft]\footnotesize
\item\emph{Notes}: This table show simulation results analogous to Panels B-D of Table
  \ref{tab: main_sim_binM}, with 2 and 10 bins used for discretizing the
  outcome variable. The first column shows the value of $t$, which
  determines the distance from the null, as described in the main text. The second column shows the lower-bound on the fraction of always-takers affected by treatment, $\bar{\nu}$. The remaining columns
  contain the rejection probabilities for each of the inference methods
  considered. Panels A-C use 2 bins to discretize the outcome variable and Panels
  D-F use 10 bins. Since \cite{kitagawa_test_2015}
  does not require a discrete outcome variable, we use the outcome variable as-is when running this test (hence the results for K do not depend on the number of bins). Rejection probabilities are computed over 500
  simulation draws, under a 5\% significance level.
  \end{tablenotes}
\end{threeparttable}
\end{table}

\begin{table}[!htbp] \centering
  \caption{Simulation results for \citet{baranov_maternal_2020} with non-binary $M$ and different discretizations of the outcome}
    \label[appendixtable]{tab: app_sim_baranov_nonbinM} 
    \begin{threeparttable}
\small{
\begin{tabular}{@{\extracolsep{5pt}} cccccc} 
\hline 
\hline
\\[-2.0ex]
\multicolumn{6}{@{} l}{Panel A: Baranov et al, 40 clusters, 2 bins}
\\
 & $\bar{\nu}$ LB & ARP & CS & FSSTdd & FSSTndd \\ 
\cline{3-6}
t=0 & $0$ & $0.056$ & $0.092$ & $0.150$ & $0.112$ \\ 
t=0.5 & $0.119$ & $0.092$ & $0.206$ & $0.356$ & $0.326$ \\ 
t=1 & $0.255$ & $0.290$ & $0.856$ & $0.944$ & $0.922$ \\ 
 \hline
\\[-2.0ex] \multicolumn{6}{@{} l}{Panel B: Baranov et al, 80 clusters, 2 bins}
 \\
 & $\bar{\nu}$ LB & ARP & CS & FSSTdd & FSSTndd \\ 
\cline{3-6} 
t=0 & $0$ & $0.054$ & $0.058$ & $0.146$ & $0.110$ \\ 
t=0.5 & $0.119$ & $0.110$ & $0.392$ & $0.546$ & $0.514$ \\ 
t=1 & $0.255$ & $0.288$ & $0.986$ & $0.998$ & $0.998$ \\ 
 \hline
\\[-2.0ex] \multicolumn{6}{@{} l}{Panel C: Baranov et al, 200 clusters, 2 bins}
 \\
 & $\bar{\nu}$ LB & ARP & CS & FSSTdd & FSSTndd \\ 
\cline{3-6}
t=0 & $0$ & $0.042$ & $0.048$ & $0.100$ & $0.076$ \\ 
t=0.5 & $0.119$ & $0.104$ & $0.792$ & $0.892$ & $0.860$ \\ 
t=1 & $0.255$ & $0.422$ & $1$ & $1$ & $1$ \\ 
 \hline
\\[-2.0ex] \multicolumn{6}{@{} l}{Panel D: Baranov et al, 40 clusters, 10 bins}
 \\
 & $\bar{\nu}$ LB & ARP & CS & FSSTdd & FSSTndd \\ 
\cline{3-6}
t=0 & $0$ & $0.038$ & $0.102$ & $0.386$ & $0.264$ \\ 
t=0.5 & $0.119$ & $0.036$ & $0.256$ & $0.556$ & $0.464$ \\ 
t=1 & $0.255$ & $0.126$ & $0.818$ & $0.960$ & $0.932$ \\ 
 \hline
\\[-2.0ex] \multicolumn{6}{@{} l}{Panel E: Baranov et al, 80 clusters, 10 bins}
 \\
  & $\bar{\nu}$ LB & ARP & CS & FSSTdd & FSSTndd \\ 
\cline{3-6}
t=0 & $0$ & $0.048$ & $0.032$ & $0.282$ & $0.176$ \\ 
t=0.5 & $0.119$ & $0.050$ & $0.238$ & $0.650$ & $0.566$ \\ 
t=1 & $0.255$ & $0.134$ & $0.986$ & $0.998$ & $0.998$ \\ 
 \hline
\\[-2.0ex] \multicolumn{6}{@{} l}{Panel F: Baranov et al, 200 clusters, 10 bins}
 \\
 & $\bar{\nu}$ LB & ARP & CS & FSSTdd & FSSTndd \\ 
\cline{3-6} 
t=0 & $0$ & $0.048$ & $0.006$ & $0.182$ & $0.094$ \\ 
t=0.5 & $0.119$ & $0.068$ & $0.464$ & $0.936$ & $0.894$ \\ 
t=1 & $0.255$ & $0.264$ & $1$ & $1$ & $1$ \\ 
\\[-2.0ex]
\hline 
\end{tabular} 
}

\begin{tablenotes}[flushleft]\footnotesize
\item\emph{Notes}: This table show simulation results analogous to \Cref{tab:
    main_sim_nonbinM}, with 2 and 10 bins used for discretizing the outcome
  variable. The first column shows the value of $t$, which determines the
  distance from the null, as described in the main text. The second column shows the lower-bound on the fraction of always-takers affected by treatment, $\bar{\nu}$. The remaining columns contain the
  rejection probabilities for each of the inference methods considered. Panels
  A-C use 2 bins to discretize the outcome variable and Panels D-F use 10 bins.
  Rejection probabilities are computed over 500 simulation draws, under a 5\%
  significance level.
  \end{tablenotes}
\end{threeparttable}
\end{table}

\begin{table}[!ht]
  \centering
    \caption{Average median cell count for DGPs calibrated to \citet{baranov_maternal_2020}}
  \label[appendixtable]{tab: cell_count_nonbinM} 
  \begin{minipage}{.72\textwidth}
  \begin{center}
\begin{tabular}{@{\extracolsep{5pt}} lcccccc} 
\\[-1.8ex]\hline 
\hline \\[-1.8ex]
\\[-2.0ex] \multicolumn{7}{@{} l}{Panel A: Baranov et al, 40 clusters} \\
 \\[-1.5ex]
 & \multicolumn{3}{c}{Binary $M$} & \multicolumn{3}{c}{Non-binary $M$} \\
 & 2 bins & 5 bins & 10 bins & 2 bins & 5 bins & 10 bins \\
\cline{2-4} \cline{5-7} \\[-2.0ex]
t = 0.0 & 19.976 & 13.999 & 9.751 & 15.236 & 8.943 & 5.904 \\
t = 0.5 & 19.729 & 13.535 & 9.546 & 15.362 & 8.639 & 5.625 \\
t = 1.0 & 19.584 & 13.245 & 9.299 & 15.654 & 8.832 & 5.417 \\
\hline \\[-1.83ex]
\\[-2.0ex] \multicolumn{7}{@{} l}{Panel B: Baranov et al, 80 clusters} \\
 \\[-1.5ex]
 & \multicolumn{3}{c}{Binary $M$} & \multicolumn{3}{c}{Non-binary $M$} \\
 & 2 bins & 5 bins & 10 bins & 2 bins & 5 bins & 10 bins \\
\cline{2-4} \cline{5-7} \\[-2.0ex]
t = 0.0 & 39.978 & 28.086 & 18.980 & 30.105 & 17.461 & 11.115 \\
t = 0.5 & 39.298 & 26.851 & 18.923 & 30.562 & 16.720 & 10.916 \\
t = 1.0 & 38.797 & 26.194 & 17.780 & 31.762 & 17.485 & 10.327 \\
\hline \\[-1.83ex]
\\[-2.0ex] \multicolumn{7}{@{} l}{Panel C: Baranov et al, 200 clusters} \\
 \\[-1.5ex]
 & \multicolumn{3}{c}{Binary $M$} & \multicolumn{3}{c}{Non-binary $M$} \\
 & 2 bins & 5 bins & 10 bins & 2 bins & 5 bins & 10 bins \\
\cline{2-4} \cline{5-7} \\[-2.0ex
]t = 0.0 & 99.985 & 70.007 & 46.312 & 74.170 & 42.033 & 27.104 \\
t = 0.5 & 98.084 & 66.870 & 47.000 & 75.869 & 38.988 & 27.176 \\
t = 1.0 & 95.952 & 65.184 & 43.070 & 80.675 & 43.885 & 24.863 \\
\hline
\end{tabular}
\end{center}

\vspace{-5pt} \footnotesize{\emph{Notes:} This table reports the median number
  of independent clusters per cell, averaged over 500 simulation replications. A cell corresponds to each support point of $(D, M, Y^{disc})$, where $Y^{disc}$ is
  the discretized (to either 2, 5 or 10 bins) version of the outcome $Y$. The first
  three columns are calculated from the DGPs calibrated to Baranov et al. with
  binary $M$ (i.e., DGPs considered in Table 1 Panels B-D and Appendix Table 1)
  and the last three columns are calculated from the DGPs calibrated to Baranov
  et al. with non-binary $M$ (i.e., DGPs considered in Table 2 Panels B-D and
  Appendix Table 2).}
\end{minipage}
\end{table}

\clearpage

\section{Additional Empirical Results\label{sec:empirical appendix}}

\paragraph{Alternative sample for \citet{bursztyn_misperceived_2020}.} In our application to \citet{bursztyn_misperceived_2020} in the main text, we restrict attention to the $75$
\unskip percent of men who under-estimate other men's openness at baseline, which increases the plausibility of the monotonicity assumption. We now present analogous results using the full sample, which are similar. \Cref{fig:bursztyn_bin_appendix} is analogous to \Cref{fig:bursztyn_bin} but using the full sample, with similar qualitative patterns. The estimated lower bound on the fraction of never-takers affected, imposing monotonicity, is $8$
\unskip percent, and bounds for the average effect for never-takers are $0.08$
\unskip to $0.13$
\unskip. The lower bound on the fraction affected remains non-zero allowing for up to $5$
\unskip percent of the population to be defiers. 

\begin{figure}[!ht]
  \includegraphics[width =
  0.5\linewidth]{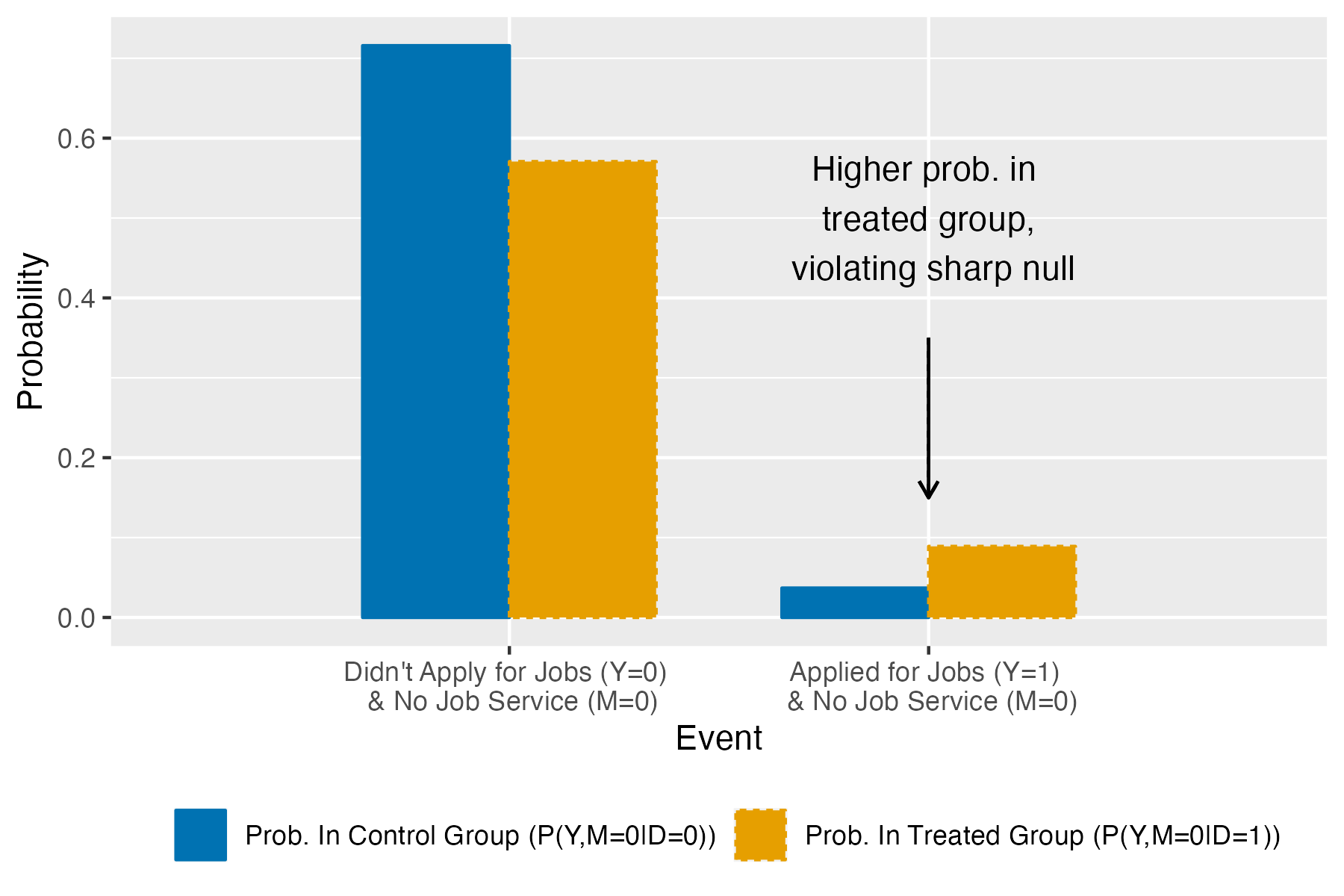}
    \caption{Illustration of Testable Implications in \citet{bursztyn_misperceived_2020} Using Full Sample}\label[appendixfigure]{fig:bursztyn_bin_appendix}
    \floatfoot{Note: This figure is analogous to \Cref{fig:bursztyn_bin} except it uses the full sample rather than restricting to men who initially underestimate others' beliefs.}
\end{figure}

\paragraph{Alternative tests.} In the main text, we report statistical tests of the sharp null using CS, using a discretization with 5 bins for the \citet{baranov_maternal_2020} application. \Cref{tbl:pvals-bins} presents results for the \citet{baranov_maternal_2020} application alternatively using either 2 bins or 10 bins, with qualitatively similar conclusions. \Cref{tbl:pvals} reports test results using the tests of ARP and FSST instead of CS (using 5 bins for the \citet{baranov_maternal_2020} application).\footnote{Recall that the reported $p$-value is the smallest value of $\alpha$ for which the test rejects. Since ARP uses a two-stage procedure, it is difficult to analytically compute the $p$-value. We therefore compute the test for $\alpha$ values on a grid with interval-length $0.01$ between $0.01$ and $0.1$ and interval-length $0.1$ between $0.15$ and $0.95$, and report the smallest grid point at which the test rejects.} The qualitative pattern across the tests is similar. One notable difference is that we do not reject the null for the relationship-quality mechanism in \citet{baranov_maternal_2020} using ARP, although this is perhaps unsurprising given the low power of ARP in simulations calibrated to this mechanism. 

\begin{table}[!ht]
    \centering
    
\begin{tabular}{lccc}
\toprule
\multicolumn{1}{c}{ } & \multicolumn{3}{c}{Number of bins} \\
\cmidrule(l{3pt}r{3pt}){2-4}
Mediator & 2 & 5 & 10\\
\midrule
Grandmother & 0.003 & 0.023 & 0.065\\
Relationship & 0.005 & 0.028 & 0.001\\
Grandmother + Relationship & 0.198 & 0.654 & 0.999\\
\bottomrule
\end{tabular}

    \caption{$p$-values for tests for the sharp null in \citet{baranov_maternal_2020} using alternative bin choices}
    \label[appendixtable]{tbl:pvals-bins}
\end{table}

\begin{table}[!htb]
    
\begin{tabular}{@{\extracolsep{5pt}} cccccc} 
\\[-1.8ex]\hline 
\hline \\[-1.8ex] 
Application & M & CS & ARP & FSSTdd & FSSTndd \\ 
\hline \\[-1.8ex] 
Bursztyn et al (main sample) & Job-search Sign-up & $0.020$ & $0.030$ & $0.018$ & $0.020$ \\ 
Bursztyn et al (full sample) & Job-search Sign-up & $0.019$ & $0.020$ & $0.021$ & $0.022$ \\ 
Baranov et al & Grandmother & $0.023$ & $0.030$ & $0.026$ & $0.047$ \\ 
Baranov et al & Relationship & $0.028$ & $0.650$ & $0.037$ & $0.049$ \\ 
Baranov et al & Grandmother + Relationship & $0.654$ & $0.550$ & $0.115$ & $0.256$ \\ 
\hline \\[-1.8ex] 
\end{tabular} 

    \caption{$p$-values for tests for the sharp null using alternative procedures}
    \label[appendixtable]{tbl:pvals}
\end{table}

\end{document}